\RequirePackage{tikz}
\documentclass[sn-mathphys]{sn-jnl}
\pdfoutput=1

\usepackage{hyperref}
\usepackage{amsmath,amssymb,amsthm,amscd}
\usepackage{mathtools}
\usepackage[all,cmtip]{xy}
\usepackage{xcolor}
\usepackage{blkarray}
\usepackage{float}

\newtheorem{theorem}{Theorem}[section]
\newtheorem{lemma}[theorem]{Lemma}
\newtheorem{proposition}[theorem]{Proposition}
\newtheorem{definition}[theorem]{Definition}
\newtheorem{conjecture}[theorem]{Conjecture}

\usepackage{nicematrix,tikz}
\usetikzlibrary{decorations.pathreplacing}

\providecommand{\abs}[1]{\lvert#1\rvert}

\def\rank{\text{rank}}
\def\corank{\text{corank}}

\def\Ad{A_{\vec{d}}}
\def\Adr{A_{\vec{d}}^{\text{r}}}
\def\Sd{S_{\vec{d}}}
\def\SdOne{S_{\vec{d},1}}
\def\Sdr{S_{\vec{d}}^{\text{r}}}
\def\SdrOne{S_{\vec{d},1}^{\text{r}}}
\def\SdrTwo{S_{\vec{d},2}^{\text{r}}}
\def\SdrOneTwo{S_{\vec{d},i}^{\text{r}}}

\def\Xd{X_{\vec{d}}}
\def\Xdp{X_{\vec{d}'}}
\def\Xdr{X^{\text{r}}_{\vec{d}}}
\def\Xab{X[\vec{a}\vert\vec{b}]}
\newcommand{\Xnc}[1]{X^{\text{n.c.}}_{{#1}}}
\def\Xdnc{\Xnc{\vec{d}}}

\newcommand{\gvHLc}[1]{\color{blue}\textbf{#1}}
\newcommand{\gvHLi}[1]{\color{red}\textbf{#1}}

\title{New non-commutative resolutions of determinantal Calabi-Yau threefolds from hybrid GLSM}
\abstract{
We study topological strings on non-commutative resolutions of singular Calabi-Yau threefolds that are double covers of $\mathbb{P}^3$, ramified over determinantal octic surfaces.
 Using conifold transitions to complete intersections in toric ambient spaces, we prove that any small resolution has 2-torsional exceptional curves and is necessarily non-K\"ahler.
        The same transitions imply that M-theory develops a $\mathbb{Z}_2$ gauge symmetry on the singular space.
	We then construct gauged linear sigma models with hybrid phases that flow to the worldsheet theories of strings propagating on the determinantal double solids in the presence of a flat but topologically non-trivial B-field.
	Localizing the sphere partition function allows us to calculate the fundamental periods of the mirror Calabi-Yau manifolds, then we check agreement with the periods of the Borisov-Li mirrors.
	We find that the corresponding variations of Hodge structure either correspond to one of the 14 hypergeometric cases or to a double cover thereof.
	We then use mirror symmetry and integrate the holomorphic anomaly equations to calculate $\mathbb{Z}_2$-refined Gopakumar-Vafa invariants for several examples.
}

\author[1]{\fnm{Sheldon} \sur{Katz}}\email{katzs@illinois.edu}
\affil[1]{\orgdiv{Department of Mathematics}, \orgname{University of Illinois Urbana-Champaign}, \orgaddress{\city{Urbana}, \state{IL}, \postcode{61801}, \country{USA}}}

\author[2]{\fnm{Thorsten} \sur{Schimannek}}\email{schimannek@lpthe.jussieu.fr}
\affil[2]{\orgdiv{Laboratoire de Physique Th\'eorique et Hautes Energies (LPTHE), UMR 7589}, \orgname{CNRS-Sorbonne Universit\'e}, \orgaddress{\street{Campus Pierre et Marie Curie, 4 place Jussieu}, \postcode{F-75005}, \city{Paris}, \country{France}}}

\jyear{2023}

\numberwithin{equation}{section}
\renewcommand{\theequation}{\thesection.\arabic{equation}}

\begin{document}

\maketitle
\flushbottom

\tableofcontents

\section{Introduction}
Recent work~\cite{Schimannek:2021pau,Katz:2022lyl} has initiated the study of topological strings on non-commutative resolutions of singular compact Calabi-Yau threefolds, building on the pioneering work in~\cite{Aspinwall:1995rb,Caldararu:2010ljp,Addington:2012zv}.
Based on physical arguments from F- and M-theory and explicit calculations using mirror symmetry, it was proposed that the corresponding topological string free energies encode a torsion refinement of the usual Gopakumar-Vafa invariants~\cite{Schimannek:2021pau}.
Physically, the refined invariants resolve the charge under a discrete gauge symmetry that M-theory develops on the singular Calabi-Yau.
A geometric interpretation of these invariants was conjectured in~\cite{Katz:2022lyl} and implies that the refinement accounts for the presence of torsion in the homology of curves on non-K\"ahler small resolutions.

In this paper we will construct a rich class of new examples.
Studying their properties will then allows us to build on the previous results in multiple ways, both physically and mathematically.

\paragraph{Singular Calabi-Yau without K\"ahler small resolution}
The prototypical example for the non-commutative resolutions that we are interested in arises -- at least conjecturally -- whenever a projective Calabi-Yau threefold $X$ has isolated nodal singularities such that the exceptional curves in any small resolution are torsion in homology~\cite{Katz:2022lyl}.
As a consequence, the Calabi-Yau does not admit a small resolution that is K\"ahler.
However, any non-K\"ahler small resolution $\widehat{X}$ of $X$ then has a non-trivial Brauer group\footnote{We assume throughout this paper the often-used conjecture that the Brauer group of 
a compact complex manifold coincides with its cohomological Brauer group.}
\begin{align}
    \text{Br}(\widehat{X})\simeq \text{Tors}\,H_2(\widehat{X},\mathbb{Z})\simeq H^2(\widehat{X},U(1))\,,
\end{align}
and therefore admits topologically non-trivial flat B-fields.
After shrinking the exceptional curves back to zero size, the B-field ``stabilizes'' the singularities in the sense that complex structure deformations which remove the nodes are obstructed and the string worldsheet theory becomes regular, a phenomenon that has been originally described in~\cite{Vafa:1994rv,Aspinwall:1995rb}.

\paragraph{Flat B-fields and non-commutative resolutions}
From the open string perspective, the B-field leads to a non-commutative deformation of the coordinate ring, or, more generally, the structure sheaf of the space, see e.g.~\cite{Schomerus:1999ug,Seiberg:1999vs}.
Physics then suggests the existence of a sheaf of non-commutative algebras ${\mathcal{B}}$ on $X$ such that the topological B-branes branes are elements of the derived category $D^b(X,{\mathcal{B}})$ of sheaves of ${\mathcal{B}}$ modules~\cite{Kapustin:1999di,Berenstein:2000ux,Berenstein:2001jr}.

While such a sheaf is in general difficult to construct, the results from~\cite{Schimannek:2021pau,Katz:2022lyl} suggest that there always exists a derived equivalence $D^b(\widehat{X},\alpha)\simeq D^b(X,{\mathcal{B}})$, where $\alpha\in\text{Br}(\widehat{X})$ is the cohomology class of the B-field and $D^b(\widehat{X},\alpha)$ is the derived category of $\alpha$-twisted sheaves on $\widehat{X}$.  Furthermore, this derived equivalence should be induced from an Azumaya algebra $\widehat{\mathcal{B}}$ on $\widehat{X}$ representing the Brauer class $\alpha$, together with an isomorphism $\pi_*(\widehat{\mathcal{B}})\simeq {\mathcal{B}}$, with $\pi:\widehat{X}\rightarrow X$.

Motivated by this, we refer to a choice of small resolution $\widehat{X}$ together with a topologically non-trivial flat B-field as a non-commutative resolution $X^{\text{n.c.}}$ of $X$, even though the existence of a corresponding sheaf ${\mathcal{B}}$ on $X$ is in general conjectural~\footnote{For a discussion of the relationship to the different notions of non-commutative crepant resolution~\cite{vdb02,bergh2022noncommutative} and crepant categorical resolution~\cite{Kuznetsov2008ls} in this context we refer to~\cite{Katz:2022lyl}.}.
Locally, evidence was found in~\cite{Katz:2022lyl} that that the geometry of $X^{\text{n.c.}}$ around the nodes should admit a description in terms of Szendrői's non-commutative conifold~\cite{Szendroi:2007nu}.

\paragraph{Clifford non-commutative resolutions}
In this paper we will focus on so-called Clifford type non-commutative resolutions.
These arise when the singular Calabi-Yau $X$ is a double cover of a Fano threefold $B$ that is ramified over a symmetric determinantal surface.
One can then construct a sheaf of even parts of Clifford algebras $\mathcal{B}_0$ on $B$ and at least in certain cases it has been shown that $D^b(B,\mathcal{B}_0)\simeq D^b(\widehat{X},\alpha)$~\cite{Kuznetsov2008,addington2009derived,Kuznetsov2013}.

In these cases, the sheaf $\mathcal{B}_0$ is automatically the pushforward to $B$ of a sheaf $\mathcal{B}$ of non-commutative algebras on $X$.  In general, $\mathcal{B}$ is only a sheaf of Azumaya algebras on the smooth locus of $X$, and does not extend to an Azumaya algebra on $X$~\footnote{We follow the convention from the mathematical literature where the term ``Azumaya algebra'' is used synonymously with ``sheaf of Azumaya algebras''.}.
However, the smooth locus can instead be compactified to a small resolution $\widehat{X}$ of $X$ and we can try to solve this ``extension problem" on $\widehat{X}$ instead.
Among the examples that we study in this paper we will find some where this extension problem has been solved as a special case of results from~\cite{Kuznetsov2013}.
We will explicitly describe how the geometry of $\widehat{X}$ is reflected in those cases in the representation theory of the sheaves ${\mathcal{B}_0}$ and $\mathcal{B}$, and make connections between torsion refined invariants and moduli spaces of ${\mathcal{B}_0}$-modules.

\paragraph{Determinantal octic double solids}
\begin{table}[t!]
\centering
\begin{align*}
\renewcommand{\arraystretch}{1.2}
\begin{array}{|c|c|c|c|c|}\hline
\Xdnc&n_{\vec{d}}&\text{Mirror Picard-Fuchs operator } \mathcal{D}_{\vec{d}}&\Delta(z)&\text{Smooth cousin } \widetilde{X}_{\vec{d}}\\\hline
\Xnc{(5,1^3)}&64&\theta ^4-2^85 z^2(5 \theta +1) (5 \theta +3) (5 \theta +7) (5 \theta +9)&1-2^85^5z^2&X_{10}\subset\mathbb{P}^4(1^3,2,5)\\
\Xnc{(4,2^2)}&72&\theta ^4-2^4 z (4 \theta +1)(2 \theta +1)^2 (4 \theta +3)&1-2^{10}z&X_{4,2}\subset\mathbb{P}^5\\
\Xnc{(3^2,1^2)}&76&\theta ^4-2^83^2z^2 (3 \theta +1)^2 (3 \theta +5)^2&1-2^83^6z^2&X_{6,6}\subset\mathbb{P}^5(1^2,2^2,3^2)\\
\Xnc{(3,1^5)}&80&\theta ^4-2^83 z^2 (\theta +1)^2 (3 \theta +1) (3 \theta +5)&1-2^83^3z^2&X_{6,2}\subset\mathbb{P}^5(1^5,3)\\
\Xnc{(2^4)}&80&\theta ^4-2^4z (2 \theta +1)^4&1-2^8z&X_{2,2,2,2}\subset\mathbb{P}^7\\
\Xnc{(1^8)}&84&\theta ^4-2^8z^2 (\theta +1)^4&1-2^8z^2&X_{2,2,2,2}\subset\mathbb{P}^7\\
\hline
\end{array}
\end{align*}
\caption{There are six inequivalent decompositions $\vec{d}\in \mathbb{N}^k$ of degree $\vert\vec{d}\vert=8$ and length $l\ge 3$.
We list the number of nodes $n_{\vec{d}}$ of the determinantal octic double solid $\Xd$, the Picard-Fuchs operator that annihilates the periods of the mirror of the non-commutative resolution $\Xdnc$ and the corresponding discriminant polynomial $\Delta(z)$. We also indicate the smooth complete intersection Calabi-Yau 3-folds that share the same mirror variations of Hodge structure over $\mathbb{Q}$ or one that is related by a 2:1 covering.}
\label{tab:summary}
\end{table}
Generalizing earlier work in~\cite{Katz:2022lyl}, we will show that a large class of examples can be obtained from Calabi-Yau double covers of Fano threefolds $B$ that are ramified over symmetric determinantal surfaces $S$.

Many examples already arise from the case $B=\mathbb{P}^3$, with $S$ a determinantal octic surface.
Such a Calabi-Yau can be referred to as a \textit{determinantal octic double solid}.
Concretely, given a vector $\vec{d}=\mathbb{N}^k$ with $\vert \vec{d}\vert =d_1+\ldots +d_k=8$ we study double solids $\Xd$ that are ramified over the surface
\begin{align}
	\Sd=\{\,\text{det}\,A_{k\times k}=0\,\}\subset\mathbb{P}^3\,,
\end{align}
where $A_{k\times k}$ is a symmetric $k\times k$ matrix with entries $A_{i,j}$ that are generic homogeneous polynomials of degree $(d_i+d_j)/2$ in the homogeneous coordinates on $\mathbb{P}^3$.
We refer to $\vec{d}$ as a \textit{decomposition of degree $\vert \vec{d}\vert$} and to the number of non-zero entries in $\vec{d}$ as the \textit{length} of the decomposition.
As will be explained further in Section~\ref{sec:sods}, we will always assume that $d_i=d_j\text{ mod }2$ to ensure that $\Xd$ is irreducible.

The length one decomposition $\vec{d}=(8)$ corresponds to the smooth octic Calabi-Yau $X_{(8)}\subset\mathbb{P}^4(1^4,4)$.
However, for decompositions $\vec{d}$ of length $l\ge 2$ we obtain special sub-loci of the complex structure moduli space of $X_{(8)}$ where the corresponding Calabi-Yau $\Xd$ exhibits $n_{\vec{d}}$ isolated nodal singularities.

For any choice of $\vec{d}$ we construct conifold transitions between the determinantal octic double solid $\Xd$ and a smooth K\"ahler Calabi-Yau threefold $\widehat{\Xdr}$ that is a complete intersection in a projective bundle over $\mathbb{P}^3$. 
This will allow us to prove that $\Xd$ admits a K\"ahler small resolution if and only if the length $l$ of $\vec{d}$ is $l\le 2$.
On the other hand, if $l\ge 3$ we will prove that any small resolution $\widehat{\Xd}$ of $\Xd$ has $H_2(\widehat{\Xd},\mathbb{Z})\simeq \mathbb{Z}\oplus\mathbb{Z}_2$ and is non-K\"ahler because the exceptional curves represent the torsion class.
These results are contained in Proposition~\ref{prop:11classes} which will be proven in Section~\ref{sec:geometry}.

In the rest of the paper we then study the Clifford non-commutative resolutions $\Xdnc$ of the determinantal octic double solids $\Xd$ associated to the six inequivalent decompositions of length $l\ge 3$ that are listed in Table~\ref{tab:summary}.

Although we focus for concreteness on the case $B=\mathbb{P}^3$, let us stress again that our techniques can be equally applied to other Fano bases and to double covers that are not Calabi-Yau.

\paragraph{Smooth duals and their absence}
In the examples that have been studied in~\cite{Schimannek:2021pau,Katz:2022lyl}, the stringy K\"ahler moduli space of the non-commutative resolution contained at least one other large volume limit that is associated to an ordinary smooth Calabi-Yau threefold $Y$.
This can be seen as a consequence of a derived equivalence
\begin{align}
    D^b(X^{\text{n.c.}})\simeq D^b(Y)\,,
\end{align}
which either appeared as an instance of homological projective duality~\cite{Kuznetsov2008,Caldararu:2010ljp,Katz:2022lyl} or of the twisted derived equivalences that relate different torus fibrations that share the same Jacobian~\cite{Caldararu2002,Donagi2008,Schimannek:2021pau}.

If one can construct a mirror to $Y$, this is also mirror to $X^{\text{n.c.}}$, and so $Y$ and $X^{\text{n.c.}}$ are so-called double mirrors~\cite{Aspinwall:1993yb}.
The usual B-model techniques can then be applied in order to calculate topological string free energies associated to $X^{\text{n.c.}}$~\cite{Schimannek:2021pau,Katz:2022lyl}.
In particular, the mirror of $X^{\text{n.c.}}$ is itself an ordinary smooth Calabi-Yau threefold.
A combinatorial construction of Clifford non-commutative resolutions with a smooth double mirror $Y$, using reflexive Gorenstein cones, was developed in~\cite{borisov2016clifford,Li2020,Borisov2019}.

However, the existence of such a smooth dual $Y$ is not an intrinsic property of non-commutative resolutions.
As we will find in this paper, it actually appears to be the exception and from the examples in Table~\ref{tab:summary} only $X^{\text{n.c.}}_{(1^8)}$ exhibits a smooth dual given by the complete intersection of four quadrics $X_{2,2,2,2}\subset\mathbb{P}^7$.
In the absence of a smooth dual, we will employ two different methods that can be used to study the topological string A-model on $\Xdnc$.

Our first method will be to construct gauged linear sigma models (GLSM) that realize the Clifford non-commutative resolution in a hybrid phase, generalizing the results from~\cite{Caldararu:2010ljp,Sharpe:2013bwa}.
We calculate the sphere partition function via localization and use the relationship with the K\"ahler potential  on the moduli space~\cite{Jockers:2012dk} in order to obtain the periods of the mirror Calabi-Yau.

As our second method we will use the combinatorial construction of $\Xdnc$ from~\cite{borisov2016clifford} in order to obtain the mirror and to calculate the fundamental period directly.
We will find complete agreement with the result from localization.

\paragraph{Smooth cousins}
Somewhat surprisingly, we find that in each of our examples the mirror Picard-Fuchs operator corresponds to one of the 14 well-known hypergeometric cases~\footnote{The fact that different geometries can share the same period geometry is of course well known, with a prominent example being the mirrors of the quintic and of its $\mathbb{Z}_5$ quotient~\cite{Aspinwall:1994uj,Doran:2005gu}. Note that the so-called ``14-th case'', that corresponds to a singular Calabi-Yau threefold~\cite{Clingher_2016}, does not appear.}.
The latter are usually associated with A-model geometries that are hypersurfaces or complete intersections in weighted projective spaces~\cite{Hosono:1993qy,Doran:2005gu}.

Let us stress, that -- apart from the one exception $\vec{d}=(1^8)$ -- these are not the previously discussed smooth duals and they are in particular not derived equivalent with $\Xdnc$.
We will instead refer to these geometries as ``smooth cousins''.
While the $\mathbb{Q}$-variations of Hodge structures of the corresponding A-model Calabi-Yau manifolds are the same as that of the corresponding $\Xd^{\text{n.c.}}$, we will use the results from~\cite{Katz:2022lyl} to show that the integral structure always differs.

\paragraph{Higgs transitions in M-theory}
\begin{figure}[ht!]
\centering
	\includegraphics[width=\linewidth]{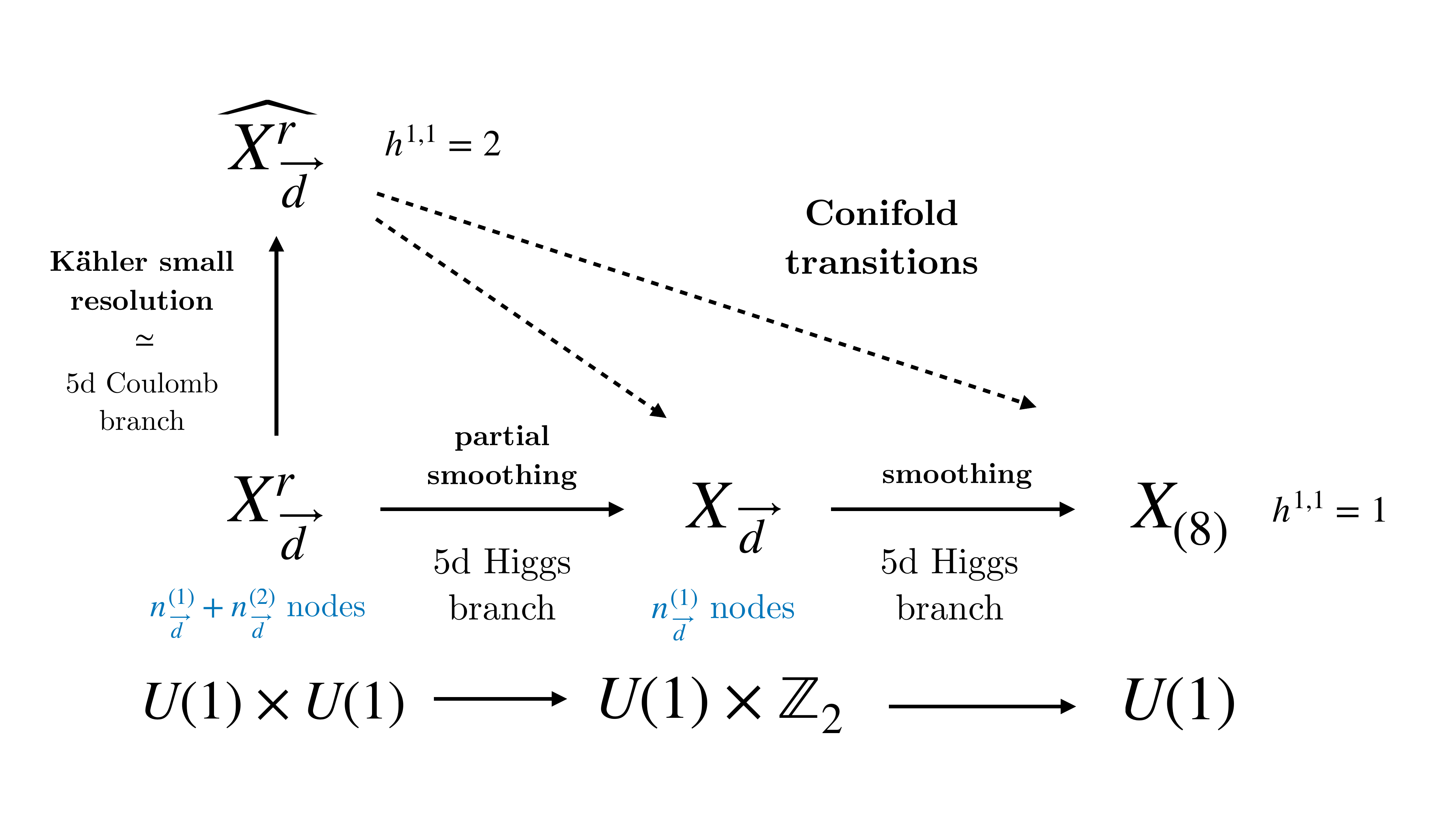}
\caption{The relationship between the different geometries for a given decomposition $\vec{d}\in\mathbb{N}^k$ of degree $8$ and length $l\ge 3$, as well as the physical interpretation of the transitions in the five-dimensional M-theory compactification.}
\label{fig:intro}
\end{figure}
From a physical perspective, the conifold transitions from $\widehat{\Xdr}$ to $\Xd$ can be interpreted as Higgs transitions, see Figure~\ref{fig:intro}.

The five-dimensional M-theory compactification on $\widehat{\Xdr}$ exhibits a $G=U(1)\times U(1)$ gauge symmetry and the spectrum contains -- among other fields -- hypermultiplets $\Phi^{(i)}$ with charges $(i,0)$ for $i=1,2$.
Crucially, the hypermultiplets $\Phi^{(2)}$ will turn out to only exist when the length of the decomposition is $l\ge 3$.

While the theory on $\widehat{\Xdr}$ is on a generic point of the Coulomb branch, a singular Calabi-Yau $\Xdr$ is obtained by setting the Coulomb branch parameter associated to $U(1)_1$ to zero.
The fields in $\Phi^{(i=1,2)}$ then become massless.
For decompositions of length $\vec{d}\ge 3$ the transition to $\Xd$ amounts to giving a non-zero vacuum expectation value to the scalar fields in the hypermultiplets $\Phi^{(2)}$.
As a result, we find that the gauge symmetry of M-theory on the corresponding $\Xd$ is $G'=U(1)\times \mathbb{Z}_2$.

Note that giving a vacuum expectation also to the scalar fields in $\Phi^{(1)}$ further breaks the gauge group to $G''=U(1)$ and one obtains the M-theory compactification for a generic choice of complex structure on $X_{(8)}$.

Analogous transitions between torus fibered Calabi-Yau that lead to discrete gauge symmetries have been studied in the F-theory context for example in~\cite{Braun:2014oya,Mayrhofer:2014laa,Klevers:2014bqa,Cvetic:2015moa,Oehlmann:2019ohh,Knapp:2021vkm,Schimannek:2021pau}, while the relationship between discrete symmetries and torsional homology has also been discussed e.g. in~\cite{Berasaluce-Gonzalez:2012awn,Grimm:2015ona,Braun:2017oak,Casas:2023wlo}.
We provide to the best of our knowledge the first examples where the existence of torsion curves in a small resolution after the transition has been rigorously proven.

\paragraph{Conifold transitions with B-fields}
We will also use the GLSM and mirror symmetry to study the difference between the transitions from $\widehat{\Xdr}$ to $\Xd$ and $\Xd^{\text{n.c.}}$.
In this way we can explicitly show that the hybrid GLSM phases can be interpreted as corresponding to $\Xd$ with a flat topologically non-trivial B-field, and therefore to $\Xd^{\text{n.c.}}$.

In particular, we confirm that $\Xd$ and $\Xd^{\text{n.c.}}$ have disconnected moduli spaces, away from the locus where the underlying space degenerates to $\Xdr$, and, as a result, different mirror Calabi-Yau manifolds.

Further evidence for this is provided by the results from localization and from the Borisov-Li mirror construction.
Recall that the topological string A-model is invariant under complex structure deformations.
Since in each case $\Xd$ can be deformed to $X_{(8)}$, one can use the A-model on $X_{(8)}$ to define the A-model on $\Xd$.
On the other hand, our localization calculation will lead to different variations of Hodge structures for the mirror Calabi-Yau of $\Xd^{\text{n.c.}}$ for each decomposition.

\paragraph{Torsion refined Gopakumar-Vafa invariants} 
We calculate torsion refined Gopakumar-Vafa invariants using B-model techniques, following~\cite{Schimannek:2021pau,Katz:2022lyl}.
This includes the direct integration of the holomorphic anomaly equations~\cite{Bershadsky:1993ta,Bershadsky:1993cx}, where the holomorphic ambiguity can be fixed by the conifold gap condition, Castelnuovo type vanishing conditions, further enumerative predictions and certain regularity constraints, as was first done for smooth compact Calabi-Yau in~\cite{Huang:2006hq} and recently carried out to higher genera in~\cite{Katz:2022lyl,Alexandrov:2023zjb,Gu:2023mgf}.
 
 In order to do geometric computations, it was conjectured in~\cite{Katz:2022lyl} that $D^b(\widehat{X},\alpha)$ supports Bridgeland stability conditions whose heart contains all coherent sheaves of dimension at most 1 on any small resolution $\widehat{X}'$ of $X$. While no attempt was made to construct such stability conditions, nor do we do so here, we were able to perform many calculations which matched B-model calculations by making the following reasonable assumption about the limiting stability condition at large radius: if $C$ is any curve in any $\widehat{X}'$, then the sheaf $\mathcal{O}_C$ (identified with an object of $D^b(\widehat{X},\alpha)$ by appropriate Fourier-Mukai equivalences and a trivialization of $\alpha_C$) is stable.
 We can then apply standard techniques to compute some of the invariants and to compare with the results from mirror symmetry.

\paragraph{Outline}
In Section~\ref{sec:prelims} we provide some background on singular Calabi-Yau threefolds, torsion refined Gopakumar-Vafa invariants and topologically non-trivial B-fields.
We also summarize the current state of the mathematical definition of the invariants.
In Section~\ref{sec:geometry} we study the geometry of determinantal double solids $\Xd$ and use conifold transitions to prove Proposition~\ref{prop:11classes} that describes the homology of small resolutions. We also interpret the conifold transitions as Higgs transitions in M-theory.
In Section~\ref{sec:cliffordresolution} this is followed by a detailed study of the corresponding Clifford non-commutative resolutions.
Section~\ref{sec:glsm} contains the discussion of the corresponding GLSM and the calculation of the sphere partition functions.
In Section~\ref{sec:gorenstein} we then calculate the fundamental periods of the mirrors directly, using the combinatorial construction from~\cite{borisov2016clifford}.
The associated variations of Hodge structure and the integral structure of the periods will be discussed in Section~\ref{sec:branes}.
In Section~\ref{sec:direct} we use mirror symmetry and integrate the holomorphic anomaly equations to calculate torsion refined Gopakumar-Vafa invariants for some of the examples.
In Section~\ref{sec:enumerativeGeometry}, we discuss some consequences from the conifold transitions on the enumerative predictions for the invariants, confirming earlier observations from~\cite{Katz:2022lyl}.
Finally, in Section~\ref{sec:outlook}, we provide an outlook on some of the many interesting directions that are left for future work.

\section{Preliminaries}
\label{sec:prelims}
In this section we will first give a physically oriented introduction to the definition from M-theory of torsion refined Gopakumar-Vafa invariants of Calabi-Yau threefolds $X$ without a small K\"ahler resolution.  After that, we will explain the role of flat but topologically non-trivial B-fields in the A-model topological string free energies.
We will then summarize the current state of the mathematical definition of the invariants in terms of the enumerative geometry of all small resolutions.  If we have a non-commutative resolution $(X,\mathcal{B})$, we sketch a proposal for an alternative definition of the torsion refined invariants in terms of an enumerative theory of $\mathcal{B}$-modules. 

\subsection{Singular Calabi-Yau without K\"ahler small resolutions}
\label{sec:singCYnonKahler}
A K\"ahler Calabi-Yau threefold with $n$ isolated nodal singularities always admits $2^n$ small resolutions but in general none of those is K\"ahler.
A theorem by Werner states that given any small resolution $\widehat{X}$ of a projective nodal threefold $X$, the nodal threefold does admit a K\"ahler small resolution iff none of the exceptional curves on $\widehat{X}$ is either torsion or trivial in $H_2(\widehat{X},\mathbb{Z})$~\cite{Werner,WernerTranslate}.

In fact, it is easy to see that a manifold $M$ can not be K\"ahler if it contains a holomorphic curve $C$ with homology class $[C]\in H_2(M,\mathbb{Z})$ such that $n\cdot [C]\sim 0$ for some $n\in \mathbb{N}_{>0}$.
Any 2-form $\omega$ would have to satisfy
\begin{align}
    \int_C\omega =\frac{1}{n}\int_{n\cdot C}\omega=0\,.
\end{align}

Note that if the exceptional curve that resolves a given node is homologically trivial, it bounds a 3-chain that after shrinking the curve turns into a 3-cycle.
The latter leads to a complex structure modulus in the singular Calabi-Yau that can be used to deform away the singularity.

From a physical perspective, such nodes lead to uncharged hypermultiplets in M-theory and the deformation gives vacuum expectation values to the corresponding scalar fields~\cite{Arras:2016evy,Grassi:2018rva}.

On the other hand, it has been argued in~\cite{Schimannek:2021pau,Katz:2022lyl}, based on earlier observations in the context of torus fibrations in F-theory~\cite{Braun:2014oya,Morrison:2014era,Mayrhofer:2014laa,Anderson:2014yva,Klevers:2014bqa,Cvetic:2015moa,Oehlmann:2019ohh,Knapp:2021vkm}, that nodes which are resolved by torsional exceptional curves lead to massless hypermultiplets that are charged under discrete gauge symmetries in M-theory.

More precisely, the claim is that given any small resolution $\widehat{X}$ the five-dimensional gauge group takes the form~\cite{Katz:2022lyl}
\begin{align}
	G_{5d}=\text{Hom}\left(H_2(\widehat{X},\mathbb{Z}),U(1)\right)\,,
\end{align}
while the charge lattice can be identified with $H_2(\widehat{X},\mathbb{Z})$.

The nodes can still be deformed away, corresponding again to turning on vacuum expectation values for the charged scalar fields.
However, we will decide to keep them and thus to preserve the discrete gauge symmetry.

As was discussed in~\cite{Schimannek:2021pau,Katz:2022lyl}, the discrete gauge symmetry in M-theory can then be used to define a torsion refinement of the Gopakumar-Vafa invariants~\cite{Gopakumar:1998ii,Gopakumar:1998jq}~\footnote{The possibility of such a torsion refinement for K\"ahler Calabi-Yau threefolds with torsion homology was also discussed in~\cite{Dedushenko:2014nya} while some torsion refined instanton numbers for a smooth Calabi-Yau threefold have been calculated in~\cite{Braun:2007tp,Braun:2007vy}.}.
Moreover, the torsion in the homology of $\widehat{X}$ leads to the possibility of turning on a flat but topologically non-trivial B-field that stabilizes the singularities and leads to string compactifications on non-commutative resolutions of the singular Calabi-Yau~\cite{Schimannek:2021pau,Katz:2022lyl}.

\subsection{Torsion refined Gopakumar-Vafa invariants from physics}
The original work~\cite{Gopakumar:1998ii,Gopakumar:1998jq}, see also~\cite{Dedushenko:2014nya}, defines Gopakumar-Vafa invariants physically in terms of the BPS spectrum of the five-dimensional effective theory that arises from M-theory on a Calabi-Yau threefold.

The little group of massive particles in the effective theory is
\begin{align}
    Spin(4)=SU(2)_L\times SU(2)_R\,.
\end{align}
On a smooth Calabi-Yau $Y$ the $U(1)$ charge lattice can be identified with $H_2(Y,\mathbb{Z})$ and one can define the multiplicity $N_{j_L,j_R}^\beta$ of BPS states with charge $\beta\in H_2(Y,\mathbb{Z})$ that transform in the $SO(4)$ representation
\begin{align}
    \left[\left(\frac12,0\right)\oplus 2(0,0)\right]\otimes (j_L,j_R)\,.
\end{align}
Using $I_g=([\frac12]+2[0])^g$ in terms of $SU(2)$ representations $[j]$, the Gopakumar-Vafa invariants $n_g^\beta$ are then defined via the relation
\begin{align}
    \sum\limits_{g=0}^\infty n_g^\beta I_g=\sum\limits_{j_L,j_R}(-1)^{2j_R}(2j_R+1)N^\beta_{j_L,j_R}\cdot[j_L]\,,
\end{align}
taking $\mathrm{Tr}(-1)^{F_R}$ over the $SU(2)_R$ representations.

It was argued in~\cite{Schimannek:2021pau,Katz:2022lyl} that the same definition can be used for a Calabi-Yau threefold $X$ with isolated nodal singularities, if the charge lattice is replaced with the homology $H_2(\widehat{X},\mathbb{Z})$ of any small resolution $\widehat{X}$ of $X$.

In this paper we will focus on nodal Calabi-Yau threefolds $X$ with
\begin{align}
    H_2(\widehat{X},\mathbb{Z})\simeq \mathbb{Z}\oplus\mathbb{Z}_2\,.
    \label{eqn:h2isom}
\end{align}
The corresponding M-theory gauge group is $G\simeq U(1)\times\mathbb{Z}_2$ and we refer to the $U(1)$ charge as the \textit{degree} $d$ and to the $\mathbb{Z}_2$ charge as the \textit{parity} $p\in\{0,1\}$.
The corresponding torsion refined Gopakumar-Vafa invariants at genus $g$ will be denoted as $n_g^{d,p}$.

It is important to point out that the isomorphism~\eqref{eqn:h2isom} is not canonical and a change of basis can lead to a shift $p\mapsto p+d\text{ mod }2$.
Physically, this ambiguity is reflected in the different choices of isomorphisms $G\simeq U(1)\times\mathbb{Z}_2$.

\subsection{Fractional B-fields and non-commutative resolutions}
Compactifying the five-dimensional M-theory vacuum on a circle, the theory becomes dual to the Type IIA compactification on the same Calabi-Yau.

It was argued in~\cite{Aspinwall:1995rb}, based on observations in the context of discrete torsion~\cite{Vafa:1994rv}, that flat but topologically non-trivial B-fields can then stabilize singularities, in the sense that the string worldsheet theory becomes regular and the deformations that remove the singularities are then obstructed.

From the perspective of the five-dimensional M-theory compactification, the additional degree of freedom from the Type IIA B-field can be understood as a choice of discrete bundle along the circle that is used to compactify to four dimensions~\cite{Dierigl:2022zll}.
The choice of bundle can be thought of as a discrete Coulomb branch and a non-trivial twist generates a mass for the particles that transform non-trivially under the corresponding gauge transformation.
The Calabi-Yau is expected to be fully (non-commutatively) resolved, if all of the charged hypermultiplets have obtained a non-zero mass.

While in~\cite{Aspinwall:1995rb} the large blowup of a nodal Calabi-Yau $X$ was used as a model for the allowed B-field backgrounds on a nodal Calabi-Yau, one can equivalently use a small non-K\"ahler resolution $\widehat{X}$~\cite{Schimannek:2021pau,Katz:2022lyl}.
The topological types of flat B-field backgrounds are then classified by the Brauer group $\text{Br}(\widehat{X})$, with isomorphisms
\begin{align}
    \text{Br}(\widehat{X})\simeq \text{Tors}\,H^3(\widehat{X},\mathbb{Z})\simeq \text{Hom}\left(H_2(\widehat{X},U(1)\right)\,,
    \label{eqn:brauer}
\end{align}
while also $\text{Tors}\,H^3(\widehat{X},\mathbb{Z})\simeq \text{Tors}\,H_2(\widehat{X},\mathbb{Z})$~\cite{Batyrev:2005jc}.

Already in~\cite{Aspinwall:1995rb} it was pointed out that via~\eqref{eqn:brauer} a non-trivial choice $\alpha\in \text{Br}(\widehat{X})$ leads to an additional phase in the string instanton action associated to a curve $\beta\in H_2(\widehat{X},\mathbb{Z})$
\begin{align}
    e^{-S(\beta)}=\alpha(\beta) \exp\left(2\pi i\int_\beta\pi^* \omega_{\mathbb{C}}\right)\,,
    \label{eqn:instaction}
\end{align}
where $\pi$ is the projection from $\widehat{X}$ to $X$ and $\omega_{\mathbb{C}}$ is the  usual complexified K\"ahler class on $X$.

For ease of exposition, let us assume that
\begin{align}
    H_2(\widehat{X},\mathbb{Z})\simeq H_2(X,\mathbb{Z})\oplus\mathbb{Z}_n\,,
\label{eqn:splith2}
\end{align}
and use a particular choice of isomorphism to represent the class $\beta$ of a curve in $\widehat{X}$ as $(\tilde{\beta},q)$.
This induces a corresponding isomorphism $\text{Br}(\widehat{X})\simeq \mathbb{Z}_n$ and the phase associated to the curve for a particular choice $[k]\in\text{Br}(\widehat{X})$ takes the form $\alpha(\beta)=e^{2\pi i \frac{k}{n}q}$.

As was discussed in~\cite{Katz:2022lyl}, one can interpret $k/n$ as the value of a ``quantized complexified K\"ahler parameter'' that is associated to the torsional exceptional curves and can only take values $k/n,\,k=0,\ldots,n-1$.
We therefore sometimes refer to a non-trivial choice $[k]\in \text{Br}(\widehat{X})$ as a \textit{fractional B-field}~\footnote{It would be more precise to refer to it as the fractional part of the B-field but we choose to commit this abuse of language in the interest of readability.}.
Different choices of fractional B-fields correspond to different Type IIA string backgrounds and, via~\eqref{eqn:instaction},  to different A-model topological string partition functions.

Explicit examples of topological strings on singular Calabi-Yau with fractional B-fields have first been studied in~\cite{Schimannek:2021pau}, using extremal transitions and modular properties of the topological string partition function to identify corresponding large volume limits with fractional B-fields in the stringy K\"ahler moduli spaces of smooth genus one fibered Calabi-Yau threefolds.

Based on~\cite{Aspinwall:1995rb}, the resulting A-model topological string free energies $Z_{\text{top.}}$ were then pointed out in~\cite{Schimannek:2021pau} to encode the torsion refined Gopakumar-Vafa invariants as
\begin{align}
\begin{split}
    &Z_{\text{top.}}\left(\omega_{\mathbb{C}},k\right)\\
    =&\sum\limits_{\tilde{\beta}\in H_2(X,\mathbb{Z})}\sum\limits_{q=0}^{n-1}\sum\limits_{g=0}^\infty\sum\limits_{m=1}^\infty \frac{n^{\tilde{\beta},q}_g}{m}\left(2\sin\,\frac{\lambda m}{2}\right)^{2g-2} e^{2\pi i \frac{k}{n}q}e^{2\pi i \int_{\tilde{\beta}}\omega_{\mathbb{C}}}\,.
    \end{split}
\end{align}
Another crucial observation is that the partition function for $k=0$, corresponding to the singular Calabi-Yau without a fractional B-field, can be identified with that of the smooth deformation.
Combining the information from the partition functions from all inequivalent choices of fractional B-field $k$ then allowed to calculate the invariants.

For examples that are not torus fibered, torsion refined GV-invariants have first been calculated in~\cite{Katz:2022lyl}, where a conjectural mathematical definition of the invariants was formulated that we will now discuss.

\subsection{Summary for Mathematicians}

We continue to consider projective Calabi-Yau threefolds $X$ with only nodal singularities such that every small resolution $\widehat{X}$ is not K\"ahler, merely a complex manifold of dimension 3.  As discussed earlier in this section, the exceptional curves $C_p\simeq\mathbb{P}^1$ over the nodes $p\in X$ are all trivial or torsion in $H_2(\widehat{X},\mathbb{Z})$.  The topological string free energy can be determined by physical methods.  The analogous situation for smooth Calabi-Yau threefolds suggests that this free energy can be interpreted as a generating function for enumerative invariants.  One of our goals in this paper is to build on our previous work \cite{Katz:2022lyl} and proceed towards a rigorous mathematical definition of these torsion refined invariants.  Before reviewing the viewpoint taken in our previous work \cite{Katz:2022lyl}, we make some comments about the difficulties involved.

Since $\widehat{X}$ is not K\"ahler, there is no natural symplectic structure which can be used to define Gromov-Witten invariants of $\widehat{X}$.  There is certainly no symplectic form $\omega$ of $\widehat{X}$ which is compatible with the complex structure in the sense that the bilinear form on tangent vectors
\begin{align}
    g(v,w)=\omega(v,Jw)\,,
\end{align}
is positive definite and symmetric.  If we had a compatible symplectic structure, then on the one hand $\int_{C_p}\omega$ would be zero since $C_p$ is trivial or torsion in homology, yet also positive being the area of $C_p$ with respect to the Riemannian metric $g$.

If one forges ahead and tries to define Gromov-Witten invariants using holomorphic maps anyway, one quickly sees that the $C_p$ create a problem.  Suppose that $[C_p]$ is $n$-torsion. Then for any $m$, one can take a tree $D=D_1\cup\ldots\cup D_{mn}$ of $mn$ copies of $\mathbb{P}^1$ and define maps $f_m:D\to\widehat{X}$ taking each $D_i$ isomorphically to $C_p$.  Then $(f_m)_*[D]=mn[C_p]=0$ for any $m$.  Now given any stable map $f:C\to\widehat{X}$ such that $f(C)$ intersects $C_p$, we can glue $f$ to $f_m$ and get a new stable map $f':C\cup D\to\widehat{X}$ with $f_*[C]=f'_*[C\cup D]$.  In this way, we get moduli spaces of stable maps $\overline{M}_g(\widehat{X},\beta)$ with infinitely many connected components, making it impossible to define a Gromov-Witten invariant.

If one tries instead to define Gopakumar-Vafa invariants to $\widehat{X}$ as invariants associated to moduli space of stable sheaves along the lines of \cite{HST,katz2008GV,kiem2008GV,Maulik2018GV}, there is the immediate problem that there is no simple notion of stability of sheaves since $\widehat{X}$ does not have a polarization which could be used to define Gieseker stability.  The best one can do is to choose a polarization $L$ on $X$ and pull it back via $\rho:\widehat{X}\to X$.  But now $\rho^*(L)\cdot C_p=0$ for each exceptional curve, and difficulties can arise with sheaves whose support contains an exceptional curve.  

However, the torsion in $H_2(\widehat{X},\mathbb{Z})$ comes to the rescue.  We have $H_2(\widehat{X},\mathbb{Z})_{\mathrm{tors}}\simeq H^3(\widehat{X},\mathbb{Z})_{\mathrm{tors}}$, so as mentioned in the previous section we have flat but topologically non-trivial B-fields on $\widehat{X}$ and a corresponding non-trivial element $\alpha\in\mathrm{Br}(\widehat{X})$.  We can circumvent the difficulties and consider enumerative invariants attached to moduli spaces of $\alpha$-twisted sheaves instead of sheaves by adapting the geometric methods mentioned above which have been used to define Gopakumar-Vafa invariants when applied to ordinary sheaves.  

For ordinary sheaves on a smooth Calabi-Yau threefold $Y$, we can view Gieseker stability as the large volume limit of Bridgeland stability conditions whose parameter space includes the punctured polydisk
\begin{align}
    \left\{B+iJ\in \left(H^2\left(Y,\mathbb{R}\right)/H^2\left(Y,\mathbb{Z}\right)\right)+i K(Y),\ J\gg0
    \right\}\,,
\end{align}
where $K(Y)\subset H^2\left(Y,\mathbb{R}\right)$ is the K\"ahler cone of $Y$.  The heart of this limiting stability condition contains $\mathrm{Coh}_{\le1}(Y)$, the abelian category of coherent sheaves on $Y$ of dimension at most 1.  We have conjectured that there are Bridgeland stability conditions on the twisted derived category $D^b(\widehat{X},\alpha)$ whose parameter space includes the punctured polydisk
\begin{align}
    \left\{B+iJ\in \left(H^2\left(\widehat{X},U(1)\right)\right)+i \rho^*K(X),\ J\gg0, c(B)=\alpha
    \right\}\,,
\end{align}
where $c(B)\in H^3(\widehat{X},\mathbb{Z})_{\mathrm{tors}}$ is the characteristic class of the B-field.  We again have a limiting notion of stability at large volume. We have further conjectured that the heart of this limiting stability condition contains $\cup\mathrm{Coh}_{\le1}(\widehat{X}')$,\footnote{If we can represent $\alpha$ by a sheaf of Azumaya algebras $\widehat{\mathcal{B}}$, then  $\widehat{\mathcal{B}}$ is isomorphic to $\mathrm{End}_{\mathcal{O}_{\widehat{X}}}(\mathcal{E})$ for some $\alpha$-twisted locally free sheaf $\mathcal{E}$.  In this situation, we can and will embed $\mathrm{Coh}_{\le1}(\widehat{X})$ in $\mathrm{Coh}_{\le1}(\widehat{X},\alpha)$ by $F\mapsto F\otimes\mathcal{E}$.} where $\widehat{X}'$ is any small resolution of $X$.  The union makes sense because of explicit Fourier-Mukai equivalences $D^b(\widehat{X}')=D^b(\widehat{X})$ \cite{Bridgeland}.

Now let $\beta\in H_2(\widehat{X},\mathbb{Z})$ and for ease of exposition we continue to assume that the torsion is $\mathbb{Z}_n$ and we have fixed an explicit splitting (\ref{eqn:splith2}).  We put $\tilde{\beta}=\rho_*(\beta)$ and can write $\beta=(\tilde{\beta},q)$ with $0\le q\le n-1$.  Let $M_\beta$ be the moduli space of stable $\alpha$-twisted sheaves of pure dimension 1 on some small resolution $\widehat{X}'$ with the above identifications.  Since $X$ is projective, we have a well-defined Chow variety $\mathrm{Chow}_{\tilde\beta,1}(X)$ of 1-dimensional cycles on $X$ and a morphism
\begin{align}
    \pi_\beta:M_\beta\to \mathrm{Chow}_{\tilde\beta,1}(X)\,.
\label{eqn:hilbertchow}
\end{align}
We apply the method of \cite{Maulik2018GV} to construct invariants $n_g^\beta=n_g^{\tilde{\beta},q}$. To rigorously define these invariants, we would need to define perverse sheaves of vanishing cycles associated to an analogue of holomorphic Chern-Simons theory for twisted sheaves and check that we can choose CY orientation data in the terminology of \cite{Maulik2018GV}.  Instead, we simply note that we will restrict our explicit computations to smooth moduli spaces $M$.  In this case, we have the perverse sheaf $M[\dim M]$ which supports a CY orientation.  We will further simplify our calculation by using the methods of \cite{Katz:1999xq} which have been shown to correctly reproduce the invariants of \cite{Maulik2018GV} under certain hypotheses \cite{Zhao}, which hold in our cases.

These torsion refined invariants satisfy the usual relation for a conifold transition
\begin{align}
    n_g^{\tilde{\beta}}(Y)=\sum_{q=0}^{j-1}n_g^{\tilde\beta,q}(\widehat{X})\,,
\end{align}
where $Y$ is a smooth complex structure deformation of $X$.

\smallskip
We also attempt to describe the invariants in terms of non-commutative resolutions.  Suppose that we can represent $\alpha\in\mathrm{Br}(\widehat{X})$ by an Azumaya algebra $\widehat{\mathcal{B}}$ on $\widehat{X}$ and suppose further that ${\mathcal{B}}:=\rho_*\widehat{\mathcal{B}}$ is such that $(X,{\mathcal{B}})$  is a noncommutative resolution of $X$ and such that 
\begin{align}
    R\rho_*:D^b(\widehat{X},\widehat{\mathcal{B}})\rightarrow D^b(X,{\mathcal{B}})
\end{align}
is a derived equivalence.  We will give examples where this condition holds.  Since sheaves of ${\mathcal{B}}$-modules are also sheaves of $\mathcal{O}_X$-modules, we have a natural notion of Gieseker stability of 1-dimensional sheaves of ${\mathcal{B}}$-modules and a natural map $M'\to \mathrm{Chow}_1(X)$ from the moduli space $M'$ of such sheaves to the Chow variety $\mathrm{Chow}_1(X)$ of curve classes on $X$, analogous to (\ref{eqn:hilbertchow}).  

In the case of a determinantal octic double solid $\Xd$ introduced in Section~\ref{sec:geometry}, we have a double cover $\pi:X=\Xd\to \mathbb{P}^3$ branched over a determinantal octic surface.  In Section~\ref{sec:evenClifford} we describe a sheaf $\mathcal{B}$ on $\Xd$ and a sheaf $\mathcal{B}_0$ of even Clifford algebras on $\mathbb{P}^3$ such that $\pi_*\mathcal{B}=\mathcal{B}_0$. 
 The categories of $\mathcal{B}$-modules on $\Xd$ and of $\mathcal{B}_0$-modules on $\mathbb{P}^3$ are equivalent.  In Section~\ref{sec:derived}, we show for these $\Xd$ that zero-dimensional sheaves of ${\mathcal{B}}_0$-modules correspond to points of some small resolution of $X$.  We have a more precise form of this correspondence in terms of $\alpha$-twisted sheaves in some cases, and the more precise correspondence holds more generally assuming a conjectured form of a derived equivalence (Conjecture~\ref{conj:hatb0}). This suggests the possibility that appropriate components of $M'$ could coincide with the moduli spaces $M_\beta$ described above.

In Section~\ref{sec:derived}, we also observe that the dimensions of certain skyscraper sheaves of points get multiplied by a numerical factor in passing from ordinary sheaves to twisted sheaves.  For similar reasons, when the map $M'\to \mathrm{Chow}_1(X)$ is restricted to such components and the image is likewise restricted, that the resulting map coincides with (\ref{eqn:hilbertchow}) up to a numerical factor which does not affect the computation of invariants using the procedure in \cite{Maulik2018GV}.  In summary, we anticipate the development of an enumerative theory of sheaves of ${\mathcal{B}}_0$-modules, which could coincide with our torsion refined invariants.

\section{Geometry}
\label{sec:geometry}
In this section we will study the geometry of determinantal double solids $\Xd$.
We first introduce the notion of a \textit{normalized decomposition} $\vec{d}\in\mathbb{N}^k$ and state Proposition~\ref{prop:11classes} in Section~\ref{sec:sods}.
A crucial role will be played by special degenerations $\Xdr$ that will be introduced in Section~\ref{sec:verysingular}.
We then prove Lemma~\ref{lem:resolution} about the structure of the K\"ahler small resolutions $\widehat{\Xdr}$.
This will allow us to prove the main result of this section, Proposition~\ref{prop:11classes}, about the homology of curves on small resolutions of $\Xd$.
\subsection{Determinantal double solids}
\label{sec:sods}
We consider double covers $X$ of $\mathbb{P}^3$ that are ramified over a determinantal surface $S$ of even degree $d\in 2\mathbb{N}$ with $0<d\le 8$.\footnote{Many of the results in this section have generalizations to more general Fano threefolds $B$  together with certain line bundles $\mathcal{L}_1,\ldots, \mathcal{L}_k$ satisfying $\mathcal{L}_1\otimes\dots\otimes \mathcal{L}_k\simeq K_B^{\otimes-2}$.}  
They can be constructed as (non-generic) hypersurfaces,
\begin{align}
	p^2=\det\,A\left(x_{1,\ldots,4}\right)\,,
 \label{eqn:doublesolid}
\end{align}
where $[x_1:\ldots:x_4:p]$ are homogeneous coordinates on $\mathbb{P}^4(1,1,1,1,d/2)$ and $A$ is a symmetric $k\times k$ matrix such that the determinant is a polynomial of degree $d$.

This can be achieved by taking a decomposition $d=\abs{\vec{d}}\equiv d_1+\ldots+d_k$ into $\vec{d}\in\mathbb{N}^k$, with $d_i=d_j\,\text{mod}\,2$ for all $i,j$, and assume that $A_{i,j}$ is a generic homogeneous polynomial of degree $(d_i+d_j)/2$~\footnote{We always use $\abs{\vec{v}}$ to denote the 1-norm of an integer vector $\vec{v}\in\mathbb{N}^k$.}.
We will denote a symmetric matrix corresponding to such a vector by $\Ad$, the associated determinantal surface by $\Sd$ and the double cover of $\mathbb{P}^3$ branched over $\Sd$ by $\Xd$.  We call $\Xd$ a determinantal double solid.  In the special case $d=8$, $\Xd$ is a Calabi-Yau variety, which we call a determinantal octic double solid.

Note that if there would be $i,j$ with $d_i\ne d_j\,\text{mod}\,2$, the corresponding entries $A_{i,j}$ would have to vanish and the matrix could be brought into block diagonal form such that the variety $X$ becomes reducible. 

Clearly $\Xd$ is independent of the order of the entries in $\vec{d}$.
It is also easy to see that two vectors $\vec{d},\vec{d}'$ that only differ by entries that are zero lead to the same $\Xd$. Consider $\vec{d}=(d_1,\ldots,d_{k'},0)$ and, using $l=k-k'$, decompose $\Ad$ as
\begin{align}
	\Ad=\left(\begin{array}{cc}
		A_{k'\times k'}&B_{l\times k'}^T\\
		B_{l\times k'}&C_{l\times l}
	\end{array}\right)\,.
\end{align}
Since the entries of $C_{l\times l}$ are constants, and we can assume $C_{l\times l}$ to be non-degenerate, there is a basis transformation
\begin{align}
	T^T\Ad T=\left(\begin{array}{cc}A-B^TC^{-1}B&0\\0&C\end{array}\right)\,,\quad T=\left(\begin{array}{cc}\mathbf{1}&0\\-C^{-1}B&\mathbf{1}\end{array}\right)\,.
\end{align}
that is globally defined on $\mathbb{P}^3$ and allows us to identify $\Xd$ with a generic double cover $\Xdp$.

It will be useful to assume that $k$ is even, which is automatically the case if the entries of the decomposition are odd and can be achieved by padding with zero if the entries are even.
This leads us to the following definition:
\begin{definition}
A decomposition $\vec{d}\in\mathbb{N}^k$ is called \textit{normalized} if $k$ is even, $d_i\ge d_j$ for $i<j$, $d_i=d_j\text{ mod }2$ and $d_{k-1}\ne 0$, i.e. $\vec{d}$ contains at most one zero.
We call the number of non-zero entries the \textit{length} of the decomposition and $d=\abs{\vec{d}}$ the \textit{degree}.
We denote a decomposition as even (odd) if all of the entries are even (odd).
\end{definition}

The following is implied by~\cite[Theorem 2.2]{Catanese1981}, see also~\cite[Theorem 1.10]{Harris1984}:
\begin{lemma}
Let $\vec{d}\in\mathbb{N}^k$ be a normalized decomposition.
A generic determinantal surface $\Sd\subset\mathbb{P}^3$ has isolated nodes
{\normalfont\begin{align}
	\SdOne=\{\,\corank\,\Ad=2\,\}\,,
\end{align}}
and is smooth everywhere else.
The number of nodes $n_{\vec{d}}\equiv \# \SdOne$ is given by
\begin{align}
	n_{\vec{d}}=\frac12\sum\limits_{i=2}^k d_i\rho_{i}\rho_{i-1}\,,\quad \rho_h=\sum_{i=1}^hd_i\,.
\end{align}
\label{lem:nodes}
\end{lemma}

These nodes persist in the double cover $\Xd$.  If $\Sd$ is described locally near a node by the equation $x^2+y^2+z^2=0$, then $\Xd$ is locally described by the equation $w^2=x^2+y^2+z^2$, which is also a node.

The goal of the rest of Section~\ref{sec:geometry} will be to prove the following result:
\begin{proposition}
    Let $\vec{d}\in\mathbb{N}^k$ be an even or odd decomposition of degree $d\in 2\mathbb{N}$ with $0<d\le 8$ and let $l\in\mathbb{N}$ be the length of the decomposition.
    The corresponding double covers $\Xd$ fall into three different classes:
\begin{enumerate}\setlength\itemsep{1em}
	\item $l=1$: The branch locus $\Sd$ is smooth and $X_{(d)}$ is a generic hypersurface of degree $d$ in $\mathbb{P}^4(1,1,1,1,d/2)$ with $H_2(X_{(d)})=\mathbb{Z}$.
	\item $l=2$: $\Xd$ has isolated nodal singularities but admits a K\"ahler small resolution $\widehat{\Xd}$ with $H_2(\widehat{\Xd})=\mathbb{Z}^2$.
	\item $l\ge 3$: $\Xd$ has isolated nodal singularities and does not admit a K\"ahler small resolution.
		In those cases the exceptional curves in a small non-K\"ahler resolution $\widehat{\Xd}$ are $2$-torsion and $H_2(\widehat{\Xd})=\mathbb{Z}\oplus\mathbb{Z}_2$.
\end{enumerate}
	\label{prop:11classes}
\end{proposition}

While the first case is (almost) trivial, we will prove the cases $k\ge 2$ of the proposition in Section~\ref{sec:proofProp}.
In the third case we will proceed by describing certain specializations $A^{\text{r}}_{k\times k}$ of the corresonding matrices, such that additional nodes arise in the associated double cover $\Xdr$.
We will then construct small K\"ahler resolutions $\widehat{\Xdr}$ of $\Xdr$ as complete intersections in toric ambient spaces.
Finally, we argue that under the transition $\widehat{\Xdr}\rightarrow \widehat{\Xd}$, the exceptional curves that are preserved in $\widehat{\Xd}$ are $2$-torsional.
The second case corresponds to $\widehat{\Xdr}= \widehat{\Xd}$ and follows directly.

\begin{table}[ht!]
\centering
	\begin{align*}\renewcommand{\arraystretch}{1.4}
		\begin{array}{|c|c|c|c|}\hline
			\vec{d}&n_{\vec{d}}&H_2(\widehat{\Xd},\mathbb{Z})&\\\hline
			(8)			& 0	& \mathbb{Z}			& \\\hline\hline
						&	&				& \widehat{\Xd} \\\hline
			(7,1)			& 28	& \mathbb{Z}^2			& X[3\vert 1]\\
			(6,2)			& 48	& \mathbb{Z}^2			& X[2\vert 2]\\
			(5,3)			& 60	& \mathbb{Z}^2			& X[1\vert 3]\\
			(4,4)			& 64	& \mathbb{Z}^2			& X[0\vert 4]\\\hline\hline
						&	&				& \widehat{\Xdr}\\\hline
			(5,1^3)		& 64	& \mathbb{Z}\times\mathbb{Z}_2	& X[2,0\vert 1^2]\\
			(4,2^2)			& 72	& \mathbb{Z}\times\mathbb{Z}_2	& X[1,0\vert 2,1]\\
			(3^2,1^2)		& 76	& \mathbb{Z}\times\mathbb{Z}_2	& X[1^2\vert 1^2]\\
			(3,1^5)		& 80	& \mathbb{Z}\times\mathbb{Z}_2	& X[1,0^2\vert 1^3]\\
				(2^4)		& 80	& \mathbb{Z}\times\mathbb{Z}_2	& X[0^2\vert 2^2]\\
					(1^8)	& 84	& \mathbb{Z}\times\mathbb{Z}_2	& X[0^4\vert 1^4]\\\hline
		\end{array}
	\end{align*}
	\caption{Determinantal octic double solids, their degenerations and resolutions.}
	\label{eqn:detocdos}
\end{table}
 Table~\ref{eqn:detocdos} lists the octic cases $\abs{\vec{d}}=8$, providing the number of nodes of $\Xd$ as well as $H_2(\widehat{\Xd},\mathbb{Z})$.  For the length two cases, the data for $\widehat{\Xd}$ is provided in the notation $X[\vec{a}\vert\vec{b}]$ that will be introduced in Section~\ref{sec:smallrescompint}.  In the cases of length at least three,  the data is provided for $\widehat{\Xdr}$\,.

\subsection{Degenerate determinantal double solids}
\label{sec:verysingular}
Before we study the small resolutions of our determinantal double solids $\Xd$, we will first introduce special degenerations $\Xdr$ of $\Xd$.
When the length of $\vec{d}$ is at least 3, the $\Xdr$ have additional nodes.
In Section~\ref{sec:smallrescompint} we will then explicitly construct K\"ahler small resolutions $\widehat{\Xdr}$.

We consider a symmetric $\Ad$ matrix associated to some vector $\vec{d}\in\mathbb{N}^{2n}$ as before.
However, now we will distinguish vectors that differ in the order of their entries or by entries that are zero.

We can then define the associated degenerations
\begin{align}
	A_{\vec{d}}=\left(\begin{array}{cc}
		A_{k'\times k'}&B_{l\times k'}^{T}\\
		B_{l\times k'}&C_{l\times l}
	\end{array}\right)\quad\rightarrow\quad
	A_{\vec{d}}^{\text{r}}=\left(\begin{array}{cc}
		A_{k'\times k'}&B_{l\times k'}^{T}\\
		B_{l\times k'}&0
	\end{array}\right)\,,
 \label{eqn:adegen}
\end{align}
with $l=n-1$ and $k'=n+1$ and use $\Sdr$ to refer to the corresponding surface.
We refer to the corresponding double cover $\Xdr$ as a \textit{degenerate determinantal double solid}.

\begin{lemma}
	A generic surface $\Sdr$ has isolated nodes $\SdrOne\cup \SdrTwo$, where
	{\normalfont\begin{align}
	\SdrOne=\{\,\corank\,\Adr=2\,\}\,,\quad \SdrTwo=\{\,\corank\,B_{l\times k'}=1\,\}\,,
	\end{align}}
	and is smooth everywhere else.
	\label{lemma:genericLemma}
\end{lemma}
\begin{proof}
	Step 1: We first show that $\SdrTwo$ are isolated nodes of $\Sdr$. It is clear that $\SdrTwo\subset\Sdr$.
	The fact that $\SdrTwo$ consists of isolated points follows from~\cite[Proposition 1.3]{Porteous1971}, see also~\cite{Harris1984}, and one furthermore obtains
	\begin{align}
		\#\SdrTwo=\int\limits_{\mathbb{P}^3}c(E')/c(F')\,,\quad E'=\bigoplus\limits_{i=1}^{k'}\mathcal{O}_{\mathbb{P}^3}(-d_i H/2)\,,\quad F'=\bigoplus\limits_{i=k'+1}^l\mathcal{O}_{\mathbb{P}^3}(d_i H/2)\,.
  \label{eqn:sd2}
	\end{align}
A simple way to interpret (\ref{eqn:sd2}) is by treating $\mathcal{O}_{\mathbb{P}^3}(aH)$ formally for any $a\in\mathbb{Q}$ and assigning a chern class
\begin{align}
c\left(\bigoplus_i\mathcal{O}_{\mathbb{P}^3}(a_iH)\right)=\prod_i\left(1+a_iH\right) \in H^*(\mathbb{P}^3,\mathbb{Q})\,.
\end{align}
While the resulting integral computing $\#\SdrTwo$ is a priori rational, the proof of the formula shows that it is an integer which correctly counts $\#\SdrTwo$.

	Given a point $p\in\SdrTwo$, we can assume that $\det\,A_{k'\times k'}\ne 0$ and choose local analytic coordinates $x_1,x_2,x_3$ on a neighbourhood $U\simeq\mathbb{C}^3$ as well as a suitable basis such that to leading order
	\begin{align}
		\Adr\simeq \left(\begin{array}{ccc|ccc|ccccc}
			1& & &&&&x_1&&&&\\
			 &1& &&&&x_2&&&&\\
			 & &1&&&&x_3&&&&\\\hline
			 & & &1&&&&1&&&\\
			 & & &&\ddots&&&&\ddots&&\\
			 & & &&&1&&&&1&\\\hline
			x_1&x_2&x_3&&&&&&&&\\
			&&&1&&&&&&&\\
			&&&&\ddots&&&&&&\\
			&&&&&1&&&&&
		\end{array}\right)\,.
	\end{align}
	The determinant is $\det\,\Adr\simeq -x_1^2-x_2^2-x_3^2$ and therefore $p$ is a node.

	Step 2: Let us now show that $\SdrOne$ are also isolated nodes of $\Sdr$. 
	We slightly modify the proof of~\cite[Theorem 2.2]{Catanese1981}.
	We denote the determinants of the $(2n-1)\times (2n-1)$ minors of $\Adr$ that are obtained by deleting the $i$-th row and the $j$-th column by $b_{ij}$.
	Similarly we denote those of the $(2n-2)\times (2n-2)$ minors with deleted rows $j,l$ and deleted columns $i,k$ by $d^{jl}_{ik}$.
	The determinants of the $l\times l$ minors of $B_{l\times k'}$ that are obtained by deleting the columns $i,j$ will be denoted by $\tilde{b}_{ij}$.

	For a given point $p\in \SdrOne$, by genericity we have that $\rank\,B_{l\times k'}=l$ so there exists at least one pair $i\ne j \le k'$ with $\tilde{b}_{i,j}(p)\ne 0$.
	By~\cite[Theorem~2]{Silvester2000} it follows that $d^{i,j}_{i,j}=(-1)^{l}\tilde{b}_{i,j}^2$, so that $d^{i,j}_{i,j}(p)\ne 0$.
	The proof of~\cite[Theorem 2.2]{Catanese1981} implies that the intersection of $b_{ii}=0$, $b_{jj}=0$, and $b_{i,j}=0$ is transverse and, according to the inverse function theorem, we can use these minors as local coordinates around $p$.
	It then follows from the relation
	\begin{align}
		d^{ij}_{ij}\cdot \text{det}\,\Adr=b_{ii}b_{jj}-b_{ij}^2\,.
		\label{eqn:minorrelation}
	\end{align}
	that $p$ is an isolated node.

	Step 3: We claim that $\Sdr\backslash (\SdrOne\cup \SdrTwo)$ is smooth. The proof is deferred to the proof of Lemma~\ref{lem:resolution}.
\end{proof}

The number $\#\SdrOne$ of nodes of the first type are again given by (\ref{lem:nodes}), as the calculation can be reformulated as
\begin{align}
	\#S_{\vec{d},1}=4\left[c_1(E)c_2(E)-c_3(E)\right]\,,\quad E=\bigoplus\limits_{i=1}^k\mathcal{O}_{\mathbb{P}^3}(d_i H/2)\,,
\end{align}
and so the degenerated form of $\Adr$ does not matter in the calculation.  

\subsection{Small resolutions by complete intersections}\label{sec:smallrescompint}
In this section we construct complete intersections in projective bundles that are small K\"ahler resolutions of the degenerate determinantal double solids introduced in Section~\ref{sec:verysingular}.

Given a vector $\vec{a}\in\mathbb{N}^n$ with $n>0$ and $\abs{a}\le 4$, we consider the rank $n+1$ vector bundle
\begin{align}
	E=\mathcal{O}_{\mathbb{P}^3}\oplus\bigoplus\limits_{i=1}^{n}\mathcal{O}_{\mathbb{P}^3}(-a_i)\,,
	\label{eqn:Ebundle}
\end{align}
on $\mathbb{P}^3$.
The inclusion $\mathcal{O}_{\mathbb{P}^3}\xhookrightarrow{} E$ induces a holomorphic section $S$ of the projective bundle $\mathbb{P}(E)$  and this allows us to identify $\mathbb{P}^3$ with a subspace of $\mathbb{P}(E)$.  We also have the projection $\pi:\mathbb{P}(E)\to \mathbb{P}^3$.
We denote the hyperplane class on $\mathbb{P}^3$ by $H$.
The tautological bundle on $\mathbb{P}(E)$ is $\mathcal{O}_{\mathbb{P}(E)}(-1)$ and the first Chern class of its inverse will be denoted by $\xi$.

Note that $\mathbb{P}(E)$ is a toric variety and the Mori cone on $\mathbb{P}(E)$ is generated by $C_F=\xi^{n-1}\cdot \pi^{*}(H)^3,\,C_B=S\cdot \pi^{*}(H)^2$ while the dual basis of the K\"ahler cone is given by $J_1=\xi,\,J_2=\pi^{*}(H)$.
The anti-canonical class is
\begin{align}
	-K_{\mathbb{P}(E)}=(n+1)J_1+\left(4-\abs{\vec{a}}\right)J_2\,.
\end{align}

Given a vector $\vec{b}\in\mathbb{N}^n$ we define the effective divisor classes
\begin{align}
	D_1=2J_1+b_1J_2\,,\quad D_{i=2,\ldots,n}=J_1+b_iJ_2\,,
\end{align}
and denote the three dimensional complete intersection of $n$ generic divisors in the corresponding linear systems by $\Xab$.
\begin{lemma}
    Given $n>0$ and  $\vec{a},\vec{b}\in\mathbb{N}^n$ with $\abs{a}\le 3$, $\abs{\vec{b}}\le 4-\abs{\vec{a}}$ and $b_i>0$, the associated complete intersection $X=\Xab$ is smooth and $H_2(X,\mathbb{Z})=\mathbb{Z}^2$.
    \label{lem:cicy}
\end{lemma}
\begin{proof}
    The classes $D_{i=1,\ldots,n}$ are contained in the K\"ahler cone of $\mathbb{P}(E)$. Since $\mathbb{P}(E)$ is a smooth toric variety, this implies that the $D_i$ are basepoint-free and ample. The claims respectively follow from Bertini's theorem and the Lefschetz hyperplane theorem.
\end{proof}

\begin{lemma}
        Consider a normalized decomposition $\vec{d}\in\mathbb{N}^{2n}$ of length greater than one.
	The Calabi-Yau threefold $X= \Xab$ with
	\begin{align}
		a_i=\frac{d_i-d_{n+1}}{2}\,,\quad b_j=\frac{d_{n+1}+d_{n+j}}{2}\,,
	\end{align}
	for $i,j=1,\ldots,n$, is a smooth and simply connected small K\"ahler resolution of $\Xdr$ with $H_2(X,\mathbb{Z})=\mathbb{Z}^2$.
	The homology classes $C^{(i)}$ of the exceptional curves that resolve the nodes $p\in\Xdr\cap \SdrOneTwo$ are independent of $p$ and satisfy $C^{(2)}=2C^{(1)}$.
	\label{lem:resolution}
\end{lemma}
\begin{proof}
	First note that $\abs{\vec{a}}+\abs{\vec{b}}=d/2\le 4$ and therefore in particular $\abs{\vec{a}}\le 4$.
        Moreover, since the length of $\vec{d}$ is greater than one, we have $d_{n+1}>0$ and therefore $b_{j=1,\ldots,n}>0$.
        Lemma~\ref{lem:cicy} then implies that $X=\Xab$ is smooth and $H_2(X,\mathbb{Z})=\mathbb{Z}^2$.

        Let again $l=n-1$ and $k'=n+1$.
	Denoting the homogeneous coordinates on the base $\mathbb{P}^3$ and the fiber $\mathbb{P}^{n}$ respectively by $x=[x_1:\ldots:x_4]$ and $y=[y_1:\ldots:y_{k'}]$, the defining equations of $\Xab$ take the form
\begin{align}
	B_{l\times k'}(x)y=0\,,\quad y^{\,T}A_{k'\times k'}(x){y}=0\,,
 \label{eqn:Xabequations}
\end{align}
	where the entries $a_{i,j}$ of the matrix
	\begin{align}
	A_{\vec{d}}^{\text{r}}=\left(\begin{array}{cc}
		A_{k'\times k'}&B_{l\times k'}^{T}\\
		B_{l\times k'}&0
	\end{array}\right)\,,
	\end{align}
	are generic homogeneous polynomials of degree $(d_i+d_j)/2$.

 We refer to the space of solutions for $y$ in the first equation in (\ref{eqn:Xabequations}) as the kernel of $B_{l\times k'}(x)$.  This kernel is isomorphic to $\mathbb{P}^2$ if $x\in\SdrTwo$ and is isomorphic to $\mathbb{P}^1$ otherwise, linearly embedded in either case.
We see that
over a generic point $x\in\mathbb{P}^3$ the solution set of (\ref{eqn:Xabequations}) for ${y}\in\mathbb{P}^n$ consists of two points.

The resulting generic double cover over $\mathbb{P}^3$ is ramified over points $x\in\mathbb{P}^3$ where the restriction of $A_{k'\times k'}(x)$ to the kernel of $B_{l\times k'}(x)$ has itself a non-trivial kernel
\begin{align}
\begin{split}
	&\text{ker}\left(A_{k'\times k'}(x)\big\vert_{\text{ker}\,B_{l\times k'}(x)}\right)\\
 =&\{\,{y}\in\text{ker}\,B_{l\times k'}(x)\,\,\vert\,\,{z}^{\,T}A_{k'\times k'}{y}=0,\,\forall {z}\in\text{ker}\,B_{l\times k'}(x)\,\}\,.
 \end{split}
\end{align}
	In other words, there has to exist some ${y}\in\mathbb{C}^{k'}\backslash \{0\}$ and ${p}\in\mathbb{C}^l$, such that
\begin{align}
	B_{l\times k'}(x){y}=0\,,\quad A_{k'\times k'}(x){y}+B_{l\times k'}(x)^T{p}=0\,.
\end{align}
Away from the points $\SdrTwo=\{\,\corank\,B_{l\times k'}=1\,\}$, this coincides with
\begin{align}
	\Sdr=\{\,\text{det}\,\Adr=0\,\}\,,\quad 
	\Adr=\left(\begin{array}{cc}
			A_{k'\times k'}&B_{l\times k'}^T\\B_{l\times k'}&0
		\end{array}\right)\,.
\end{align}

	Over $\Sdr\backslash(\SdrOne\cup\SdrTwo)$ the solution set in $\mathbb{P}^n$ is a point.
	On the other hand, over a point $p\in\SdrOne$ one obtains a line $C^{(1)}_p$ and over a point $p\in\SdrTwo$ a smooth plane conic curve $C^{(2)}_p$, so that
        \begin{align}
            C^{(i)}_p\cdot (J_1,J_2)=(i,0)\,.
            \label{eqn:cijint}
        \end{align}
        Since $H_2(X,\mathbb{Z})=\mathbb{Z}^2$ it follows that the corresponding homology classes $C^{(i)}$ are independent of $p$ and satisfy $C^{(2)}=2C^{(1)}$.
 
	Contracting these curves produces a small resolution $\pi:\,\Xab\rightarrow \Xdr$.  More precisely, if we let $\rho:\Xdr\to \mathbb{P}^3$ be the projection, then $\rho^*(\Sd)$ is a divisor on $\Xdr$ with multiplicity 2, so that $\rho^*(\Sd)=2D$ for some divisor $D$ on $\Xdr$. Since $\mathcal{O}_{\Xdr}(\rho^*(\Sd))\simeq\mathcal{O}_{\Xdr}(dJ_2)$, we have $\mathcal{O}_{\Xdr}(D)\simeq\mathcal{O}_{\Xdr}(\frac{d}2J_2)$, and $D$ is the divisor of a section $s_D\in H^0(\Xdr,\mathcal{O}_{\Xdr}(\frac{d}2J_2)$.  The contraction map $\pi$ is then defined by
 \begin{align}
     \Xab\rightarrow \mathbb{P}\left(1,1,1,1,\frac{d}2\right),\qquad (x_1,\ldots,x_4,\vec{y})\mapsto (x_1,\ldots,x_4,s_D)\,,
 \end{align}
 where $s_D$ has been multiplied by a scalar if necessary so that the equation of the image of $\pi$ is precisely (\ref{eqn:doublesolid}).

 Since $\Xab$ is smooth and $\pi$ is an isomorphism on the complement of $\SdOne\cup\SdrTwo$ (identified with its preimage in $\Xdr$), we conclude that $\Xdr$ is smooth on the complement of $\SdOne\cup\SdrTwo$, hence $\Sdr$ is also smooth on the complement of $\SdOne\cup\SdrTwo$.
	
Simply connectedness follows from~\cite[Corollary 1.19]{Clemens1983}.
\end{proof}

Note that the special case $d=8,\vec{a}=0$ has also been discussed in~\cite{Green:1988bp,Candelas:1989ug}.

\subsection{Proof of Proposition~\ref{prop:11classes}}
\label{sec:proofProp}
\begin{proof}
	The case $l=1$ follows from Lemma~\ref{lemma:genericLemma} together with \cite[Lemma 1.23]{Clemens1983}.

	Case $l=2$: For a normalized $\vec{d}\in\mathbb{N}^2$ we note that $\Xd=\Xdr$ and the claim follows from Lemma~\ref{lem:resolution}.

	Case $l\ge 3$: We again replace $\vec{d}$ by the corresponding normalized decomposition and can then consider the degeneration $\Xdr$ and the corresponding small resolution $X=\Xab$ as in Lemma~\ref{lem:resolution}.
	We denote again the exceptional curve over $p\in \Xdr\cap \SdrOneTwo$ by $C^{(i)}_p$ and the corresponding homology classes by $C^{(i)}$.
	From Lemma~\ref{lem:resolution} it follows that $H_2(X,\mathbb{Z})=\mathbb{Z}^2$ and
	\begin{align}
		[C^{(2)}]=2[C^{(1)}]\,.
	\end{align}

	Now $X$ becomes homotopy equivalent to a small resolution $\widehat{\Xd}$ of $\Xd$ after gluing the boundary of a 3-disk to each $C^{(2)}_p$.
	As a result, $C^{(2)}$ becomes trivial while $C^{(1)}$ becomes 2-torsional.

	More explicitly, we have the Mayer-Vietoris sequence
	\begin{align}
		0\rightarrow H_3(\widehat{\Xdr},\mathbb{Z})\rightarrow H_3(\widehat{\Xd},\mathbb{Z})\rightarrow \mathbb{Z}^{\#\SdrTwo}\rightarrow H_2(\widehat{\Xdr},\mathbb{Z})\rightarrow H_2(\widehat{\Xd},\mathbb{Z})\rightarrow 0\,.
	\end{align}
	The relations $C^{(2)}_p-C^{(2)}_{p'}=\partial S^{(3)}_{p,p'}$, with $S^{(3)}_{p,p'}$ a 3-chain on $\widehat{\Xdr}$ for each $p,p'\in \SdrTwo$, reduce this to
	\begin{align}
		0\rightarrow\mathbb{Z}\rightarrow H_2(\widehat{\Xdr},\mathbb{Z})\rightarrow H_2(\widehat{\Xd},\mathbb{Z})\rightarrow 0\,,
  \label{eqn:h2fromtransition}
	\end{align}
	with the image of $\mathbb{Z}$ being precisely the subgroup generated by $C^{(2)}$.
	As a result we obtain that $H_2(\widehat{\Xd},\mathbb{Z})\simeq\mathbb{Z}\oplus\mathbb{Z}_2$.

	Using~\cite[Theorem 3.5]{Werner,WernerTranslate}, this also implies that $\widehat{\Xd}$ is not K\"ahler.

\end{proof}

\subsection{Higgs transitions in M-theory}
The geometric transition from $\widehat{\Xdr}$ to $\Xd$ which we have used in order to determine $H_2(\widehat{\Xd},\mathbb{Z})$ has a direct physical interpretation as a Higgs transition in M-theory.

As is well known from physics, see e.g.~\cite{Witten:1996qb}, the M-theory 3-form field $C_3$ can be expanded along harmonic 2-forms $\omega_{i=1,\ldots,b_2(X)}$ on a Calabi-Yau threefold $X$ as
\begin{align}
    C_3=\ldots +\sum_{i=1}^{b_2(X)}\omega^i\wedge A_i\,,
\end{align}
in order to generate $U(1)$ gauge fields $A_{i=1,\ldots,b_2(X)}$ in the five-dimensional effective theory.
Morever, charged hypermultiplets arise from M2-branes wrapping curves $C$ in the Calabi-Yau with the $U(1)^{b_2(X)}$ charges given by $q_i=\omega_i\cdot C$.

For M-theory on $\widehat{\Xdr}$ it follows from Lemma~\ref{lem:resolution} that the gauge symmetry which comes from the harmonic form respectively associated to $J_1,J_2$, is $U(1)_1\times U(1)_2$.
Using~\eqref{eqn:cijint} one also finds that the $\#\SdrOne+\#\SdrTwo$ hypermultiplets $\Phi^{(1)},\Phi^{(2)}$ associated to M2-branes wrapping exceptional curves $C^{(i)}$ have charges $(i,0)$.

Note that the hypermultiplets $\Phi^{(2)}$ only exist of $\#\SdrTwo>0$ and therefore if the length of the decomposition $\vec{d}$ is $l\ge 3$, while for $l=2$ one has $\Xd=\Xdr$.

The mass of the hypermultiplets $\Phi^{(1)},\Phi^{(2)}$ is proportional to the volume of the corresponding exceptional curves and physically we are on the Coulomb branch of the theory.
Contracting the curves, going from $\widehat{\Xdr}$ to $\Xdr$, means that we set the Coulomb branch parameter associated to $U(1)_1$ to zero so that the hypermultiplets become massless.

On the origin of the $U(1)_1$ Coulomb branch we can then turn on a vacuum expectation value for the scalar fields in $\Phi^{(1)},\Phi^{(2)}$.
Turning on a generic vacuum expectation value for the scalars in $\Phi^{(2)}$, while keeping the expectation value of scalars in $\Phi^{(1)}$ at zero, geometrically means that we deform away the nodes $\SdrTwo$ and therefore deform $\Xdr$ into $\Xd$.

Physically, the vacuum expectation value of the charge $2$ scalars breaks the gauge symmetry from $U(1)_1\times U(1)_2$ to $\mathbb{Z}_2\times U(1)_2$ and the remaining massless hypermultiplets $\Phi^{(1)}$ carry non-trivial charge under the discrete gauge symmetry.

From the perspective of M-theory on $\Xd$ we therefore find the following situation, depending on the length $l$ of the decomposition $\vec{d}$:
\begin{enumerate}
\setlength{\itemsep}{5pt}
\item $l=1$: M-theory on $X_{(8)}$ has gauge symmetry $U(1)$ and no massless charged matter.
\item $l=2$: M-theory on $\Xd\simeq \Xdr$ has gauge symmetry $U(1)_1\times U(1)_2$ and $\#\SdOne$ massless hypermultiplets with charge $(1,0)$.
\item $l\ge 3$: M-theory on $\Xd$ has gauge symmetry $\mathbb{Z}_2\times U(1)_2$ and $\#\SdOne$ massless hypermultiplets with charge $(-,0)$.
\end{enumerate}

This is the physical equivalent of Proposition~\ref{prop:11classes} and explicitly confirms our more general claim from~\cite{Katz:2022lyl} that discrete symmetries in M-theory arise from torsion in the homology of curves in small resolutions, even if the latter are non-K\"ahler.

\section{Brauer group and sheaves of Clifford algebras}
\label{sec:cliffordresolution}
We will now discuss from a mathematical perspective the Clifford non-commutative resolutions $\Xdnc$ of the determinantal octic double solids $\Xd$ associated to decompositions $\vec{d}$ of degree $d=8$ and length $l\ge 3$.
Readers that are less interested in the mathematical details can safely skip this section but might still appreciate the discussion of the representation theory of Clifford algebras in Section~\ref{sec:Clifford}.

We first discuss the general construction of the Clifford non-commutative resolutions in Section~\ref{sec:evenClifford}.
Then, in Section~\ref{sec:k4}, we focus on decompositions of length $l=3,4$ and explicitly describe how the exceptional curves in any small resolution $\widehat{\Xd}$ correspond to modules over a sheaf of Clifford algebras on $\mathbb{P}^3$.
The representation theory of the latter will be described in Section~\ref{sec:Clifford}.
The derived equivalence between $(\widehat{X},\widehat{\mathcal{B}})$ and $(\mathbb{P}^3,\mathcal{B}_0)$ as well as a conjectured generalization to decompositions of length $l=6,8$ will then be discussed in Section~\ref{sec:derived}.

\subsection{Sheaves of even Clifford algebras and Azumaya algebras}
\label{sec:evenClifford}
Consider again a normalized decomposition $\vec{d}\in\mathbb{N}^{k}$ with $k=2n$. 
If the decomposition is even, we put 
\begin{align}
    \mathcal{L}=\mathcal{O}_{\mathbb{P}^3}, \qquad W=\bigoplus_{i=1}^k\mathcal{O}_{\mathbb{P}^3}\left(-\frac{d_i}2\right)\,.
\label{eqn:Weven}
\end{align}
Otherwise, if the decomposition is odd, we put 
\begin{align}
\mathcal{L}=\mathcal{O}_{\mathbb{P}^3}(-1),\qquad W=\bigoplus_{i=1}^k\mathcal{O}_{\mathbb{P}^3}\left(\frac{1-d_i}2\right)\,.
\label{eqn:Wodd}
\end{align}
In either case, a symmetric matrix $\Ad$ can be interpreted as a map $\sigma:\,\mathcal{L}\rightarrow S^2W^\vee$.  

We denote the associated sheaf of even parts of Clifford algebras on $\mathbb{P}^3$ that has been constructed in~\cite{Kuznetsov2008} by $\mathcal{B}_0$.  As an $\mathcal{O}_{\mathbb{P}^3}$-module, we have
\begin{align}
    \mathcal{B}_0\simeq \mathcal{O}_{\mathbb{P}^3}\oplus \left(\Lambda^2W\otimes\mathcal{L} \right)\oplus\left(\Lambda^4W\otimes \mathcal{L}^2\right)\oplus \cdots \oplus \left(\Lambda^kW \otimes \mathcal{L}^{n}\right)
\,.
\end{align}
Let $p\in \mathbb{P}^3$ and denote by $\mathrm{Cl}(\Ad(p))$ the Clifford algebra associated to the quadratic form determined by the matrix $\Ad(p)$.  Then $\mathcal{B}_0\otimes\mathcal{O}_p\simeq \mathrm{Cl}^0(\Ad(p))$.

It was shown in~\cite{Kuznetsov2008} that $Z(\mathcal{B}_0)=\mathcal{O}_{\mathbb{P}^3}\oplus(\Lambda^kW \otimes \mathcal{L}^{n})$ is the center of $\mathcal{B}_0$.  
We have $\Lambda^kW \otimes \mathcal{L}^{n}\simeq\mathcal{O}_{\mathbb{P}^3}(-d/2)$, and the map $\mathcal{O}_{\mathbb{P}^3}(-d/2)\otimes \mathcal{O}_{\mathbb{P}^3}(-d/2)\to \mathcal{O}_{\mathbb{P}^3}$ used to describe multiplication in $Z(\mathcal{B}_0)$ is just multiplication by $\det(\Ad)$.  Thus $\pi:\Xd\to\mathbb{P}^3$ is the relative spectrum
\begin{align}
    \Xd=\mathrm{Spec}_{\mathcal{O}_{\mathbb{P}^3}}\left(Z\left(\mathcal{B}_0\right)\right)\,.
\end{align}
From this description, we see that $\pi_*\mathcal{O}_{\Xd}=Z(\mathcal{B}_0)$.  Since $\mathcal{B}_0$ is clearly a $Z(\mathcal{B}_0)$-module, we can view $\mathcal{B}_0$  as an $\mathcal{O}_{\Xd}$-module.  More precisely, we have a sheaf ${\mathcal{B}}$ of $\mathcal{O}_{\Xd}$-modules such that $\pi_*{\mathcal{B}}=\mathcal{B}_0$.

It was shown in \cite{Kuznetsov2008} that  ${\mathcal{B}}$  restricts to an Azumaya algebra on the complement of the nodes $\SdOne$.  However, ${\mathcal{B}}$ is not Azumaya at the nodes.  Let $f:\widehat\Xd\to \mathbb{P}^3$ be the natural map. By the isomorphism 
\begin{align}
    \Xd-\SdOne\simeq f^{-1}\left(\Xd-\SdOne\right)
\end{align}
we can view ${\mathcal{B}}\vert_{\Xd-\SdOne}$ as an Azumaya algebra on $f^{-1}\left(\Xd-\SdOne\right)$.  It would be desirable to extend this Azumaya algebra to all of $\widehat{\Xd}$.  In the next section, we explain how to do this for the length $l=3,4$ cases, with $k=4$, following \cite{Kuznetsov2013}.

\subsection{The k=4 case.}
\label{sec:k4}

In~\cite{Kuznetsov2013}, an Azumaya algebra $\widehat{\mathcal{B}}$ on $\widehat{\Xd}$ was constructed in the $k=4$ case extending 
${\mathcal{B}}\vert_{\Xd-\SdOne}$ .  It was further shown that $f_*(\widehat{\mathcal{B}})={\mathcal{B}}_0$ and the induced map 
\begin{align}
    Rf_*:D^b(\widehat\Xd,\widehat{\mathcal{B}})\to D^b(\mathbb{P}^3, \mathcal{B}_0)
\label{eq:derivedequiv}
\end{align}
is a derived equivalence with inverse
\begin{align}
  D^b(\mathbb{P}^3, \mathcal{B}_0)\to   D^b(\widehat\Xd,\widehat{\mathcal{B}}), \qquad F\mapsto Lf^*(F)\stackrel{L}{\otimes}_{\mathcal{B}_0}\widehat{\mathcal{B}}\,.
  \label{eq:inversederivedequiv}
\end{align}
If $\alpha$ is the Brauer class represented by $\widehat{\mathcal{B}}$, then if desired we can also express (\ref{eq:derivedequiv}) as
\begin{align}
    D^b(\widehat{\Xd},\alpha) = D^b(\mathbb{P}^3,\mathcal{B}_0)\,.
\end{align}

For a node $p$, let $C_p\subset\widehat{\Xd}$ be the exceptional $\mathbb{P}^1$ over $p$.  It was also shown in~\cite{Kuznetsov2013} that $\widehat{\mathcal{B}}\vert_{C_p}\simeq \underline{End}_{\mathcal{O}_{\mathbb{P}^1}}(\mathcal{O}\oplus \mathcal{O}(-1))$.  It follows that $\mathcal{O}_{C_p}\oplus \mathcal{O}_{C_p}(-1)$ is a $\widehat{\mathcal{B}}$-module supported on $C_p$.  More generally, there is a 1-1 correspondence between sheaves of $\mathcal{O}_{C_p}$-modules and $\widehat{\mathcal{B}}$-modules supported on $C_p$ given by 
\begin{align}
    F\mapsto F\otimes_{\mathcal{O}_C}\left(\mathcal{O}_{C_p}\oplus \mathcal{O}_{C_p}(-1)\right)\,.
\label{eqn:CtoB}
\end{align}
Under this correspondence, $\mathcal{O}_{C_p}$ corresponds to $\mathcal{O}_{C_p}\oplus \mathcal{O}_{C_p}(-1)$.

This correspondence allows us to make an explicit identification between the derived categories of $\widehat{\mathcal{B}}$-modules supported scheme-theoretically on $C_p$ and $\mathcal{B}_0$-modules supported supported scheme-theoretically on $p$, i.e.\ $\mathcal{B}_0\otimes\mathcal{O}_p=\mathrm{Cl}^0(\Ad(p))$-modules.  

For illustration, we start with the short exact sequence of $\mathcal{O}_{C_p}$-modules
\begin{align}
    0 \rightarrow \mathcal{O}_{C_p}(-1) \rightarrow  \mathcal{O}_{C_p} \to \mathcal{O}_q\rightarrow 0\,,
\label{eqn:P1ptses}
\end{align}
where $q\in C_p$.
We arrive at a short exact sequence of sheaves of $\widehat{\mathcal{B}}$-modules
\begin{align}
0\rightarrow\mathcal{O}_{C_p}(-1)\oplus\mathcal{O}_{C_p}(-2) \rightarrow\mathcal{O}_{C_p}\oplus\mathcal{O}_{C_p}(-1) \rightarrow\mathcal{O}_q\oplus\mathcal{O}_q(-1)\rightarrow 0\,.
\label{eqn:sesB}
\end{align}
\noindent

Applying $Rf_*$ and noting that $\mathcal{O}_{C_p}\oplus\mathcal{O}_{C_p}(-1)$ and $\mathcal{O}_q\oplus\mathcal{O}_q(-1)$ have no higher cohomology, while $\mathcal{O}_{C_p}(-1)\oplus\mathcal{O}_{C_p}(-2)$ has only $H^1$, we get an exact triangle
\begin{align}
\begin{split}
    &R^1f_*(\mathcal{O}_{C_p}(-1)\oplus \mathcal{O}_{C_p}(-2))[-1]\\
    \rightarrow& f_*(\mathcal{O}_{C_p}\oplus \mathcal{O}_{C_p}(-1)) \rightarrow f_*(\mathcal{O}_q\oplus \mathcal{O}_q(-1))\stackrel{+1}{\rightarrow}\,.
    \end{split}
\label{eqn:pointshift}
\end{align}
However, the first term is not a $\mathrm{Cl}^0(\Ad(p))$-module, only a shift of a $\mathrm{Cl}^0(\Ad(p))$-module in the derived category.  However, we can rotate (\ref{eqn:pointshift}) and arrive at a short exact sequence
\begin{align}
\begin{split}
    0&\rightarrow f_*(\mathcal{O}_{C_p}\oplus \mathcal{O}_{C_p}(-1)) \rightarrow f_*(\mathcal{O}_q\oplus \mathcal{O}_q(-1))\\
    &\rightarrow R^1f_*(\mathcal{O}_{C_p}(-1)\oplus \mathcal{O}_{C_p}(-2)) \rightarrow 0\,,
    \end{split}
\label{eqn:point}
\end{align}
of $\mathrm{Cl}^0(A(x))$-modules.  Since $H^0(C_p,\mathcal{O}_{C_p}\oplus \mathcal{O}_{C_p}(-1))$ and $H^1(C_p,\mathcal{O}_{C_p}(-1)\oplus \mathcal{O}_{C_p}(-2))$ are 1-dimensional, while $H^0(\mathcal{O}_q\oplus \mathcal{O}_q(-1))$ is 2-dimensional, we see that the three nontrivial representations of $\mathrm{Cl}^0(\Ad(p))$ in (\ref{eqn:point}) have dimensions 1,2,1 respectively.  We will identify these representations explicitly in Section~\ref{sec:Clifford}.

We can alternatively reexpress this calculation to arrive at (\ref{eqn:point}) directly.  We ``rotate" (\ref{eqn:P1ptses}) to obtain the exact triangle in $D^b(\widehat{\Xd})$

\begin{align}
\mathcal{O}_{C_p}\rightarrow \mathcal{O}_q\rightarrow \mathcal{O}_{C_p}(-1)[1]\stackrel{+1}{\rightarrow}\,,
\label{eqn:rotate}
\end{align}
which leads to the exact triangle in $D^b(\widehat{\Xd},\widehat{\mathcal{B}})$
\begin{align}
\mathcal{O}_{C_p}\oplus\mathcal{O}_{C_p}(-1) \rightarrow \mathcal{O}_q\oplus \mathcal{O}_q(-1) \rightarrow (\mathcal{O}_{C_p}(-1)\oplus\mathcal{O}_{C_p}(-2))[1]\stackrel{+1}{\rightarrow}\,.
\label{eqn:rotatedB}
\end{align}
Then the exact triangle~(\ref{eqn:point}) corresponds to~(\ref{eqn:rotatedB}) under the derived equivalence~(\ref{eq:derivedequiv}).

Since  $\mathcal{O}_{C_p}\oplus\mathcal{O}_{C_p}(-1)$ and $(\mathcal{O}_{C_p}(-1)\oplus\mathcal{O}_{C_p}(-2))[1]$ are non-isomorphic objects of $D^b(\widehat{\Xd},\widehat{\mathcal{B}})$, the 1-dimensional $\mathrm{Cl}^0(\Ad(p))$-modules $f_*(\mathcal{O}_{C_p}\oplus \mathcal{O}_{C_p}(-1))$ and $R^1f_*(\mathcal{O}_{C_p}(-1)\oplus \mathcal{O}_{C_p}(-2))$ are not isomorphic.  Furthermore (\ref{eqn:point}) does not split since (\ref{eqn:sesB}) does not split.  In particular, the 2-dimensional representation $f_*(\mathcal{O}_q\oplus \mathcal{O}_q(-1))$ of $\mathrm{Cl}^0(\Ad(p))$ is indecomposable.

If we now flop $C_p\subset\widehat{\Xd}$ to obtain another non-K\"ahler small resolution $\widehat{\Xd'}$ with exceptional curve $C_p'$, then under the derived equivalence 
\begin{align}
D^b(\widehat{\Xd})=D^b(\widehat{\Xd'})
\label{eqn:BridgelandFlop}
\end{align}
of \cite{Bridgeland}, the object $\mathcal{O}_{C_p'}$ of $D^b(\widehat{\Xd'})$ corresponds to the object $\mathcal{O}_{C_p}(-1)[1]$ of $D^b(\widehat{\Xd})$.
Referring back to (\ref{eqn:rotate}) and (\ref{eqn:point}), we see that the two 1-dimensional $\mathrm{Cl}^0(\Ad(p))$-modules described above correspond to the two small resolutions of the node $p$. 

Furthermore, given a point $q\in C_p'$, we let $E_q\in D^b(\widehat{\Xd})$ correspond to $\mathcal{O}_q\in D^b(\widehat{\Xd'})$ via (\ref{eqn:BridgelandFlop}).  Then $E_q$ fits into an exact triangle \cite{Bridgeland}
\begin{align}
\mathcal{O}_{C_p}(-1)[1]\rightarrow E_q\rightarrow \mathcal{O}_{C_p}\stackrel{+1}{\rightarrow}\,.
\end{align}
Repeating the previous argument, we conclude that $Rf_*(E_q\stackrel{L}{\otimes}(\mathcal{O}_{C_p}\oplus\mathcal{O}_{C_p}(-1)))$ is a two-dimensional representation of $\mathrm{Cl}^0(\Ad(p))$.  In perfect analogy to (\ref{eqn:point}), it is indecomposable and arises as an extension of two 1-dimensional representations: the exact same representations which appear in (\ref{eqn:point}) except that they are interchanged.

In conclusion, both small resolutions appear quite symmetrically in this description, just as they do in our proposal for torsion refined invariants.

\smallskip
For later use, we include a comment about the construction of $\widehat{\mathcal{B}}$ in \cite{Kuznetsov2013}.  The matrix $\Ad$ determines a bundle of quadric surfaces over $\mathbb{P}^3$.  Consider the relative Fano variety $F\to \mathbb{P}^3$ of lines in these quadric surfaces $Q(p)$ parametrized by $p\in \mathbb{P}^3$.  When $Q(p)$ has rank~4, i.e.\ when $p\not\in\Sd$, the fiber of the relative Fano variety over $p$ is the union of two $\mathbb{P}^1$s.  When $Q(p)$ has rank 3, i.e.\ when $p\in S-\SdOne$, the fiber over $p$ is a single $\mathbb{P}^1$. 
 These fibers fit together to form a $\mathbb{P}^1$ bundle over $\Xd-\SdOne$.  By running the relative minimal model program for $F\to\mathbb{P}^3$, Kuznetsov finds a birational transformation of $F$ (a flip) which transforms $F$ into a $\mathbb{P}^1$-bundle over $\widehat{\Xd}$.  The Azumaya algebra $\widehat{\mathcal{B}}$ is associated to this $\mathbb{P}^1$ bundle.   

Before turning to the $k>4$ case, we describe the representation theory of the Clifford algebras of interest in the next section.

\subsection{Clifford algebras and their representations}
\label{sec:Clifford}

In this section, we begin by reviewing and extending some standard results about  Clifford algebras, referring the reader to \cite{LM} for more details.  We are interested in representations of even Clifford algebras, but in our situation we can describe these representations in terms of representations of Clifford algebras by a straightforward extension of the well-known result for nondegenerate quadratic forms proven for example in \cite{LM}.

\begin{lemma}
    Let $V$ be a finite-dimensional vector space and $Q$ a quadratic form on $V$.  Let $\mathrm{Cl}(V,Q)$ denote the associated Clifford algebra.  Let $Q_1$ be a nondegenerate quadratic form on a 1-dimensional vector space $V_1$.  Then there is an isomorphism of $\mathbb{C}$-algebras $\mathrm{Cl}^0(V_1\oplus V,Q_1\oplus Q)\simeq \mathrm{Cl}(V,Q)$.
\label{lem:Cliffordisom}
\end{lemma}

\begin{proof}
    We choose $e\in V_1$ with $Q_1(e,e)=-1$ and define $\phi:V\to \mathrm{Cl}^0(V_1\oplus V,Q_1\oplus Q)$ by $\phi(v)=ev$.  Then for $v_1,v_2\in V$ we have
    \begin{equation}
        \phi(v_1)\phi(v_2)+\phi(v_2)\phi(v_1)=v_1v_2+v_2v_1=2Q(v_1,v_2)\,,
    \end{equation}
so by the universal property of $\mathrm{Cl}(V,Q)$, $\phi$ extends to a map $\mathrm{Cl}(V,Q)\rightarrow\mathrm{Cl}^0(V_1\oplus V,Q_1\oplus Q)$, which is readily checked to be an isomorphism upon choosing a basis for $V$.
\end{proof}

Clifford algebras are $\mathbb{Z}_2$-graded, and the $\mathbb{Z}_2$ grading of these Clifford algebras and their $\mathbb{Z}_2$-graded modules will simplify our analyses. Recall the notion of the $\mathbb{Z}_2$-graded tensor product of $\mathbb{Z}_2$-graded algebras $A$ and $B$, which we denote by $A\widehat\otimes B$.  As a vector space, it is the usual tensor product $A\otimes B$ with the grading
\begin{align}
\left(A\widehat\otimes B\right)^0=\left(A^0\otimes B^0\right)\oplus \left(A^1\otimes B^1\right)\qquad
\left(A\widehat\otimes B\right)^1=\left(A^0\otimes B^1\right)\oplus \left(A^1\otimes B^0\right)\,,
\end{align}
with a sign introduced in the product
\begin{align}
    (a_1\otimes b_1)(a_2\otimes b_2)=(-1)^{\deg(b_1)\deg(a_2)}(a_1a_2)\otimes (b_1b_2),\  a_1,a_2\in A,\ b_1,b_2\in B\,.
\end{align}

Similarly, if $M$ is a $\mathbb{Z}_2$-graded right $A$-module and $N$ is a $\mathbb{Z}_2$-graded right $B$-module, we can define the $\mathbb{Z}_2$-graded right $A\widehat\otimes B$-module $M\widehat\otimes N$ by imposing an analogous sign on the natural action of $A\widehat\otimes B$ on $M\widehat\otimes N$.
\begin{align}
    (m\otimes n)\cdot(a\otimes b)=(-1)^{\deg{n}\deg{a}}(m\cdot a)\otimes (n\cdot b)\,.
\end{align}
On occasion, we will omit ``$\mathbb{Z}_2$" from terminology and simply write graded tensor products or graded modules, when the $\mathbb{Z}_2$ is clear in context.

The utility of the $\mathbb{Z}_2$ grading arises from the following lemma.
\begin{lemma}
Let $(V_1,Q_1)$ and $(V_2,Q_2)$ be finite-dimensional vector spaces with quadratic forms.  Then
\begin{align}
\mathrm{Cl}(V_1\oplus V_2,Q_1\oplus Q_2)\simeq \mathrm{Cl}(V_1,Q_1)\widehat\otimes \mathrm{Cl}(V_2,Q_2)\,.
\end{align}
\label{lem:cliffsum}
\end{lemma}

We will also make use of a result about representations of Clifford algebras, generalizing a result in \cite{LM}.  
\begin{lemma}
Let $Q$ be a nonzero quadratic form on a finite-dimensional vector space $V$. Then there is an equivalence of categories between the category of right $\mathbb{Z}_2$-graded $\mathrm{Cl}(V,Q)$-modules and the category of (ungraded) right $\mathrm{Cl}^0(V,Q)$-modules, which associates to a $\mathbb{Z}_2$-graded $\mathrm{Cl}(V,Q)$-module $M$ the $\mathrm{Cl}^0(V,Q)$-module $M^0$, and to a $\mathrm{Cl}^0(V,Q)$-module $M^0$ the $\mathbb{Z}_2$-graded $\mathrm{Cl}(V,Q)$ module $M=M^0\otimes_{\mathrm{Cl}^0(V,Q)}\mathrm{Cl}(V,Q)$.
\label{lem:gradungrad}
\end{lemma}

Lemma~\ref{lem:gradungrad} is false for $Q=0$.  For example, we will construct a 1-dimensional representation $\Pi\textbf{1}$ of $\mathrm{Cl}(\mathbb{C},0)$ below with $(\Pi\textbf{1})^0=0$ and $(\Pi\textbf{1})^1=\mathbb{C}$, so that $(\Pi\textbf{1})^0\otimes_{\mathrm{Cl}^0(V,Q)}\mathrm{Cl}(V,Q)$ is 0 instead of $\Pi\textbf{1}$.

\begin{proof}
The proof of the lemma in the case of a nondegenerate quadratric form in \cite{LM} works in this more general situation. We just have to show that the functors are inverse to each other.
Suppose that $M^0$ is a $\mathrm{Cl}^0(V,Q)$-module.  Then
\begin{align}
    \left(M^0\otimes_{\mathrm{Cl}^0(V,Q)}\mathrm{Cl}(V,Q)\right)^0\simeq M^0\otimes_{\mathrm{Cl}^0(V,Q)}\mathrm{Cl}^0(V,Q)\simeq M^0\,.
\end{align}
In the other direction, let $M=M^0\oplus M^1$ be a $\mathrm{Z}_2$-graded $\mathrm{Cl}(V,Q)$ module.  We have to show that the natural map of $\mathrm{Z}_2$-graded $\mathrm{Cl}(V,Q)$ modules
\begin{align}
M^0\otimes_{\mathrm{Cl}^0(V,Q)}\mathrm{Cl}(V,Q)\rightarrow M,\qquad m\otimes \omega\mapsto m\cdot \omega
\end{align}
is an isomorphism, and for that we just have to check that the even and odd parts are isomorphisms.  The even part
\begin{align}
M^0\otimes_{\mathrm{Cl}^0(V,Q)}\mathrm{Cl}^0(V,Q)\rightarrow M^0
\end{align}
is clearly an isomorphism.  

For the odd part, let $v\in V$ be such that $Q(v)\ne0$.  Right multiplication by $v$ gives $\mathbb{C}$-linear maps $M^i\to M^{i+1}$ and $\mathrm{Cl}^i(V,Q)\to \mathrm{Cl}^{i+1}(V,Q)$
for $i\in \mathbb{Z}_2$.  Then the maps $M^0\to M^1$ and $M^1\to M^0$ are isomorphisms of vector spaces, since their compositions are multiplication by $Q(v)$.  Similarly, $\mathrm{Cl}^0(V,Q)\to \mathrm{Cl}^{1}(V,Q)$ is an isomorphism.  Now consider the commutative diagram of right $\mathrm{Cl}^0(V,Q)$-modules and $\mathrm{Cl}^0(V,Q)$-module homomorphisms

\begin{equation}
\begin{array}{ccc}
M^0\otimes_{\mathrm{Cl}^0(V,Q)}\mathrm{Cl}^0(V,Q)&\rightarrow& M^0\\
{1\otimes \cdot v}\downarrow \phantom{1\otimes \cdot v}&&\phantom{\cdot v}\downarrow \cdot v\\
M^0\otimes_{\mathrm{Cl}^0(V,Q)}\mathrm{Cl}^1(V,Q)&\rightarrow &M^1
\end{array}
\end{equation}
Since the vertical maps and the top map are isomorphisms, the bottom map is as well.
\end{proof}

Note that under this correspondence, the dimension of a $\mathbb{Z}_2$-graded $\mathrm{Cl}(V,Q)$-module is twice the dimension of the corresponding $\mathrm{Cl}^0(V,Q)$-module. 

\medskip
Our main interest is in sheaves of (right) $\mathcal{B}_0$-modules, or equivalently, of ${\mathcal{B}}$-modules.  In Section~\ref{sec:derived}, we attempt to generalize the $k=4$ case by suggesting a derived equivalence under which points of small resolutions of $\Xd$ correspond to $\mathcal{B}_0$-modules which are (scheme-theoretically) supported at a point $p\in\mathbb{P}^3$, with stalk of dimension $2^{n-1}$ (recall that $k$ is even and $n=k/2$).  In other words, we set out to classify $2^{n-1}$-dimensional $\mathcal{B}_0\otimes\mathcal{O}_p=\mathrm{Cl}^0(\Ad(p))$-modules, or equivalently $2^{n}$-dimensional $\mathbb{Z}_2$-graded $\mathrm{Cl}(\Ad(p))$-modules.

We denote the Clifford algebra of a maximal rank quadratic form on an $m$-dimensional vector space by $\mathrm{Cl}_m$.  As is well-known, the representation theory of $\mathrm{Cl}_m$ depends on the parity of $m$:

\medskip
\begin{tabular}{l}
     $\mathrm{Cl}_{2r}$ has a unique irreducible representation $S_{2r}$ of dimension $2^r$, \\[.2em]
     $\mathrm{Cl}_{2r+1}$ has two irreducible representations $S_{2r+1}',\ S_{2r+1}''$ of dimension $2^r$.
\end{tabular}
\medskip

By Lemma~\ref{lem:Cliffordisom} we have $\mathrm{Cl}_m^0\simeq \mathrm{Cl}_{m-1}$.  If $m=2r$ is even, then using this identification together with Lemma~\ref{lem:gradungrad}, we get two inequivalent minimal $\mathrm{Z}_2$-graded $\mathrm{Cl}_{2r}$-modules $\widehat{S}'_{2r}:=S'_{2r-1}\otimes_{\mathrm{Cl}^0_{2r}}\mathrm{Cl}_{2r}$ and $\widehat{S}''_{2r}:=S''_{2r-1}\otimes_{\mathrm{Cl}^0_{2r}}\mathrm{Cl}_{2r}$.  If $m=2r+1$ is odd, we get a unique minimal $\mathrm{Z}_2$-graded $\mathrm{Cl}_{2r+1}$-module $\widehat{S}_{2r+1}:=S_{2r}\otimes_{\mathrm{Cl}^0_{2r+1}}\mathrm{Cl}_{2r+1}$.  Thus
\medskip
\begin{tabular}{l}
     $\mathrm{Cl}_{2r}$ has two irreducible $\mathbb{Z}_2$-graded representations $\widehat{S}'_{2r},\widehat{S}''_{2r}$ of dimension $2^r$, \\[.2em]
     $\mathrm{Cl}_{2r+1}$ has a unique irreducible $\mathbb{Z}_2$-graded representation $\widehat{S}_{2r+1}$ of dimension $2^{r+1}$.
\end{tabular}
\medskip
Given a $\mathbb{Z}_2$-graded module $M=M^0\oplus M^1$ over any $\mathbb{Z}_2$-graded ring $A$, we define the $\mathbb{Z}_2$-graded $A$-module $\Pi M$ by switching the parities, i.e.
\begin{align}
    \left(\Pi M\right)^0=M^1,\qquad \left(\Pi M\right)^1=M^0\,.
\end{align}
For later use, we note here that $\Pi\widehat{S}_{2r+1}\simeq \widehat{S}_{2r+1}$, since there is only one graded representation of dimension $2^{r+1}$.  Also, $\widehat{S}''_{2r}\simeq\Pi\widehat{S}'_{2r}$ as $\mathbb{Z}_2$-graded modules.  This can be checked directly for $r=1$, and then the statement follows inductively using $\mathrm{Cl}_{2r}\simeq \mathrm{Cl}_{2}\widehat{\otimes}\mathrm{Cl}_{2r-2}$.

We can now describe the $2^{n-1}$-dimensional $\mathcal{B}_0\otimes\mathcal{O}_p=\mathrm{Cl}^0(\Ad(p))$-modules in terms of their associated $\mathbb{Z}_2$-graded $2^n$-dimensional $\mathrm{Cl}(\Ad(p))$-modules.

\smallskip
If $\rank(\Ad(p))=k$, i.e.\ if $p\not\in \Sd$, then $\mathrm{Cl}(\Ad(p))\simeq\mathrm{Cl}_k$.
We then have the $2^n$-dimensional graded modules $\widehat{S}'_k$ and $\widehat{S}''_k$ which we identify with the two points of $\pi^{-1}(p)$ in $\Xd$.   We will see this identification and similar identifications below more explicitly in the $k=4$ case.

\smallskip
If $\rank(\Ad(p))=k-1$, i.e.\ if $p$ is a smooth point of $\Sd$, by Lemma~\ref{lem:cliffsum} we have $\mathrm{Cl}(\Ad(p))\simeq \mathrm{Cl}_{k-1}\widehat\otimes\mathrm{Cl}(V_1,0)$ with $\dim V_1=1$.
Given a $\mathbb{Z}_2$-graded $\mathrm{Cl}_{k-1}$-module $M$ and a $\mathbb{Z}_2$-graded $\mathrm{Cl}(V_1,0)$-module $N$, then $M\widehat\otimes N$ is a $\mathbb{Z}_2$-graded $ \mathrm{Cl}_{k-1}\widehat\otimes\mathrm{Cl}(V_1,0)$-module.

Now  $\mathrm{Cl}_{k-1}$ has the irreducible $\mathbb{Z}_2$-graded representation $M=\widehat{S}_{k-1}$,  which has dimension $2^n$, so we have to take $N$ to be 1-dimensional in order to get a $2^n$-dimensional graded representation of $\mathrm{Cl}(\Ad(p))$.  Since $V_1$ is nilpotent in $\mathrm{Cl}(V_1,0)$, it must act trivially on $N$, which determines the entire $\mathrm{Cl}(V_1,0)$-module structure on $N$. We give $N$ the structure of a $\mathbb{Z}_2$-graded module by assigning it a parity.  Let $\textbf{1}$ be this 1-dimensional representation with even parity.  Then the other graded $\mathrm{Cl}(V_1,0)$-module is $\Pi\textbf{1}$.  We then get $\mathbb{Z}_2$-graded $\mathrm{Cl}_{k-1}\widehat\otimes\mathrm{Cl}(V_1,0)$ modules $\widehat{S}_{k-1}\widehat\otimes \textbf{1}$ and $\widehat{S}_{k-1}\widehat\otimes (\Pi\textbf{1})$.  We show that these are isomorphic as follows.  The graded $\mathrm{Cl}_{k-1}\widehat\otimes\mathrm{Cl}(V_1,0)$-module structures on $\widehat{S}_{k-1}\widehat\otimes \textbf{1}$ and $\widehat{S}_{k-1}\widehat\otimes (\Pi\textbf{1})$ are completely determined by the action of the subring $\mathrm{Cl}_{k-1}\widehat\otimes\mathbb{C}\simeq \mathrm{Cl}_{k-1}\otimes\mathbb{C}\simeq \mathrm{Cl}_{k-1}$, as $V_1$ acts trivially.  As $\mathrm{Cl}_{k-1}$-modules, they are both irreducible, hence both are isomorphic to the unique irreducible $\mathbb{Z}_2$-graded module $\widehat{S}_{k-1}$:
\begin{align}
\widehat{S}_{k-1}\widehat\otimes \textbf{1}\simeq \widehat{S}_{k-1}\simeq \widehat{S}_{k-1}\widehat\otimes \left(\Pi\textbf{1}\right)\,.
\end{align}
So there is only one graded representation of $\mathrm{Cl}_{k-1}\widehat\otimes\mathrm{Cl}(V_1,0)$, which 
 we identify with the unique point of $\pi^{-1}(p)$. 

\smallskip
Finally, we assume $k\ge4$  and consider the case $\rank(\Ad(p))=k-2$, i.e.\ $p$ is a node of $\Sd$ and $\Xd$.  We have $\mathrm{Cl}(\Ad(p))\simeq \mathrm{Cl}_{k-2}\widehat\otimes \mathrm{Cl}(V_2,0)$ for a 2-dimensional vector space $V_2$.  As with $V_1$, we have 1-dimensional representations of $\mathrm{Cl}(V_2,0)$ on which $V_2$ acts trivially.  Assigning this representation even parity, we again call this representation $\textbf{1}$.  Then we have $2^{n-1}$-dimensional graded representations
\begin{align}
    \widehat{S}'_{k-2}\widehat\otimes\textbf{1},\ \widehat{S}'_{k-2}\widehat\otimes\left(\Pi\textbf{1}\right),\ \widehat{S}''_{k-2}\widehat\otimes\textbf{1},\ \widehat{S}''_{k-2}\widehat\otimes\left(\Pi\textbf{1}\right)\,.
\end{align}

We claim
\begin{align}
 \widehat{S}'_{k-2}\widehat\otimes\textbf{1}\simeq \widehat{S}''_{k-2}\widehat\otimes\left(\Pi\textbf{1}\right) ,\qquad \widehat{S}'_{k-2}\widehat\otimes\left(\Pi\textbf{1}\right)\simeq\widehat{S}''_{k-2}\widehat\otimes\textbf{1}\,,
 \label{eqn:gradedisos}
\end{align}
so there are only two inequivalent representations among these. It suffices to show that these are isomorphisms of graded $\mathrm{Cl}_{k-2}$-modules. We only explain the first isomorphism as the second is similar.    We have $\widehat{S}'_{k-2}\widehat\otimes\textbf{1}\simeq \widehat{S}'_{k-2}$ as $\mathrm{Cl}_{k-2}$-modules. 
But we can only see that $\widehat{S}''_{k-2}\widehat\otimes\left(\Pi\textbf{1}\right)\simeq \Pi\widehat{S}''_{k-2}$ as $\mathbb{Z}_2$-graded vector spaces without a further argument, as the actions of $a\in\mathrm{Cl}_{k-2}$ on these two $\mathbb{Z}_2$-graded vector spaces differ by the sign $(-1)^{\deg(a)}$ due to the definition of the module structure on the $\mathbb{Z}_2$-graded tensor product.  Now let $e_1,\ldots,e_{k-2}$ be an orthonormal basis for $\mathbb{C}^{k-2}$ and put $\epsilon_{k-2}=e_1\cdots e_{k-2}$. Since $k-2$ is even and the $e_i$ anticommute with $\epsilon_{k-2}$, we see that multiplication by $\epsilon_{k-2}$ defines an even linear map
\begin{align}
     \Pi\widehat{S}''_{k-2} \rightarrow \Pi\widehat{S}''_{k-2}\,,
\end{align}
which is an isomorphism of $\mathbb{Z}_2$-graded $\mathrm{Cl}_{k-2}$-modules if one module has the standard $\mathrm{Cl}_{k-2}$-module structure while in the other module structure, multiplication by $a\in\mathrm{Cl}_{k-2}$ is modified by the sign $(-1)^{\deg(a)}$ due to the above-mentioned anticommutativity.  Thus, we have isomorphisms of $\mathbb{Z}_2$-graded $\mathrm{Cl}_{k-2}$-modules
\begin{align}
   \widehat{S}''_{k-2}\widehat\otimes\left(\Pi\textbf{1}\right)\simeq \Pi\widehat{S}''_{k-2}\simeq  \widehat{S}'_{k-2}\,,
   \label{eqn:parityisos}
\end{align}
and $\widehat{S}'_{k-2}\widehat\otimes\textbf{1}\simeq \widehat{S}''_{k-2}\widehat\otimes\left(\Pi\textbf{1}\right)$ as claimed.

We now turn to our main interest, $2^n$-dimensional graded representations of $\mathrm{Cl}_{k-2}\widehat\otimes\mathrm{Cl}(V_2,0)$.  Since the irreducible graded representations of $\mathrm{Cl}_{k-2}$ have dimension $2^{n-1}$, we seek $2$-dimensional graded representations of $\mathrm{Cl}(V_2,0)$ to tensor with.  While we can take direct sums of the 1-dimensional graded representations $\textbf{1}$ and $\Pi\textbf{1}$, the more interesting representations are 2-dimensional indecomposable graded representations $P=P^0\oplus P^1$.  The action of $v\in V_2$ on $P$ is given by linear maps
\begin{align}
P^0\rightarrow P^1,\qquad P^1\rightarrow P^0\,,
\end{align}
either of which we denote by $r_v$.  For $v_1,v_2\in V_2$, $r_{v_1}$ and $r_{v_2}$ anticommute since $v_1,v_2$ anticommute in $\mathrm{Cl}(V_2,0)$.  In particular $r_v^2=0$.  If all $r_v=0$, then it is easy to see that $P$ is a direct sum of 1-dimensional representations.  So we assume that $r_v\ne 0$ for some $v$.  It is straightforward to check that $r_w$ must be a scalar multiple of $r_v$ for any $w\in V_2$. 
Therefore the linear map
\begin{align}
    V_2\rightarrow \mathrm{End}(P),\ v\mapsto r_v\,,
\end{align}
has 1-dimensional image, hence 1-dimensional kernel $U\subset V_2$.  Furthermore, it follows that the action of $\Lambda^2V_2\subset\Lambda^*V_2\simeq\mathrm{Cl}(V_2,0)$ on $P$ is trivial.  These facts determine $P$ completely as a $\mathrm{Cl}(V_2,0)$-module, and we denoted it by $P_U$.  If desired, we can describe $P_U$ explicitly as the ideal
\begin{equation}
    P_U=U\oplus\Lambda^2V_2\subset \mathrm{Cl}(V_2,0)\,.
\end{equation}
We can and will give $P_U$ the structure of a graded module by putting $P^0=\Lambda^2V_2$ and $P^1=U$.

We note the short exact sequence of graded $\Lambda^*V_2$-modules
\begin{align}
    0\rightarrow \Lambda^2V_2 \rightarrow P_U \rightarrow P_U \big\slash \Lambda^2V_2 \rightarrow 0\,,
\end{align}
which is immediately identified with 
\begin{align}
    0\rightarrow \textbf{1} \rightarrow P_U \rightarrow \Pi\textbf{1} \rightarrow 0\,.
\label{eqn:ClV2ext}
\end{align}

We can now write down all possible indecomposable $2^n$-dimensional graded $\mathrm{Cl}_{k-2}\widehat\otimes\mathrm{Cl}(V_2,0)$ modules:
\begin{align}
       \widehat{S}'_{k-2}\widehat\otimes P_U,\ \widehat{S}'_{k-2}\widehat\otimes\left(\Pi P_U\right),\ \widehat{S}''_{k-2}\widehat\otimes P_U,\ \widehat{S}''_{k-2}\widehat\otimes\left(\Pi P_U\right)\,,
\end{align}
and we claim that
\begin{align}
       \widehat{S}'_{k-2}\widehat\otimes P_U\simeq \widehat{S}''_{k-2}\widehat\otimes\left(\Pi P_U\right),\qquad \widehat{S}'_{k-2}\widehat\otimes\left(\Pi P_U\right)\simeq \widehat{S}''_{k-2}\widehat\otimes P_U\ \,.
\label{eqn:parityisos2}
\end{align}
Thus, we have two families $\widehat{S}'_{k-2}\widehat\otimes P_U$ and $\widehat{S}''_{k-2}\widehat\otimes P_U$ of $2^n$-dimensional graded representations of $\mathrm{Cl}_{k-2}\widehat\otimes\mathrm{Cl}(V_2,0)$-modules, which we identify with the points of both small resolutions of $p$.  They fit into short exact sequences
\begin{align}
    0\rightarrow \widehat{S}'_{k-2}\widehat\otimes \textbf{1} \rightarrow \widehat{S}'_{k-2}\widehat\otimes P_U \rightarrow \widehat{S}''_{k-2}\widehat\otimes \textbf{1} \rightarrow 0\,,
\label{eqn:ext1}
\end{align}
and 
\begin{align}
    0\rightarrow \widehat{S}''_{k-2}\widehat\otimes \textbf{1} \rightarrow \widehat{S}''_{k-2}\widehat\otimes P_U \rightarrow \widehat{S}'_{k-2}\widehat\otimes \textbf{1} \rightarrow 0\,,
    \label{eqn:ext2}
\end{align}
where (\ref{eqn:ext1}) and (\ref{eqn:ext2}) are respectively obtained from (\ref{eqn:ClV2ext}) by graded tensor product with $\widehat{S}'_{k-2}$ and $\widehat{S}'_{k-2}$, then employing the relevant isomorphism from (\ref{eqn:gradedisos}).

The proof of (\ref{eqn:parityisos2}) is a straightforward adaption of the proof of (\ref{eqn:parityisos}). Since $P_U$ contains both even and odd summands, the graded tensor product introduces signs in the Clifford action on both sides of the above isomorphisms.  However, they can be compensated for as before using isomorphisms $\epsilon_{k-2}\otimes 1$.

At times, it will be convenient to think of $U\subset V_2$ as a point $q=[U]\in \mathbb{P}(V_2)\simeq\mathbb{P}^1$, so we sometimes write $P_q$ instead of $P_U$.  

\subsection{Derived equivalences}
\label{sec:derived}

We begin this section with a conjecture. Conjecture~\ref{conj:hatb0} sets up a derived equivalence which allows us to relate our torsion refined invariants to the purely algebraic setting of sheaves of $\mathcal{B}_0$-modules over $\mathbb{P}^3$, simultaneously handling all small resolutions of $\Xd$. This conjecture is motivated by the $k=4$ case where it is known to hold \cite{Kuznetsov2013}, as we summarized in Section~\ref{sec:k4}.
At the end of this section, we point out some facts which are consistent with Conjecture~\ref{conj:hatb0}.

\begin{conjecture}
There exists an Azumaya algebra  $\widehat{\mathcal{B}}$ on $\widehat{X}$ extending ${\mathcal{B}}\vert_{\Xd-\SdOne}$ with $f_*(\widehat{\mathcal{B}})={\mathcal{B}}_0$, such that (\ref{eq:derivedequiv}) and (\ref{eq:inversederivedequiv}) are inverse derived equivalences.  Furthermore, for each exceptional curve $C_p$ we have an isomorphism of sheaves of algebras
\begin{align}
    \widehat{\mathcal{B}}\vert_{C_p}\simeq \underline{End}\left(\left(\mathcal{O}_{C_p}\oplus \mathcal{O}_{C_p}(-1)\right)^{\oplus 2^{n-2}}\right)\,.
\label{eqn:bcp}
\end{align}
\label{conj:hatb0}
\end{conjecture}

Assuming that this conjecture holds, there is a 1-1 correspondence between sheaves of $\mathcal{O}_{C_p}$-modules and $\widehat{\mathcal{B}}$-modules supported on $C_p$ given by 
\begin{align}
    F\mapsto F\otimes_{\mathcal{O}_C}\left(\left(\mathcal{O}_{C_p}\oplus \mathcal{O}_{C_p}(-1)\right)^{\oplus 2^{n-2}}\right)\,.
\label{eqn:CtoBgeneral}
\end{align}
Under this correspondence, $\mathcal{O}_{C_p}$ corresponds to $(\mathcal{O}_{C_p}\oplus \mathcal{O}_{C_p}(-1))^{\oplus 2^{n-2}}$.

Some comments are in order here.  Since $\mathcal{B}_0$ is a locally free sheaf of rank $2^{k-1}$ and $\widehat{\Xd}\to \mathbb{P}^3$ is generically a 2-1 cover, it follows that $\widehat{\mathcal{B}}$ must be locally free of rank $2^{k-2}$.  If $q\in \widehat{\Xd}$, then $\widehat{\mathcal{B}}\otimes \mathcal{O}_q$ is isomorphic to a $2^{k-2}$-dimensional matrix algebra, necessarily the algebra $M_{2^{n-1}}$ of $2^{n-1}\times 2^{n-1}$ matrices.  The algebra $M_{2^{n-1}}$ has a unique irreducible representation of dimension $2^{n-1}$, identified with a $\widehat{\mathcal{B}}$-module skyscraper sheaf of dimension $2^{n-1}$ supported on $q$.  By the derived equivalence (\ref{eq:derivedequiv}), these correspond with ${\mathcal{B}}_0$-module skyscraper sheaves of dimension $2^{n-1}$ supported on $p=f(q)$, or equivalently, by $2^{n-1}$-dimensional $\mathrm{Cl}^0(\Ad(p))$-modules, or equivalently again by $2^{n}$-dimensional graded $\mathrm{Cl}(\Ad(p))$-modules.  Note further that for any $q\in C_p$ we have
\begin{align}
\begin{split}
    \underline{End}\left(\left(\mathcal{O}_{C_p}\oplus \mathcal{O}_{C_p}(-1)\right)^{\oplus 2^{n-2}}\right)\otimes\mathcal{O}_q\,,
\end{split}
\end{align}
is a skyscraper sheaf supported at $q$ with fiber $M_{2^{n-1}}$, consistent with (\ref{eqn:bcp}).

We can now easily modify the identifications of the exceptional curves and their points described in Section~\ref{sec:k4} so that it applies in the general case.

We first observe that via the derived equivalence ({\ref{eq:derivedequiv}), $(\mathcal{O}_{C_p}\oplus \mathcal{O}_{C_p}(-1))^{\oplus{2^{n-2}}}$ is the $\widehat{\mathcal{B}}$-module associated to the graded representation $\hat{S}'_{k-2}\widehat\otimes \textbf{1}$ of $\mathrm{Cl}(\Ad(p))$ described in Section~\ref{sec:Clifford}.  To see this, we simply observe that $Rf_*((\mathcal{O}_{C_p}\oplus \mathcal{O}_{C_p}(-1))^{\oplus{2^{n-2}}})$ is a $2^{n-2}$-dimensional skyscraper sheaf supported at $p$. We saw that there are only two $2^{n-2}$-dimensional representations of $\mathrm{Cl}^0(\Ad(p))$, corresponding to the $2^{n-1}$-dimensional graded representations $\widehat{S}'_{k-2}\widehat\otimes \textbf{1}$ and $\widehat{S}''_{k-2}\widehat\otimes \textbf{1}$ of $\mathrm{Cl}(\Ad(p))$.  The conventional choice $\widehat{S}'_{k-2}\widehat\otimes \textbf{1}$ is determined by the choice of small resolution.

Similarly, $(\mathcal{O}_{C_p}(-1)\oplus\mathcal{O}_{C_p}(-2))^{\oplus2^{n-2}}[1]$ is the object of $D^b(\widehat{\Xd},\widehat{\mathcal{B}})$  associated to the other graded representation $\widehat{S}''_{k-2}\widehat\otimes \textbf{1}$ of $\mathrm{Cl}(\Ad(p))$ described in Section~\ref{sec:Clifford}.  We simply observe that $Rf_*((\mathcal{O}_{C_p}(-1)\oplus\mathcal{O}_{C_p}(-2))^{\oplus2^{n-2}}[1])$ is a skyscraper sheaf of dimension $2^{n-2}$ supported at $p$. Extending the $k=4$ case, we also note that $(\mathcal{O}_{C_p}(-1)\oplus\mathcal{O}_{C_p}(-2))^{\oplus2^{n-2}}[1]$ is associated to $\mathcal{O}_{C_p}(-1)[1]$ by (\ref{eqn:CtoBgeneral}), and as we have seen before, $\mathcal{O}_{C_p}(-1)[1]$ corresponds to $\mathcal{O}_{C_p'}$ via (\ref{eqn:BridgelandFlop}), where $C_p'$ is the flopped curve.  

We next claim that the graded representations $\widehat{S}'_{k-2}\widehat\otimes P_q$ correspond to the points $q\in C_p$.  To see this, we start with (\ref{eqn:P1ptses}) and apply (\ref{eqn:CtoBgeneral}) to get 
\begin{align}
\begin{split}
  0\rightarrow\left(\mathcal{O}_{C_p}(-1)\oplus\mathcal{O}_{C_p}(-2)\right)^{\oplus{2^{n-2}}}\rightarrow\left(\mathcal{O}_{C_p}\oplus\mathcal{O}_{C_p}(-1)\right)^{\oplus{2^{n-2}}}\\ \rightarrow\left(\mathcal{O}_q\oplus\mathcal{O}_q(-1)\right)^{\oplus{2^{n-2}}}\rightarrow 0  \,,
\end{split}
\end{align}
which rotates to the exact triangle
\begin{align}
\begin{split}
  \left(\mathcal{O}_{C_p}\oplus\mathcal{O}_{C_p}(-1)\right)^{\oplus{2^{n-2}}}\rightarrow\left(\mathcal{O}_q\oplus\mathcal{O}_q(-1)\right)^{\oplus{2^{n-2}}}\\ \rightarrow \left(\mathcal{O}_{C_p}(-1)\oplus\mathcal{O}_{C_p}(-2)\right)^{\oplus{2^{n-2}}}[1]\stackrel{+1}{\rightarrow}\,.
\end{split}
\end{align}
The claim follows immediately by applying $Rf_*$, using the preceding discussion to identify the first and last terms with the graded representations $\widehat{S}'_{k-2}\widehat{\otimes}\textbf{1}$ and $\widehat{S}''_{k-2}\widehat{\otimes}\textbf{1}$ respectively, and comparing to (\ref{eqn:ext1}).  This argument also completes our discussion in the $k=4$ case.

Recall~\cite{Bridgeland} that points $q'$ of the flopped $\mathbb{P}^1$ correspond to objects $E_{q'}$ of $D^b(\widehat{\Xd})$ fitting into an exact triangle which we write as
\begin{align}
    0 \rightarrow \mathcal{O}_{C_p}(-1)[1] \rightarrow E_{q'} \rightarrow \mathcal{O}_{C_p}\rightarrow 0\,.
\end{align}
Tensoring with $(\mathcal{O}_{C_p}\oplus \mathcal{O}_{C_p}(-1))^{\oplus 2^{n-2}}$, applying $Rf_*$,we similarly conclude by comparison with (\ref{eqn:ext2}) that the points $q'$ of the flopped $\mathbb{P}^1$ correspond to the graded representation $\widehat{S}''_{k-2}\widehat\otimes P_q$ as claimed.

Once again, both small resolutions appear quite symmetrically in this description, just as they do in our proposal for torsion refined invariants.

We also point out that even without assuming Conjecture~\ref{conj:hatb0}, the discussion in Section~\ref{sec:Clifford} still shows the existence of a 1-1 correspondence between $2^{n-1}$-dimensional representations of $\mathrm{Cl}^0(\Ad(p))$ and points of small resolutions lying over $p\in \mathbb{P}^3$.  In particular, if $p\in\SdOne$ is a node, then these representations of $\mathrm{Cl}^0(\Ad(p))$ are the representations $(\widehat{S}'_{k-2}\widehat\otimes P_U)^0$ and $(\widehat{S}''_{k-2}\widehat\otimes P_U)^0$, which are each parametrized by $\mathbb{P}^1$.  While it is tempting to identify these $\mathbb{P}^1$s with exceptional $\mathbb{P}^1$s of some small resolution, we don't have a precise way to make this identification in the absence of an explicit derived equivalence.

\smallskip
Continuing to assume Conjecture~\ref{conj:hatb0}, we can now describe the 2-torsion in $H_2(\widehat{\Xd},\mathbb{Z})$ in terms of sheaves of $\widehat{\mathcal{B}}$-modules.  
We have $D^b(\mathbb{P}^3,\mathcal{B}_0)= D^b(\widehat{\Xd},\alpha)$, with a twist by an element $\alpha$ of the Brauer group of $\widehat{\Xd}$.  Then $\alpha$ also defines a twist of K-theory.  We have the subspace $K_{(2)}^\alpha(\widehat{\Xd})$ of $K^\alpha(\widehat{\Xd})$ generated by D2-D0 branes \cite{Brunner:2001eg}.
The Atiyah-Hirzebruch spectral sequence gives a short exact sequence
\begin{align}
    0 \rightarrow H_0(\widehat{\Xd},\mathbb{Z}) \rightarrow K_{(2)}^\alpha(\widehat{\Xd}) \rightarrow H_2(\widehat{\Xd},\mathbb{Z}) \rightarrow 0\,,
\label{eqn:abut}
\end{align}
where we have used the Poincar\'e duality isomorphisms $H_j(\widehat{\Xd},\mathbb{Z})\simeq H^{6-j}(\widehat{\Xd},\mathbb{Z})$.   
Our model of $\alpha$-twisted K-theory will be in terms of sheaves of $\widehat{\mathcal{B}}$-modules.  

Let $\mathcal{E}$ be the rank $2^{n-1}$ $\alpha$-twisted sheaf on $\widehat{\Xd}$ determined by the cohomology class in $H^1(\widehat{\Xd},\operatorname{PGL}(2^{n-1}))$ associated with the Azumaya algebra $\widehat{\mathcal{B}}$, so that $\mathcal{E}\vert_{C_p}\simeq (\mathcal{O}_{C_p}\oplus \mathcal{O}_{C_p}(-1))^{\oplus{2^{n-2}}}$. 

There are natural maps
\begin{align}
  H_0(\widehat{\Xd},\mathbb{Z}) \rightarrow K_{(2)}^\alpha(\widehat{\Xd}),\qquad [q]\mapsto   [\mathcal{E}\otimes\mathcal{O}_q]\,,
  \label{eqn:firstmap}
\end{align}
and
\begin{align}
K_{(2)}^\alpha(\widehat{\Xd}) \rightarrow H_2(\widehat{\Xd},\mathbb{Z}),\qquad [F]\mapsto\mathrm{ch}_2(F)/2^{n-1}\,.
\label{eqn:secondmap}
\end{align}
We say a few words explaining why (\ref{eqn:secondmap}) takes values in integral homology.  For a twisted sheaf $F$, let $q$ be an element of the support of $F$.  Since $\widehat{\mathcal{B}}\otimes \mathcal{O}_q\simeq M_{2^{n-1}}$, and $M_{2^{n-1}}$ has a unique irreducible representation of dimension $2^{n-1}$, it follows that the dimension of $F\otimes\mathcal{O}_q$ must be a multiple of $2^{n-1}$.  Therefore the multiplicity of each irreducible component of the support cycle of $F$ (taking multiplicities and ranks into consideration) naturally representing $\mathrm{ch}_2(F)$ is a multiple of $2^{n-1}$ as well.  So the target of the map (\ref{eqn:secondmap}) takes values in integral homology, as claimed.

In the special case of a twisted sheaf $G\otimes\mathcal{E}$ associated to an ordinary sheaf $G$ on $\widehat{\Xd}$ supported on a curve, the image of $G\otimes\mathcal{E}$ under (\ref{eqn:secondmap}) is just $\mathrm{ch}_2(G)$, which is represented by the integral support cycle of $G$.

We let $\delta\in K_{(2)}^\alpha(\widehat{\Xd})$ denote the class of $(\mathcal{O}_{C_p}\oplus \mathcal{O}_{C_p}(-1))^{\oplus{2^{n-2}}}$.  Since $\mathrm{ch}_2(\delta)=2^{n-1}[C_p]$, we see that $\delta$ maps via ({\ref{eqn:abut}) to $[C_p]$, the nontrivial 2-torsion element of $H_2(\widehat{\Xd},\mathbb{Z})$.  

Next, we let $\gamma\in K_{(2)}^\alpha(\widehat{\Xd})$ denote the class of
$\mathcal{E}\otimes\mathcal{O}_q$, which corresponds to the class of a point via (\ref{eqn:firstmap}).  We claim that $2\delta = \gamma$. To see this, we first note that $2\delta$ maps to $2[C_p]=0$ via (\ref{eqn:secondmap}), so must be of the form $m\gamma$ for some $m\in\mathbb{Z}$.  Applying $Rf_*$ to $2\delta=m\gamma$ and noting that $Rf_*(\mathcal{E}\vert_{C_p})$ is a skyscraper sheaf of dimension $2^{n-2}$ while $Rf_*(\mathcal{E}\vert_q)$ is a skyscraper sheaf of dimension $2^{n-1}$, we conclude that $m=1$.  

We note in passing that if one could show directly that $\gamma=2\delta$, then the image of $\gamma$ in $H_2(\widehat{\Xd},\mathbb{Z})$ would be 2-torsion and we would have an independent proof that $H_2(\widehat{\Xd},\mathbb{Z})$ has 2-torsion without appealing to degenerate octic double solids as in Section~\ref{sec:verysingular}.

Finally, we show that $K_{(2)}^\alpha(\widehat{\Xd})\simeq\mathbb{Z}^2$.  Since $H_0(\widehat{\Xd},\mathbb{Z})\simeq\mathbb{Z}$ and $H_2(\widehat{\Xd},\mathbb{Z})\simeq\mathbb{Z}\oplus\mathbb{Z}_2$, we only need to show that (\ref{eqn:abut}) does not split, and for that we only have to show that $[C_p]$ is not in the image of a 2-torsion element of $K_{(2)}^\alpha(\widehat{\Xd})$. The most general K-theory class mapping to $[C_p]$ is $\gamma+m\delta$ with $m\in\mathbb{Z}$, so we only have to show that $2\delta+2m\alpha$ is nonzero for any $m$.  That can be confirmed by applying $Rf_*$.  We get a virtual $\mathcal{B}_0$-module, a sum or difference of skyscraper sheaves of total dimension $2(2^{n-2})+2m(2^{n-1})=2^{n-1}(2m+1)$, which is nonzero.

\smallskip
We now make a new proposal for the torsion refined invariants in terms of sheaves of ${\mathcal{B}}_0$-modules.  Let $\beta\in K_{(2)}^\alpha(\widehat{\Xd})$.  Using the derived equivalence $D^b(\mathbb{P}^3,\mathcal{B}_0)=D^b(\widehat{\Xd},\widehat{\mathcal{B}})$, we consider the moduli space $M_\beta$ of $\mathcal{B}_0$-modules with class $\beta$.  Since sheaves of  $\mathcal{B}_0$-modules can be viewed as ordinary sheaves of $\mathcal{O}_{\mathbb{P}^3}$-modules, we have a map $M_\beta\to \mathrm{Chow}$ from $M_\beta$ to the ordinary Chow variety of $\mathbb{P}^3$, which plays the same role as the map from the moduli space of D2-D0 branes to its support curve in \cite{Gopakumar:1998jq}.  We can then apply the same mathematical definitions which have been used to define the ordinary Gopakumar-Vafa invariants in this situation \cite{kiem2008GV,Maulik2018GV}. The result are GV invariants $n_g^\beta$.

Note that the mathematical definition of the $g=0$ GV invariants do not depend on the D0-brane charge for ordinary D2-D0 branes, \cite{Joyce2012}, and this independence is conjectured for all $g$ \cite{Toda2017}.  It seems reasonable to assume that something similar happens here, so that $n_g^\beta=n_g^{\beta+m\gamma}$ for any $m\in\mathbb{Z}$.  The conclusion would be that the $n_b^\beta$ only depend on the image of $\beta$ in $H_2(\widehat{\Xd},\mathbb{Z})$ and we have recovered the notion of torsion refined invariants!

There is clearly a lot to check here, which we leave to future work.  Issues include showing that $M_\beta$ is a good moduli space, that it supports a perverse sheaf of vanishing cycles, perhaps associated to a $\mathcal{B}_0$-module analogue of holomorphic Chern-Simons theory, and verifying $n_g^\beta=n_g^{\beta+m\gamma}$.

\smallskip
We conclude this section with a few comments about the geometry in the $k=6$ and $k=8$ cases, observing (weak) consistency with 
Conjecture~\ref{conj:hatb0}.

\paragraph{Case $k=6$} 
We start to adapt the discussion at the end of Section~\ref{sec:k4}.  
The matrix $\Ad$ determines a bundle of quadric fourfolds over $\mathbb{P}^5$.  Consider the relative Fano variety $F\to \mathbb{P}^3$ of $\mathbb{P}^2$s in these quadric fourfolds $Q(p)$ parametrized by $p\in \mathbb{P}^3$.  When $Q(p)$ has rank~6, i.e.\ when $p\not\in\Sd$, the fiber of the relative Fano variety over $p$ is the union of two $\mathbb{P}^3$s \cite{Harris}. When $Q(p)$ has rank 5, i.e.\ when $p\in S-\SdOne$, the fiber over $p$ is a single $\mathbb{P}^3$. 
 These fibers fit together to form a $\mathbb{P}^3$ bundle over $\Xd-\SdOne$.  To this $\mathbb{P}^3$-bundle is associated a rank~16 sheaf of Azumaya algebras on $\Xd-\SdOne$.  Perhaps this Azumaya algebra extends to a sheaf $\widehat{B}$ of Azumaya algebras on all of $\widehat{\Xd}$ which satisfies the conditions of Conjecture~\ref{conj:hatb0}. As a weak check, since $f:\widehat{X}\to\mathbb{P}^3$ is generically 2 to 1, $f_*(\widehat{\mathcal{B}})$ has rank 32, which is indeed the rank of $\mathcal{B}_0$.

\paragraph{Case $k=8$}
It was shown in~\cite{Addington:2012zv} that $D^b(\widehat{\Xd},\widehat{\mathcal{B}})\simeq D^b(Y)$, where $\widehat{\mathcal{B}}$ generates the Brauer group of $\widehat{\Xd}$ and $Y$ is a complete intersection of four generic quadrics in $\mathbb{P}^7$.
The latter is homologically projective dual to $D^b(\mathbb{P}^3,\mathcal{B}_0)$~\cite{Kuznetsov2008} and it follows that
\begin{align}
    D^b(\widehat{\Xd},\widehat{\mathcal{B}})\simeq D^b(Y)\simeq D^b(\mathbb{P}^3,\mathcal{B}_0)\,,
\end{align}
consistent with Conjecture~\ref{conj:hatb0}.

\section{Non-commutative resolutions from GLSMs}
\label{sec:glsm}
In this section we are going to construct and study gauged linear sigma models (GLSMs)~\cite{Witten:1993yc} that exhibit phases in which the theories flow to non-linear sigma models (NLSMs) on non-commutative resolutions of determinantal octic double solids.

First, in Section~\ref{sec:glsmModels}, we construct the GLSMs with hybrid phases associated to the non-commutative resolutions $\Xdnc$ of the determinantal octic double solids $\Xd$.
We then discuss the transitions $\widehat{\Xdr}\rightarrow \Xdnc$ in Section~\ref{sec:conifoldTransitionsWithBfields} to argue that the GLSM phases indeed describe $\Xd$ together with a topologically non-trivial flat B-field.
We then calculate the sphere partition function in Section~\ref{sec:spherepartition}, using localization, in order to extract the fundamental period of a mirror Calabi-Yau.
\subsection{The gauged linear sigma models}
\label{sec:glsmModels}
Consider again a normalized decomposition $\vec{d}\in\mathbb{N}^k$ of degree $8$ with length $l\ge 3$ and let $q=\gcd(d_1,\ldots,d_k)$, such that $q\in\{1,2\}$.
We consider GLSMs with gauge group
\begin{align}
    G=\left\{\begin{array}{cl}
    U(1)&q=1\\
    U(1)\times\mathbb{Z}_2&q=2
    \end{array}\right.\,,
\end{align}
vector $U(1)_V$ R-symmetry and chiral fields $P_{i=1,\ldots,k}$, $X_{j=1,\ldots,4}$ with the respective scalar components denoted by $p_i,x_j$.
We denote the $U(1)$ Fayet-Iliopoulous parameter by $r$, the theta angle by $\theta$ and introduce the FI-theta parameter $t=\frac{\theta}{2\pi}+ir$ as well as $v=e^{2\pi i t}$.

The charges of the fields are as follows
\begin{align}
	\begin{array}{c|cc}
		&P_{i=1,\ldots,k}&X_{j=1,\ldots,4}\\\hline
		U(1)&-d_i/q&2/q\\
            \mathbb{Z}_2&-&+\\
		U(1)_V&1-2d_i\mathfrak{q}/q&4\mathfrak{q}/q
	\end{array}\,,
	\label{eqn:glsmCharges}
\end{align}
where for cases with $q=1$ the $\mathbb{Z}_2$ is present as a subgroup of the $U(1)$ gauge symmetry.
We introduce $0<\mathfrak{q}\ll 1\in\mathbb{R}$ as an additional parameter of the R-charges that will be useful in the localization calculation of the sphere partition function but can eventually be set to zero.

 The theory satisfies the Calabi-Yau condition that the $U(1)$ charges of the fields sum to zero.
 As a result, the axial R-symmetry is free of anomalies and and the theory flows to a (2,2)-SCFT in the infrared.
The $\mathbb{Z}_2$ symmetry is required to prevent gauge-invariant chiral operators from having an R-charge $\notin2\mathbb{Z}$, which would break spacetime supersymmetry~\cite{Vafa:1989xc}.
 
Our assignment of $U(1)_V$ charges ensures that the generic superpotential, which has charge $2$, takes the form
\begin{align}
	W=\vec{P}^T\Ad(X)\vec{P}\,,
\end{align}
where $\Ad$ is a symmetric $k\times k$ matrix with entries $a_{i,j}$ that are generic homogeneous polyomials of degree $(d_i+d_j)/2$ in the chiral fields $X_{i=1,\ldots,4}$.

The special case $\vec{d}=(1^8)$ has been studied in detail in~\cite{Aspinwall:1994cf,Caldararu:2010ljp,Addington:2012zv,Katz:2022lyl} while the case $\vec{d}=(2^4)$ has previously been discussed in~\cite{Halverson:2013eua,Sharpe:2013bwa}.

As will be discussed in Section~\ref{sec:conifoldTransitionsWithBfields}, the relationship to the GLSM associated to $\widehat{\Xdr}$ suggests the possibility that there is also a topological term that has to be added to the usual GLSM Lagrangian and which introduces an additional phase depending on the $\mathbb{Z}_2$-bundle on the worldsheet.
Such a term is usually called a discrete theta angle.
Since there are no non-trivial $\mathbb{Z}_2$-bundles on the sphere, the term will not enter the localization calculation of the sphere partition function directly.

Vacua of the theory correspond to gauge equivalence classes of solutions to the F- and D-term equations.
The F-term equations take the form
\begin{align}
	\begin{split}
		\text{(i)}\quad\frac{\partial W}{\partial p_i}=&0\,,\quad i=1,\ldots,k\,,\\
		\text{(ii)}\quad\frac{\partial W}{\partial x_j}=&0\,,\quad i=1,\ldots,4\,,
	\end{split}
	\label{eqn:fterm}
\end{align}
while the D-term equation is
\begin{align}
	\frac{2}{q}\sum\limits_{i=j}^4\lvert x_j\rvert^2-\frac{1}{q}\sum\limits_{i=1}^nd_i\lvert p_i\rvert^2=r\,.
	\label{eqn:dterm}
\end{align}

We will focus in our analysis on the phase $r\gg 0$.
The D-term equations~\eqref{eqn:dterm} then imply that the $x_{j=1,\ldots,4}$ are not allowed to vanish simultaneously and can be interpreted as homogeneous coordinates on $\mathbb{P}^3$.
On the other hand, $\vec{p}=0$ solves the F-term equations and we can interpret $\Ad$ as the mass matrix of a family of Landau-Ginzburg theories -- with a $\mathbb{Z}_2$ orbifold action -- over $\mathbb{P}^3$. 
We therefore obtain a hybrid phase, very similar to the one in the moduli space of the GLSM associated to the intersection of 4 quadrics in $\mathbb{P}^7$.
In fact, as can be seen by comparing the charge vectors and the superpotential, the latter corresponds to the special case $\vec{d}=(1^8)$.

The analysis of the case $\vec{d}=(1^8)$ in~\cite{Caldararu:2010ljp,Sharpe:2013bwa,Sharpe:2013anr,Addington:2012zv,Katz:2022lyl} carries over to a generic choice for $\vec{d}$.
Away from the locus $\{\,\corank\,\Ad>0\,\}\subset \mathbb{P}^3$ the mass matrix of the Landau-Ginzburg fiber has full rank and the fields can be integrated out. However, as a consequence of decomposition~\cite{Hellerman:2006zs} with respect to the trivially acting $\mathbb{Z}_2$ symmetry, the theory has not one but two vacua~\cite{Sharpe:2013anr}.
On the other hand, over $\{\,\corank\,\Ad=1\,\}\subset \mathbb{P}^3$ the Landau-Ginzburg theory contains a single massless field with non-trivial $\mathbb{Z}_2$ action and has a unique vacuum.
Away from the points $\SdOne\subset\mathbb{P}^3$ the hybrid model therefore realizes a double cover of $\mathbb{P}^3$ that is ramified over $\Sd$.

Over the points $\SdOne\subset\mathbb{P}^3$ where the classical geometry of $\Xd$ is singular, the hybrid model remains smooth.

From the open string perspective, this can be interpreted as corresponding to a non-commutative resolution $\Xdnc$ of $\Xd$~\cite{Caldararu:2010ljp,Addington:2012zv}.
The zero branes -- and as a result all topological B-branes -- in a Landau-Ginzburg theory with quadratic superpotential have been identified in~\cite{Kapustin:2002bi} as modules over the associated Clifford algebra.
Taking into account the orbifold action one has to consider instead the even part of the Clifford algebra and given that the hybrid model is fibered over $\mathbb{P}^3$ one obtains a description of branes as modules over a sheaf of even parts of Clifford algebras over $\mathbb{P}^3$.
In this way one reproduces the categorical resolution $(\mathbb{P}^3,\mathcal{B}_0)$ of $\Xd$ due to Kuznetsov~\cite{Kuznetsov2008,Kuznetsov2013} that we have discussed in the previous Section~\ref{sec:cliffordresolution}.

From the closed string perspective, it is more natural to think about the phase as corresponding to $\Xd$ together with a topologically non-trivial flat B-field that stabilizes the singularities~\cite{Aspinwall:1995rb,Schimannek:2021pau,Katz:2022lyl}.
We will now provide further evidence that this is the correct interpretation.

\subsection{Conifold transitions with B-fields from GLSM}
\label{sec:conifoldTransitionsWithBfields}
In this section we will use the GLSM description to study the conifold transitions between $\widehat{\Xdr}$ and $\Xdnc$.
This will allow us to directly see the presence of the fractional B-field along exceptional curves of zero volume in $\Xdnc$ in the phase $r\gg 0$ of the GLSM introduced in Section~\ref{sec:glsmModels}.

For the sake of exposition we focus on the example $\vec{d}=(2^4)$ but the structure is completely analogous for every decomposition $\vec{d}\in\mathbb{N}^k$ of degree $8$ and length $l\ge 3$~\footnote{It also generalizes to analogous constructions over other Fano bases.}.
In the rest of this section we will use $Y=\widehat{X}^{\text{r}}_{(2^4)}$ to denote the small K\"ahler resolution of the specialization $X^{\text{r}}_{(2^4)}$ of $X_{(2^4)}$.

From Lemma~\ref{lem:resolution} it follows that $Y$ can be realized as a complete intersection of two hypersurfaces of respective degrees $(2,2)$ and $(1,2)$ in $\mathbb{P}^2\times\mathbb{P}^3$.
One can construct a corresponding GLSM that has gauge symmetry $G=U(1)_1\times U(1)_2$ and nine chiral fields $P_{1,2}$, $U_{1,2,3}$, $X_{1,2,3,4}$ with gauge charges and $U(1)_V$ R-charges  given by
\begin{align}
    \begin{array}{c|cccc|c}
        & P_1 & P_2 & U_{1,2,3} & X_{1,\ldots, 4}&\text{FI}\\\hline
        U(1)_1& -2&-1&1&0&r_1\\
        U(1)_2& -2&-2&0&1&r_2\\
        U(1)_V&2&2&0&0&
    \end{array}\,.
\end{align}
We denote the corresponding scalar fields by $p_{1,2},u_{1,2,3},x_{1,2,3,4}$.

The Fayet-Iliopoulous parameters are $r_1,r_2$ and denoting the theta angles by $\theta_1,\theta_2$ we define the algebraic coordinates $z_k=\exp(-2\pi r_k+i\theta_k),\,k=1,2$.
A generic superpotential takes the form
\begin{align}
    W=&\vec{\Phi}^T\left(\begin{array}{cc}
    P_1 A_{3\times 3}(X)&B_{1\times 3}(X)^T\\
    B_{1\times 3}(X)&0
    \end{array}\right)\vec{\Phi}\,,
\end{align}
where $\vec{\Phi}=(U_1,U_2,U_3,P_2)$ and $A_{3\times 3},B_{1\times 3}$ are the matrices of quadratic polynomials in $X$ from $\Adr$ in~\eqref{eqn:adegen}.
The geometric phase in which the GLSM flows to a non-linear sigma model on $Y$ corresponds to the region $r_1,r_2\gg 0$.

The complexified K\"ahler parameter associated to the rational curves $C^{(1)}$ that resolve the 80 nodes in $\SdrOne\subset \Xdr$ takes the form
\begin{align}
        t_1(z_1,z_2)=\frac{1}{2\pi i}\log\left(\frac{1-\sqrt{1-4z_1}-2z_1}{2z_1}\right)+\mathcal{O}(z_2)\,.
        \label{eqn:cplxvol}
\end{align}
As discussed in the proof of Lemma~\ref{lem:resolution}, the exceptional curves $C^{(2)}$ that resolve the 8 nodes in $\SdrTwo\subset \Xdr$ are in the homology class $[C^{(2)}]=2\cdot [C^{(1)}]$ and therefore have a complexified volume $2t_1$.
The relation~\eqref{eqn:cplxvol} is actually independent of the choice of $\vec{d}$.

Recall that the complexified K\"ahler parameters take the form
\begin{align}
    t=B+iV\,,
\end{align}
with the imaginary part corresponding to the ``classical'' K\"ahler volume and the real part measuring the holonomy of the B-field along the curve.
The latter is only defined up to shifts by integers and therefore $t\simeq t+1$.

In order to carry out the conifold transition we would like to take a limit where the complexified volume of the curves $C^{(2)}$ is equivalent to zero, i.e. $2t_1\in \mathbb{Z}$.
This only happens for $z_1=1/4$ and $z_1\rightarrow \infty$, with the respective limits of $t_1$ given by
\begin{align}
    \lim_{z_1\rightarrow \frac14}t_1(z_1,z_2)=0\,,\quad \lim_{z_1\rightarrow \infty}t_1(z_1,z_2)=\frac12\,.
    \label{eqn:tlimit}
\end{align}
Note that limits where a complexified K\"ahler parameter takes a fractional value appear quite generically in stringy K\"ahler moduli spaces and have already been observed e.g. in~\cite{Aspinwall:1993xz}.
However, as we will now see, in this case the limit actually allows us to carry out the conifold transition $Y\rightarrow\Xdnc$.

Let us first discuss the underlying geometry.
In the first limit the complexified volumes of all exceptional curves are zero and we can perform a conifold transition to $\Xd$.
On the other hand, in the second limit the exceptional curves $C^{(1)}$ and $C^{(2)}$ have zero K\"ahler volume but there is a $1/2$ B-field holonomy along the curves $C^{(1)}$.
This does not prevent us from deforming away the nodes in $\SdrTwo$, since the corresponding B-field holonomy $1$ is equivalent to $0$, but it does obstruct deformation of the nodes $\SdrOne$.
We therefore obtain a transition to $\Xd$ together with a fractional B-field that stabilizes the singularities, i.e. to $\Xd^{\text{n.c.}}$.

Let us stress that after deforming the nodes associated to $C^{(2)}$, the value of the factional B-field is frozen and we can not continuously deform the theories with $t_1=0$ and $t_1=1/2$ into each other anymore.
As a result, the moduli spaces of the theories on $\Xd$ and $\Xd^{\text{n.c.}}$ -- at a generic point of complex structure for the underlying space $\Xd$ -- are disconnected.
In particular, $\Xd$ and $\Xd^{\text{n.c.}}$ have different mirrors.
We will find further evidence in the next section when calculating the corresponding mirror fundamental periods.

From the perspective of the GLSM, setting $z_1=1/4$ while keeping $r_2\gg 0$ corresponds to a region where the Coulumb branch becomes non-compact and the gauge theoretic description is singular.
However, the limit $z_1\rightarrow \infty,\,r_1\ll 0$, while again keeping $r_2\gg 0$, can be understood from a GLSM perspective.

We will now argue that the second limit precisely leads to the GLSM description of $\Xdnc$ introduced in Section~\ref{sec:glsmModels}.

To this end, we first change the basis of the gauge group such that the new FI-parameters are given by $r_1'=-r_1,r_2'=r_2-r_1$.
We also change the basis for the R-charges such that $P_1$ is only charged under $U(1)_2'$.
The resulting charges are as follows:
\begin{align}
    \begin{array}{c|cccc|c}
        & P_1 & P_2 & U_{1,2,3} & X_{1,\ldots, 4}&\text{FI}\\\hline
        U(1)_1'& 2&1&-1&0&r_1'\\
        U(1)_2'& 0&-1&-1&1&r_2'\\
        U(1)_V'&0&1&1&0&
    \end{array}
    \label{eqn:transitionGlsm2}
\end{align}

In this basis we can directly see how to obtain the GLSM~\eqref{eqn:glsmCharges} associated to the non-commutative resolution $\Xdnc$ of $\Xd$ with $\vec{d}=(2^4)$.
At least heuristically, we can proceed as follows:
\begin{enumerate}
  \setlength{\itemsep}{5pt}
    \item Give a large vacuum expectation value to $p_1$, breaking as a result the gauge symmetry $U(1)_1'$ to a $ \mathbb{Z}_2$ subgroup that acts non-trivially on $P_2,U_{1,2,3}$.
    \item Add a deformation $C_{1\times 1}(X)P_2^2/P_1$ to the superpotential, where $C_{1\times 1}(X)$ is a generic quadratic polynomial in $X_{1,2,3,4}$, such that
    \begin{align}
    W'=&\vec{\Phi}^T\left(\begin{array}{cc}
    P_1 A_{3\times 3}(X)&B_{1\times 3}(X)^T\\
    B_{1\times 3}(X)&(P_1)^{-1}C_{1\times 1}(X)
    \end{array}\right)\vec{\Phi}\,.
\end{align}
    \item Integrate out $P_1$ to obtain a GLSM with fields
    \begin{align}
    \begin{array}{c|ccc|c}
        & P_2 & U_{1,2,3} & X_{1,\ldots, 4}&\text{FI}\\\hline
        \mathbb{Z}_2&-&-&+&\\
        U(1)_2'& -1&-1&1&r_2'\\
        U(1)_V'&1&1&0&
    \end{array}\,.
    \label{eqn:transitionGlsm3}
    \end{align}
    One can now identify $U_{1,2,3},P_2$ and $X_{1,2,3,4}$ in~\eqref{eqn:transitionGlsm2} with the fields $P_{1,2,3,4}$ and $X_{1,2,3,4}$ in~\eqref{eqn:glsmCharges} to recover the GLSM described in Section~\ref{sec:glsmModels}.
\end{enumerate}
Effectively, we performed what looks like a GLSM Higgs transition using $P_1$.
As a non-trivial cross check we verified that the topological string free energies associated to $\Xdnc$ at genus $g=0,1$, that will be further discussed in Section~\ref{sec:direct}, are obtained from those associated to $\widehat{\Xdr}$ after setting $t_1\rightarrow 1/2$.

Recall that the $\theta_1$ theta angle enters the Lagrangian of the GLSM associated to $\widehat{\Xdr}$ via a term
\begin{align}
    \mathcal{L}=\ldots +\frac{\theta_1}{2\pi }\int dv_1\,,
\end{align}
where $v_1$ is the $U(1)_1$ gauge field and $\int dv_1/(2\pi )$ is the corresponding first Chern class~\cite{Witten:1993yc}.
The IR value $t_1\rightarrow 1/2$ in the relevant limit~\eqref{eqn:tlimit} then suggests that the $\mathbb{Z}_2$ gauge symmetry in the GLSM associated to $\Xdnc$ comes with a corresponding topological term that leads to a phase depending on  the choice of $\mathbb{Z}_2$-bundle on the worldsheet.

This seems consistent with the result from~\cite{Hori:2011pd} that applying a Seiberg-like duality to the GLSM associated to $\vec{d}=(1^8)$ produces a GLSM with gauge group $U(1)\times SO(8)$ where the $SO(8)$ is equipped with a non-trivial mod $2$ theta angle.
More generally this points towards a relationship between topological terms in GLSM and fractional B-fields in the infrared theories.

While  the procedure outlined above may seem somewhat ad-hoc -- and certainly asks for a deeper physical understanding in the future --- let us point out that it produces the correct result when applied to ordinary conifold transitions without a B-field that happen in phase limits of GLSM~\cite{Avram:1995pu}.
A more careful analysis shows that the region $r_1',r_2'\gg 0$ is a so-called exoflop phase~\cite{Aspinwall:1993nu,Addington:2013gpa,Aspinwall:2014vea}~\footnote{We thank Ronen Plesser for valuable discussions about this point and for explaining to us the interpretation as an exoflop phase in an example of a conifold transition that is purely geometric as well as the relationship with~\cite{Avram:1995pu}.} and their role in extremal transitions has been discussed e.g. in~\cite{Aspinwall:2014vea}.
We will not try to develop this point of view further in this paper but instead leave it for future work.

\subsection{Localization of the sphere partition function}
\label{sec:spherepartition}
The partition function of two dimensional $\mathcal{N}=(2,2)$ gauge theories was evaluated by localization in~\cite{Benini:2012ui,Doroud:2012xw} and the results were applied to non-linear sigma models on Calabi-Yau threefolds in~\cite{Jockers:2012dk,Halverson:2013eua}, see also~\cite{Erkinger:2019umg}. A recent discussion, that also makes contact with the mathematical literature, can be found e.g. in~\cite{Erkinger:2020cjt}.

The sphere partition function is conjecturally related to the quantum corrected K\"ahler potential on the K\"ahler moduli space via~\cite{Jockers:2012dk,Gomis:2012wy}
\begin{align}
	Z_{S^2}(v,\bar{v})=e^{-K(v,\bar{v})}\,.
\end{align}
Instead of evaluating the full expression for the sphere partition function, we will use the fact that
\begin{align}
    e^{-K(v,\bar{v})}\sim \varpi_0(v)+\mathcal{O}(\log(v),\bar{v})\,,
\end{align}
in order to obtain the fundamental period $\varpi_0(v)$ of the mirror Calabi-Yau threefold.

For a GLSM with gauge group $G=U(1)$, anomaly free axial R-symmetry and chiral fields $\Phi_i$ with $U(1),U(1)_V$ charges $Q_i,q_i$, the localized sphere partition takes the form
\begin{align}
	Z_{S^2}=\frac{1}{2\pi}\sum\limits_{m\in\mathbb{Z}}\int\limits_{-\infty}^\infty d\sigma e^{-4\pi i r\sigma-i\theta m}\prod\limits_{i=1}^n\frac{\Gamma\left(\frac{q_i}{2}-iQ_i\sigma-\frac12 Q_im\right)}{\Gamma\left(1-\frac{q_i}{2}+iQ_i\sigma-\frac12 Q_im\right)}\,,
	\label{eqn:sphereZ}
\end{align}

Considering now the matter content~\eqref{eqn:glsmCharges}.
After a change of variables $\tau=\mathfrak{q}-i\sigma$, the sphere partition function~\eqref{eqn:sphereZ} takes the form
\begin{align*}
\begin{split}
	Z_{S^2}=&e^{-4\pi r \mathfrak{q}}\sum\limits_{m\in\mathbb{Z}}e^{-i\theta m}\int\limits_{\mathfrak{q}-i\infty}^{\mathfrak{q}+i\infty} \frac{d\tau}{2\pi i} e^{4\pi r\tau}\\
 &\cdot\left[\prod\limits_{i=1}^k\frac{\Gamma\left(\frac12-\frac{d_i}{2q}\left(2\tau- m\right)\right)}{\Gamma\left(\frac12+\frac{d_i}{2q}\left(2\tau+ m\right)\right)}\right]\left[\frac{\Gamma\left(\frac{1}{q}(2\tau-m)\right)}{\Gamma\left(1-\frac{1}{q}(2\tau+m)\right)}\right]^4\,.
\end{split}
\end{align*}
We consider the phase $r\gg 0$, which implies that we need to close the contour to the left.
The potential poles of the numerator are located at
\begin{align}
	\tau^{(0)}_n=\frac{1}{2}(m-qn)\,,\quad \tau^{(i=1,\ldots,k)}_n=\frac{1}{2}\left[ m +\frac{q}{d_i}\left(2n+1\right)\right]\,,
\end{align}
for $n\in\mathbb{N}$. One can check that, recalling $0<\mathfrak{q}\ll 1$, the poles at $\tau^{(i)}_n$ are cancelled by corresponding poles in the denominator.
On the other hand, the poles $\tau^{(0)}_n$ are only contained in the contour for $qn\ge m$.
This leads to
\begin{align*}
	\begin{split}
		Z_{S^2}=&\sum\limits_{n=0}^\infty\sum\limits_{m\le qn}v^{-m}\oint \frac{d\epsilon}{2\pi i} (z\bar{z})^{\frac{qn}{2}+\mathfrak{q}+\epsilon}\\
  &\cdot\left[\prod\limits_{i=1}^k\frac{\Gamma\left(\frac{1}{2}+\frac{d_i}{2}n\right)}{\Gamma\left(\frac{1}{2}-\frac{d_i}{2}n+\frac{d_i}{q}m-\epsilon\right)}\right]\left[\frac{\Gamma\left(-n-\epsilon\right)}{\Gamma\left(1+n-\frac{2}{q}m\right)}\right]^4\,.
	\end{split}
\end{align*}

To calculate the regular holomorphic part of the partition function, we can focus on the contributions from $n=0$, replace $(v\bar{v})^{\epsilon}$ by $1$ and set $\mathfrak{q}\rightarrow 0$ to obtain
\begin{align}
	\begin{split}
		Z'=&\sum\limits_{m\ge 0}v^{m}\oint \frac{d\epsilon}{2\pi i} \left[\prod\limits_{i=1}^k\frac{\Gamma\left(\frac{1}{2}\right)}{\Gamma\left(\frac{1}{2}-\frac{d_i}{q}m\right)}\right]\left[\frac{\Gamma\left(-\epsilon\right)}{\Gamma\left(1+\frac{2}{q}m\right)}\right]^4\\
		\propto&\sum\limits_{m\ge 0}2^{-16\frac{m}{q}}v^{m}\left[\prod\limits_{i=1}^k\frac{\Gamma\left(1+\frac{2d_i}{q}m\right)}{\Gamma\left(1+\frac{d_i}{q}m\right)}\right]\frac{1}{\Gamma\left(1+\frac{2}{q}m\right)^4}\,,
	\end{split}
 \label{eqn:locresult1}
\end{align}
where in the second step we have used the identity
\begin{align}
	\Gamma\left(\frac12-n\right)=\frac{(-4)^nn!}{(2n)!}\sqrt{\pi}\,.
\end{align}

It will turn out to be convenient to choose a parametrization $z(v)$ such that the complexified volume of a degree 1 curve in the $\mathbb{Z}_2$-charge $0$ sector takes the form
\begin{align}
    t(z)=\frac{1}{2\pi i}\log(z)+\mathcal{O}(z)\,.
\end{align}
This leads us to identify $z=2^{-\frac{16}{q^2}}v^{1/q}$ and we obtain the final result that the fundamental period of the mirror takes the form
\begin{align}
	\varpi_0(z)=\sum\limits_{m\ge 0}z^{qm}\left[\prod\limits_{i=1}^k\frac{\Gamma\left(1+\frac{2d_i}{q}m\right)}{\Gamma\left(1+\frac{d_i}{q}m\right)}\right]\frac{1}{\Gamma\left(1+\frac{2}{q}m\right)^4}\,.
 \label{eqn:fundamentalPeriod}
\end{align}

\section{Borisov-Li mirrors of Clifford nc-resolutions}
\label{sec:gorenstein}
We will now explicitly construct the Calabi-Yau threefold mirrors of $\Xd^{\text{n.c.}}$ using the combinatorial description of Clifford non-commutative resolutions via reflexive Gorenstein cones from~\cite{borisov2016clifford,Li2020}.
We then calculate the corresponding fundamental periods in the examples $\vec{d}=(2^4)$ and $\vec{d}=(5,1^3)$ and find that they match the result~\eqref{eqn:fundamentalPeriod} from the localization calculation in the previous section.

We only provide a lightning review of the general construction, essentially to fix our notations.
For further background on the nef partitions formulation  of mirror symmetry for complete intersections in toric varieties, the reader is referred to~\cite{borisov1993,Batyrev:1994pg}.  The Gorenstein cone formulation of mirror symmetry for complete intersections in toric varieties was introduced in \cite{ Batyrev:1996hiq} and is nicely discussed in~\cite{Batyrev:2007cq}.

\subsection{Nef partitions and Gorenstein cones}
\label{sec:nefgorenstein}
Consider a lattice $M\simeq \mathbb{Z}^k$ and denote the dual lattice by $N$. A reflexive polytope $\Delta$ in $M_\mathbb{R}=M\otimes\mathbb{R}$ determines a toric variety $\mathbb{P}_\Delta$ which is projective, Gorenstein, and Fano.  The Gorenstein condition says that the singularities are mild enough so that $\mathbb{P}_\Delta$ has a canonical bundle, and the Fano condition is that the anticanonical bundle is ample.

The complement of the dense torus $T=T^k$ in $\mathbb{P}_\Delta$ is the union $D_1\cup\ldots \cup D_\ell$ of $T$-invariant divisors $D_i$, whose sum is an anticanonical divisor on $\mathbb{P}_\Delta$.  The divisors $D_i$ are in 1-1 correspondence with the codimension 1 faces $\Theta_i\subset\Delta$.  To each $D_i$ is associated a vector $e_i\in N$ normal to $\Theta_i$ satisfying the equalities
\begin{align}
\begin{split}
\Delta=&\left\{m\in M_\mathbb{R}\mid \langle m,e_i\rangle\ge -1\ \forall\ 1\le i\le \ell\right\}\,,\\
\Theta_i=&\left\{m\in \Delta\mid \langle m,e_i\rangle =-1\right\}\,.
\end{split}
\label{eqn:facet}
\end{align}
The polar polytope of $\Delta$ is given by
\begin{align}
    \Delta^\circ = \{n\in N_\mathbb{R}\mid \langle m,n\rangle\ge-1\ \forall\ m\in \Delta \}=\text{Conv}(e_1,\ldots,e_\ell)\,.
\end{align}
The polytope $\Delta^\circ$ is also reflexive, and $(\Delta^\circ)^\circ=\Delta$.

A nef partition with $r$ parts is a partition of $\{D_1,\ldots,D_\ell\}$ into $r$ nonempty subsets $S_1,\ldots,S_r$ such that for each $j$, the divisor $E_j=\sum_{D_i\in S_j}D_i$ is generated by global sections, and in particular is nef.  Reformulating in terms of $\Delta$ as in \cite{Batyrev:1994pg}, this data determines a set of lattice polytopes 
$\Delta_j\subset \Delta$ consisting of all $m\in \Delta$ satisfying
\begin{align}
\begin{split}
    \langle m,e_i\rangle \ge -1\ {\rm if\ } D_i\in S_j\,,\\
    \langle m,e_i\rangle \ge 0\ {\rm if\ } D_i\not\in S_j\,.
    \end{split}
\end{align}
These polytopes have Minkowski sum $\Delta=\Delta_1+\ldots+\Delta_r$.

For a given lattice polytope $P\subset\mathbb{Z}^k\otimes\mathbb{R}$ and generic choice of complex coefficients $c_p\in\mathbb{C}$ for all $p\in P\cap\mathbb{Z}^k$ we define the Laurent polynomial 
\begin{align}
    F_P(x)=\sum\limits_{p\in P\cap \mathbb{Z}^k}c_p x^p\,,
\end{align}
using coordinates $(x_1,\ldots, x_k)$ on $T^k\simeq \left(\mathbb{C}^*\right)^k$.
The polynomials $F_{\Delta_j}(x)$ on $T^k$ extend to holomorphic sections of $\mathcal{O}(E_j)$ on $\mathbb{P}_\Delta$.  The zero locus of $F_{\Delta_j}(x)$ is a hypersurface  $Z_j\subset\mathbb{P}_\Delta$.

The toric variety $\mathbb{P}_\Delta$ can be blown up to a toric variety $\widetilde{\mathbb{P}_\Delta}$ determined by an appropriate choice of a fan $\Sigma$ subdividing the normal fan of $\Delta$, i.e.\ the fan formed by the cones over the faces of $\Delta^\circ$.  Letting $\widetilde{Z_j}$ be the proper transform of $Z_j$ on $\widetilde{\mathbb{P}_\Delta}$, we get a Calabi-Yau complete intersection $Z_{\Delta_{i=1,\ldots,r}}=\widetilde{Z_1}\cap\cdots\cap \widetilde{Z_r}\subset\widetilde{\mathbb{P}_\Delta}$.
The choice of fan will not affect the calculation of the mirror periods and in the following we will leave it implicit.

For each $1\le j\le r$, we put
\begin{align}
    \nabla_j=\text{Conv}\left(\{e_i\mid D_i\in S_j\}\cup\{0\}\right)\subset N_\mathbb{R}\,.
\end{align}
Then the Minkowski sum
\begin{align}
 \nabla=\nabla_1+\ldots+\nabla_r \,,
\label{eqn:nefpartition}
\end{align}
is reflexive.  We also have
\begin{align}
\nabla^{\circ}=\text{Conv}(\Delta_1,\ldots,\Delta_r)\,,\qquad \Delta^\circ =\text{Conv}(\nabla_1,\ldots,\nabla_r)\,,
\label{eqn:nefpartdelta}
\end{align}
while for all $p\in \Delta_i$ and $q\in\nabla_j$ one has $\langle p,q\rangle\ge-\delta_{i,j}$.

Having achieved a symmetric situation after an appropriate subdivision, to each polytope $\nabla_i$ we have a generic hypersurface $\widetilde{Z_i^\circ}\subset \widetilde{\mathbb{P}_{\nabla}}$ and a Calabi-Yau complete intersection $Z_{\nabla_{i=1,\ldots,r}}=\widetilde{Z_1^\circ}\cap\cdots\cap \widetilde{Z_r^\circ}$ which is mirror to $Z_\Delta\subset \widetilde{\mathbb{P}_\Delta}$.

The combinatorics of nef partitions can be rephrased, and generalized, using reflexive Gorenstein cones~\cite{Batyrev:1996hiq,Batyrev:2007cq}.
A $k$-dimensional integral cone $\sigma\subset M_\mathbb{R}$ is Gorenstein if there is an $n_\sigma\in N$ such that $\langle e,n_\sigma\rangle=1$ for each primitive integral generator $e$ of $\sigma$.  The vector $n_\sigma$ is unique, and is called the degree element $\mathrm{deg}^\vee$ in~\cite{borisov2016clifford}.

A reflexive Gorenstein cone $K\in\widetilde{M}_{\mathbb{R}}$ with $\widetilde{M}=M\oplus \mathbb{Z}^{r}$ is generated by the points $(p,\delta_{0,i})$ for all $p\in\Delta_i$ and $i=1,\ldots,r$.

The dual cone $K^\vee\in\widetilde{N}_{\mathbb{R}}$ with $\widetilde{N}=N\oplus \mathbb{Z}^{r}$ is similarly generated by the points  $(q,\delta_{0,j})$ for $q\in\nabla_j$ and $j=1,\ldots,r$.

The corresponding degree elements for $K$ and $K^\vee$ are $\text{deg}^\vee=(0,1,\ldots, 1)\in N\oplus \mathbb{Z}^r$ and $\text{deg}=(0,1,\ldots, 1)\in M\oplus \mathbb{Z}^r$ respectively, where in both cases the vector contains $r$ ones and $0$ denotes the zero element of $N$ (resp.\ $M$).  We also have $\langle \text{deg},\text{deg}^\vee\rangle=r$, which by definition indicates that the Gorenstein cones $K$ and $K^\vee$ each have index $r$.
One defines
\begin{align}
    K_{(1)}=\{\,p\in K\cap \widetilde{M}\mid \langle p,\text{deg}^\vee\rangle=1\,\}\,,
\end{align}
and similarly $K^\vee_{(1)}$.  

The nef partition in the form (\ref{eqn:nefpartdelta}) can be recovered from $K$ by the decomposition of the degree element
\begin{align}
    \text{deg}^\vee=u_1+\ldots+u_r\,,
\end{align}
with $u_i=(0,\delta_{1,i},\ldots,\delta_{r,i})\in K^\vee_{(1)}$.  
The dual nef partition (\ref{eqn:nefpartition}) can similarly be recovered from $K^\vee$ and the decomposition
\begin{align}
\text{deg}=t_1+\ldots+t_r\,,
\end{align}
of $\text{deg}$,  with $t_i=(0,\delta_{0,i},\ldots,\delta_{r,i})\in K_{(1)}$.  We have $\langle t_i,u_j\rangle=\delta_{ij}$.

To recover the nef partition description from the Gorenstein cone description without reference to the direct sum decomposition of $\widetilde{M}$ and $\widetilde{N}$, 
the lattice $N$ is recovered as the quotient lattice $\widetilde{N}/\langle u_1,\ldots, u_r\rangle \simeq N$, and the lattice $M$ is recovered as the quotient lattice $\widetilde{M}/\langle t_1,\ldots, t_r\rangle \simeq M$.\footnote{Using both of these descriptions obscures the intrinsic pairing between $M$ and $N$. In order to more easily see the pairing, one can use the quotient description of $N$ and identify $M$ with the annihilator of $\{u_1,\ldots,u_r\}$ in $\widetilde{M}$.  The restriction of the projection map $\widetilde{M}\rightarrow\widetilde{M}/\langle t_1,\ldots, t_r\rangle$ to the annihilator of $\{u_1,\ldots,u_r\}$ is an isomorphism of lattices, which can be used to compare the two descriptions of $M$ when desired.}
The polytopes $\Delta_i$ are the convex hulls of the images in $M$ of the points $p\in K$ that satisfy $\langle u_j,p\rangle =\delta_{i,j}$ for $j=1,\ldots,r$.

Using the isomorphism $M\simeq \text{Ann}(u_1,\ldots, u_r)\subset\widetilde{M}$, a generic function
\begin{align}
    c:K_{(1)} \rightarrow \mathbb{C}\,,
\label{eqn:coefficients}
\end{align}
determines the coefficients of the polynomials $F_{\Delta_i}(x)$ defining the hypersurfaces $Z_i$ used to construct the Calabi-Yau complete intersection $Z_1\cap\cdots\cap Z_r\subset\mathbb{P}_\Delta$.

The points in $p\in K_{(1)}^\vee\backslash U$ with $U=\{u_1,\ldots,u_r\}$ are associated to homogeneous coordinates $x_p$ on $\mathbb{P}_{\Delta}$.
We can also associate formal variables $p_i$ to $u_i$ for $i=1,\ldots,r$.
The function~\eqref{eqn:coefficients} then determines a polynomial 
\begin{align}
  P=\sum\limits_{m\in K_{(1)}}c(m)\prod\limits_{i=1}^rp_i^{\langle m,u_i\rangle} \prod\limits_{p\in K_{(1)}^\vee\backslash U}x_p^{\langle m,p\rangle}\,.
  \label{eqn:gorensteinP1}
\end{align}
From $\langle m,\text{deg}^\vee\rangle =1$ for all $m\in K_{(1)}$ it follows that $P$ is linear in the $p_i$.
The coefficient of $p_i$ can be identified with $F_{\Delta_i}(x)$.

More abstractly, without recourse to the nef partition formalism, we can start with dual lattices $\widetilde{M}$ and $\widetilde{N}$, dual reflexive Gorenstein cones $K\subset \widetilde{M}_\mathbb{R}$ and $K^\vee\subset \widetilde{N}_\mathbb{R}$ of index $r$, with degree elements $\deg^\vee\in K^\vee\cap\widetilde{N}$ for $K$ and $\deg\in K\cap\widetilde{M}$ for $K^\vee$, such that $\langle \deg,\deg^\vee\rangle =r$ and the generators of $K,K^\vee$ are respectively contained in $K_{(1)}$ and $K_{(1)}^\vee$.

One can then consider decompositions $\deg^\vee=u_1+\ldots u_r$ with $u_i\in K^\vee_{(1)}$ and $\deg=t_1+\ldots t_r$ with $t_i\in K_{(1)}$, satisfying $\langle t_i,u_j\rangle = \delta_{ij}$.
The above discussion can be modified to obtain a nef partition $\Delta_1,\ldots\Delta_r$ in $M_\mathbb{R}$, with $M=\text{Ann}(u_1,\ldots,u_r)$. 
This decomposition leads to a mirror nef partition $\nabla_1,\ldots,\nabla_r$ and a mirror pair of complete intersection Calabi-Yaus in toric varieties.

As was pointed out in~\cite{Batyrev:2007cq,borisov2016clifford}, in situations where different decompositions of $\text{deg}^\vee$ or $\text{deg}$ are possible, this leads to the double mirror phenomenon~\cite{Aspinwall:1993yb} where multiple large complex structure limits appear in the same moduli space.

Decompositions of the form
\begin{align}
    \text{deg}^\vee=\frac12\left(s_1+\ldots +s_{2r}\right)\,,\quad s_i\in K_{(1)}^\vee\,,
\end{align}
have been associated to Clifford type non-commutative resolutions in~\cite{borisov2016clifford}, building also on results from~\cite{Kuznetsov2008,BFK19}.
In this situation, the function (\ref{eqn:coefficients}) encodes the data of a bundle of quadrics, which has an associated bundle of even Clifford algebras \cite{Kuznetsov2008}.

Analogous to~\eqref{eqn:gorensteinP1}, we can associate formal variables $p_i$ to $s_i$ for $i=1,\ldots,{2r}$ and homogeneous coordinates $x_p$ to each $p\in K_{(1)}^\vee\backslash S$ with $S=\{s_1,\ldots,s_{2r}\}$.
The function~\eqref{eqn:coefficients} then determines a polynomial
\begin{align}
  P=\sum\limits_{m\in K_{(1)}}c(m)\prod\limits_{i=1}^{2r}p_i^{\langle m,s_i\rangle} \prod\limits_{p\in K_{(1)}^\vee\backslash S}x_p^{\langle m,p\rangle}\,,
  \label{eqn:gorensteinP2}
\end{align}
 and from $\langle \text{deg}^\vee, m\rangle =1$ it follows that $P$ is quadratic in the $p_i$.  
 
 The $p_i$ are interpreted as fiber coordinates on a line bundle $\mathcal{L}_i$ over a toric base $B$ (which may be a toric stack rather than a toric variety).  Then the polynomial (\ref{eqn:gorensteinP2}) is quadratic on the fibers and so induces a bundle of quadric hypersurfaces in $\mathbb{P}(\mathcal{L}_1\oplus\cdots\oplus \mathcal{L}_{2r})$.  In this situation, we get a sheaf of even Clifford algebras $\mathcal{B}_0$ on $B$ \cite{Kuznetsov2008}.

More generally, one can consider decompositions of the form
\begin{align}
    \text{deg}^\vee=u_1+\ldots+u_{r-r'}+\frac12\left(s_1+\ldots+s_{2r'}\right)\,,
    \label{eqn:degus}
\end{align}
which lead to non-commutative resolutions associated to quadric bundles over complete intersections in toric varieties.

\subsection{Mirrors of Clifford nc-resolutions}
\label{sec:gorensteincones}
Lev Borisov kindly provided to us the following construction of the Clifford nc-resolutions of $\Xd$ and their mirrors in terms of decompositions of degree elements of dual Gorenstein cones based on~\cite{borisov2016clifford}.

We consider a basis $e_{i=1,\ldots,4}^*,\,f_{j=1,\ldots,k}^*$ of a rank $k+4$ lattice $M'$ and denote the dual basis on the dual lattice $N'$ by $e_{i=1,\ldots,4},\,f_{j=1,\ldots,k}$.
Using $\Sigma_e=\sum_{i=1}^4e_i$, $\Sigma_f=\sum_{j=1}^kf_i$, with corresponding definitions for $\Sigma_{e^*},\Sigma_{f^*}$, and $\Sigma_{f,\vec{d}}=\sum_{j=1}^k d_i f_i$ we then define the lattice
\begin{align}
    \widetilde{M}=\left\{\,v \in M'\,\,\vert\,\,\langle v, 2\Sigma_e-\Sigma_{f,\vec{d}}\rangle =0\,,\, \langle v,\Sigma_{f}\rangle = 0\text{ mod }2\,\right\}\,,
    \label{eqn:mlattice}
\end{align}
with dual lattice
\begin{align}
\widetilde{N}=\left(N'+\mathbb{Z}\cdot \frac{\Sigma_{f}}{2}\right)\big/\left(\Sigma_e-\frac12\Sigma_{f,\vec{d}}\right)\,.
\label{eqn:ntilde}
\end{align}
The dual lattices $\widetilde{M}$ and $\widetilde{N}$ have rank $k+3$.
We freely identify elements of $M',N'$ and their images in/lifts to $\widetilde{M},\widetilde{N}$.
The elements
\begin{align}
    d_i e_l^*+2f_i^*\in \widetilde{M}\,,\quad 1\le l\le 4\,,\quad 1\le i\le k\,,
    \label{eqn:gorensteinGens}
\end{align}
generate a Gorenstein cone $K\subset \widetilde{M}_{\mathbb{R}}$ with degree element $\text{deg}^{\vee}=\frac12\Sigma_f$. The primitive elements of $K_{(1)}$ are given by
\begin{align}
    \sum_{l=1}^4 a_le_l^*+f_i^*+f_j^*\in K\,,\quad a_l\ge0,\ \sum_{i=1}^4 a_l=\frac{d_i+d_j}2, \quad 1\le i,j\le k\,.
\label{eqn:k1prim}
\end{align}
The dual cone $K^\vee \in \widetilde{N}_{\mathbb{R}}$ is generated by $e_{i=1,\ldots,4}$, $f_{j=1,\ldots,k}$ and the degree element is $\text{deg}=\Sigma_{e^*}+\Sigma_{f^*}$.  The index of the Gorenstein cones $K$ and $K^\vee$ is $\langle \deg,\deg^\vee\rangle=k/2$ (recall that $k=2n$ is even).

The relevant decomposition of the degree element $\text{deg}^{\vee}$ that corresponds to the Clifford nc-resolution of $X_{\vec{d}}$ is given by
\begin{align}
\text{deg}^{\vee}=\frac12\left(f_1+\ldots+f_k\right)\,.
\end{align}

We next sketch how this toric description corresponds to the bundle of quadrics over $\mathbb{P}^3$ described in Section~\ref{sec:sods}.  We consider the fan $\Sigma$ in $\tilde{N}_\mathbb{R}$ whose maximal cones are spanned by the $f_j$ and all but one of the $e_i$.  We claim that the resulting toric stack $X_\Sigma$ is the total space of an orbi-vector bundle over $\mathbb{P}^3$.

Define the rank 3 lattice $\underline{N}=\widetilde{N}/\langle f_1,\ldots,f_k,\deg^\vee\rangle$.  Then the images $\underline{e}_i$ of $e_i$ in $\underline{N}$ generate $\underline{N}$ and satisfy $\underline{e}_1+\underline{e}_2+\underline{e}_3+\underline{e}_4=0$.  Let $\underline{\Sigma}$ be the complete fan in $\underline{N}_\mathbb{R}$ with edges spanned by the $\underline{e}_i$. The associated toric variety $X_{\underline{\Sigma}}$ is isomorphic to $\mathbb{P}^3$.

Since the natural projection $\widetilde{N}\rightarrow\underline{N}$ sends the cones of $\Sigma$ to the cones of $\underline{\Sigma}$, we get a map of toric varieties $X_\Sigma\to X_{\widehat{\Sigma}}=\mathbb{P}^3$ which is easily seen to be a fiber bundle.  The fiber is the toric stack defined by the fan in the lattice $\langle f_1,\ldots,f_k,\deg^\vee\rangle$
with a single maximal cone spanned by the $f_j$.  If the lattice were spanned by $f_1,\ldots,f_k$, we would get the total space of a direct sum of line bundles on $\mathbb{P}^3$.  However, that lattice has index 2 in the larger lattice due to the presence of $\deg^\vee$, so we get a $\mathbb{Z}_2$ stacky quotient of that vector bundle over $\mathbb{P}^3$.   The relationship $2\deg^\vee=\Sigma_f$ shows that this $\mathbb{Z}_2$ acts as $-1$ on each of the coordinates $p_i$.  This additional $\mathbb{Z}_2$ action becomes trivial after projectivization and so we deduce a projective bundle $\mathbb{P}(\mathcal{L}_1\oplus\cdots\oplus \mathcal{L}_r)$, as stated at the end of Section~\ref{sec:nefgorenstein}. This $\mathbb{Z}_2$ perfectly matches the $\mathbb{Z}_2$ appearing in the GLSM description in (\ref{eqn:glsmCharges}) and the subsequent discussion.

Finally, to complete the translation to Section~\ref{sec:sods} we show that the coefficient of $p_ip_j$ has degree $(d_i+d_j)/2$ on $\mathbb{P}^3$.  From (\ref{eqn:gorensteinP2}) (where we now use $f_\ell$ in place of $s_\ell$), we see that the monomials involving $p_ip_j$ come from the elements (\ref{eqn:k1prim}) of $K_{(1)}$ with the indicated values of $i,j$.
Then the coefficient of $p_ip_j$ coming from one of these generators is precisely $\sum_{l=1}^4 x_l^{a_l}$, a monomial of degree $(d_i+d_j)/2$.  We can then realize any polynomial of degree $(d_i+d_j)/2$ in $x_1,\ldots,x_4$ by choosing an appropriate coefficient function $c$.

To construct a mirror decomposition of $\text{deg}$ we first choose subsets $I_i\subset \{e_1^*,\ldots,e_4^*\},\,i=1,\ldots,\frac{k}{2}$ such that $\#I_i=\frac{d_{2i-1}+d_{2i}}{2}$ and $\cup_i I_i=\{e_1^*,\ldots,e_4^*\}$.
We then define the generators
\begin{align}
    t_i=\sum\limits_{e^*\in I_i}e^*+f^*_{2i-1}+f^*_{2i}\,,
    \label{eqn:tdec}
\end{align}
and use the decomposition $\text{deg}=t_1+\ldots+t_{\frac{k}{2}}$.
This realizes the mirror of $\Xdnc$ as a toric complete intersection Calabi-Yau threefold without a topologically non-trivial flat B-field.
All we need to do to make contact with Section~\ref{sec:spherepartition} is to calculate the fundamental period.

\bigskip
In the following sections, we look at two examples.  Rather than describing our Gorenstein cones in the more symmetric form given above, we break symmetry by constructing isomorphisms $\widetilde{M}\simeq \mathbb{Z}^{4+k-1}$ that lead to an explicit description of the toric variety associated to $K^\vee$ and the Calabi-Yau threefold complete intersection associated to the decomposition of $\text{deg}$ which is more amenable to computer calculation.

\subsection{$\vec{d}=(2^4),\,n_s=80$}
We first construct a map $T:\,\mathbb{Z}^7\rightarrow M'$ that restricts to an isomorphism on $\widetilde{M}$~\eqref{eqn:mlattice}.
As a matrix we can choose this to take the form
\begin{align}
    T=\left(
\begin{array}{ccccccc}
 2 & 1 & 0 & 0 & 0 & 1 & 0 \\
 0 & 1 & 1 & 0 & 0 & 1 & 0 \\
 0 & 0 & 1 & 1 & 0 & 0 & 1 \\
 0 & 0 & 0 & 1 & 0 & 0 & 1 \\
 2 & 0 & 0 & 0 & 0 & 1 & 0 \\
 0 & 2 & 0 & 0 & 1 & 1 & 0 \\
 0 & 0 & 2 & 0 & -1 & 0 & 1 \\
 0 & 0 & 0 & 2 & 0 & 0 & 1 \\
\end{array}
    \right)\,.
    \label{eqn:gextmatrix}
\end{align}
We denote by $T'$ the full rank matrix that is obtained by adding the column $(0^7,1)$ to $T$.
We can then use $(T')^{-1}$ to translate the generators~\eqref{eqn:gorensteinGens} into $\mathbb{Z}^7\oplus\mathbb{Z}$ and drop the last component to obtain the image in $\mathbb{Z}^7$.

We can choose the subsets $I_1,I_2$ in~\eqref{eqn:tdec} such that
\begin{align}
    t_1=e_1^*+e_2^*+f_1^*+f_2^*\,,\quad t_2=e_3^*+e_4^*+f_3^*+f_4^*\,.
    \label{eqn:gextdec}
\end{align}
Note that these correspond to the columns 6,7 in~\eqref{eqn:gextmatrix} such that we can directly identify $\mathbb{Z}^7=\mathbb{Z}^5\oplus\mathbb{Z}^2=M\oplus\mathbb{Z}^2$.

To obtain the polytope $\nabla^{\circ}$ associated to the dual ``nef partition'', we therefore drop the last two components, amounting to the projection from $\widetilde{M}$ to $M$.
The resulting matrix $\left[(T')^{-1}\right]_{\text{res.}}$, that maps a vector in $M'$ to a corresponding vector in $\mathbb{Z}^5$, takes the form
\begin{align}
    \left[(T')^{-1}\right]_{\text{res.}}=\left(
\begin{array}{cccccccc}
 \frac{1}{2} & -\frac{1}{2} & \frac{1}{2} & -\frac{1}{2} & 0 & 0 & 0 & 0 \\
 1 & 0 & 0 & 0 & -1 & 0 & 0 & 0 \\
 0 & 0 & 1 & -1 & 0 & 0 & 0 & 0 \\
 1 & 1 & 1 & 0 & -1 & -1 & -1 & 0 \\
 -1 & -1 & 1 & -1 & 1 & 1 & 0 & 0 \\
\end{array}
\right)\,.
\end{align}
Applying this to the generators~\eqref{eqn:gorensteinGens}, and taking into account the decomposition~\eqref{eqn:gextdec}, we then find the polytopes $\nabla^{\circ}=\text{Conv}(\Delta_1,\Delta_2)$ with $\Delta_1,\Delta_2\subset\mathbb{Z}^5$ being as follows:
\begin{align}
\Delta_1=\text{Conv}\left(
\begin{array}{ccccccccc}
 1 & 1 & -1 & -1 & 1 & 1 & -1 & -1 & 0 \\
 0 & 2 & -2 & 0 & -2 & 0 & -2 & 0 & 0 \\
 0 & 0 & 0 & 0 & 2 & 2 & -2 & -2 & 0 \\
 0 & 0 & 0 & 0 & 0 & 0 & -2 & -2 & 0 \\
 0 & 0 & 0 & 0 & 4 & 4 & 0 & 0 & 0 \\
\end{array}
\right)
\label{eqn:2222delta1}
\end{align}
\begin{align}
\Delta_2=\text{Conv}\left(
\begin{array}{ccccccccc}
 1 & 1 & -1 & -1 & 1 & 1 & -1 & -1 & 0 \\
 2 & 2 & 0 & 0 & 0 & 0 & 0 & 0 & 0 \\
 0 & 0 & 0 & 0 & 2 & 2 & -2 & -2 & 0 \\
 0 & 2 & 0 & 2 & 0 & 2 & -2 & 0 & 0 \\
 -2 & -2 & -2 & -2 & 2 & 2 & -2 & -2 & 0 \\
\end{array}
\right)
\label{eqn:2222delta2}
\end{align}
Note that we write the points as columns.
One can check using SageMath~\cite{sagemath} that $\nabla^{\circ}$ is reflexive.

Using the dual transformation on the generators of $K^\vee$ we obtain the polytopes
\begin{align}
    \nabla_1=\text{Conv}\left(
\begin{array}{cccc}
 1 & -1 & 1 & -1 \\
 0 & 0 & -1 & 1 \\
 0 & 1 & 0 & 0 \\
 0 & 0 & 0 & 0 \\
 -\frac{1}{2} & -\frac{1}{2} & -\frac{1}{2} & \frac{1}{2} \\
\end{array}
\right)\,,\quad \nabla_2=\text{Conv}\left(
\begin{array}{cccc}
 0 & 0 & 0 & 0 \\
 0 & 0 & 0 & 0 \\
 0 & -1 & 1 & -1 \\
 0 & 0 & -1 & 1 \\
 \frac{1}{2} & \frac{1}{2} & -\frac{1}{2} & \frac{1}{2} \\
\end{array}
\right)\,,
\end{align}
in $\mathbb{Z}^4\oplus (1/2)\mathbb{Z}$.
One can again check that the Minkowski sum $\nabla=\nabla_1+\nabla_2$ is a reflexive lattice polytope and the polar dual of $\nabla^\circ$.

One can also verify that $2\cdot\text{Conv}(\nabla_1,\nabla_2)$ is also reflexive and the polar dual of $(\Delta_1+\Delta_2)/2$.
The diagram~\eqref{eqn:nefpartition} therefore appears to be modified to
\begin{align}
\begin{array}{c}
\Delta=(\Delta_1+\ldots+\Delta_r)/2\\[.2cm]
\nabla^{\circ}=\text{Conv}(\Delta_1,\ldots,\Delta_r)
\end{array}\qquad
\begin{array}{c}
\Delta^\circ=2\cdot \text{Conv}( \nabla_1,\ldots,\nabla_r)\\[.2cm]
\nabla=\nabla_1+\ldots+\nabla_r
\end{array}\,,
\label{eqn:nefpartition2}
\end{align}
while still for all $p\in \Delta_i$ and $q\in\nabla_j$ one has $\langle p,q\rangle\ge-\delta_{i,j}$.

We can shift the points in $\nabla_1,\nabla_2$ using e.g. $\vec{p}=(0,0,0,0,1/2)$ to obtain the polytopes $\nabla_1'=\nabla_1+\vec{p}$ and $\nabla_2'=\nabla_2-\vec{p}$ with
\begin{align}
    \nabla_1'=\text{Conv}\left(
\begin{array}{cccc}
 1 & -1 & 1 & -1 \\
 0 & 0 & -1 & 1 \\
 0 & 1 & 0 & 0 \\
 0 & 0 & 0 & 0 \\
 0 & 0 &0& 1\\
\end{array}
\right)\,,\quad \nabla_2'=\text{Conv}\left(
\begin{array}{cccc}
 0 & 0 & 0 & 0 \\
 0 & 0 & 0 & 0 \\
 0 & -1 & 1 & -1 \\
 0 & 0 & -1 & 1 \\
0 & 0& -1 &0 \\
\end{array}
\right)\,,
\label{eqn:nablashifted2222}
\end{align}
in order to obtain lattice polytopes in $\mathbb{Z}^5$ while still preserving $\nabla=\nabla_1'+\nabla_2'$.
However, only $\nabla_2'$ contains the origin.

Note that if we just replace $\nabla_1'$ by the convex hull of the points with the origin added, the resulting polytopes will not have a Minkowski sum equal to $\nabla$.
Instead we obtain an actual nef partition 
\begin{align}
\begin{array}{c}
\widetilde{\Delta}=\widetilde{\Delta}_1+\ldots+\widetilde{\Delta}_r\\[.2cm]
\widetilde{\nabla}^{\circ}=\text{Conv}(\widetilde{\Delta}_1,\ldots,\widetilde{\Delta}_r)
\end{array}\qquad
\begin{array}{c}
\widetilde{\Delta}^\circ=\text{Conv}( \widetilde{\nabla}_1,\ldots,\widetilde{\nabla}_r)\\[.2cm]
\widetilde{\nabla}=\widetilde{\nabla}_1+\ldots+\widetilde{\nabla}_r
\end{array}\,,
\label{eqn:nefpartitionTulde}
\end{align}
of polytopes $\widetilde{\Delta},\widetilde{\nabla}$ with $\widetilde{\nabla}_1=\text{Conv}(\nabla_1',0)$ and $\widetilde{\nabla}_2=\nabla_2'$.
Note that different choices of shift lead to different $\nabla_1',\nabla_2'$ and therefore different $\widetilde{\Delta},\widetilde{\nabla}$.
For our choice in~\eqref{eqn:nablashifted2222} the resulting nef partition happens to be such that $Z_{\widetilde{\nabla}_1,\widetilde{\nabla}_2}=\widehat{\Xdr}\,$!

Writing down the defining equations of the mirror $Z_{\nabla'_1,\nabla'_2}$ for $\Xdnc$ in coordinates $x_1,\ldots,x_5$ on the algebraic torus $(\mathbb{C}^*)^5$, and labelling the coefficients corresponding to the points in $\nabla_1',\nabla_2'$ (in the order in which they appear above) by $a_1,\ldots,a_8$, we obtain the equations
\begin{align}
\begin{split}
    p_1=&a_1 x_1+\frac{a_2 x_3}{x_1}+\frac{a_3 x_1}{x_2}+\frac{a_4 x_2 x_5}{x_1}\,,\\
    p_2=&a_5+\frac{a_6}{x_3}+\frac{a_7 x_3}{x_4 x_5}+\frac{a_8 x_4}{x_3}\,.
    \end{split}
\end{align}

In order to obtain the fundamental period $\varpi_0$, we first add a constant term $a_0$ to $p_1'$ to obtain
\begin{align}
    \tilde{p}_1=&a_0+a_1 x_1+\frac{a_2 x_3}{x_1}+\frac{a_3 x_1}{x_2}+\frac{a_4 x_2 x_5}{x_1}\,.
\end{align}
We then calculate
\begin{align}
    \varpi_0'=\int\limits_{\Gamma_0}\frac{a_0\cdots a_8}{\tilde{p_1}p_2}\prod\limits_{i=1}^5\frac{dx_i}{x_i}\,,
    \label{eqn:funddef}
\end{align}
where the contour $\Gamma_0$ is along $\vert x_i\vert=1,\,i=1,\ldots,5$, to obtain the fundamental period of the corresponding deformation of $Z_{\nabla'_1,\nabla'_2}$ following~\cite{Hosono:1994ax}.
The expression can be evaluated by expanding in powers of $1/a_0$ and $1/a_5$.
Using
\begin{align}
    w=\frac{a_3a_4a_7a_8}{a_0^2a_5^2}\,,\quad v=-\frac{a_1a_2a_6}{a_0^2a_5}\,,
\end{align}
we then obtain the result
\begin{align}
\begin{split}
    \varpi_0'(v,w)=&\sum\limits_{n,k=0}^\infty\frac{(2k+2n)!(k+2n)!}{(k!)^3(n!)^4}v^kw^n\\
    =&\sum\limits_{n=0}^\infty {_3}F_2\left(\frac12+n,1+n,1+2n;1,1;4v\right)\frac{[(2n)!]^2}{(n!)^4}w^n\,.
    \end{split}
\label{eqn:varpiprime}
\end{align}

Deforming $p_1$ to $p_1'$ can be interpreted as adding the origin to $\nabla_1'$, obtaining $\widetilde{\nabla}_1$, and therefore performing a conifold transition from the mirror of $\Xd^{\text{n.c.}}$ to the mirror of $\widehat{\Xdr}$.
As a result,~\eqref{eqn:varpiprime} is nothing but the fundamental period associated to the mirror of $\widehat{\Xdr}$.

The fundamental period $\varpi_0(z)$ of the mirror of $\Xdnc$ is now obtained by taking the limit $a_0\rightarrow 0$, leaving $z'=w/v$ finite, such that
\begin{align}
\begin{split}
    \varpi_0(z')=&\lim\limits_{u\rightarrow 0}\frac{1}{-i \sqrt{u}}\varpi_0'\left(\frac{z'}{u},\frac{1}{u}\right)\\
    =&\frac12 {_4}F_3\left(\frac12,\frac12,\frac12,\frac12;1,1,1;-z'\right)\,,\quad z'=-\frac{a_3a_4a_7a_8}{a_1a_2a_5a_6}\,.
\end{split}
\label{eqn:varpi}
\end{align}
Note that we have divided by $-i\sqrt{u}$ in order to compensate for the factor of $a_0$ in the numerator of the integrand in~\eqref{eqn:funddef}.
This exactly matches the result~\eqref{eqn:locresult1} for the fundamental period from the localization calculation in Section~\ref{sec:glsm}.
After rescaling by $2$ and changing to the coordinate $z=-2^{-8}z'$ we obtain the result~\ref{eqn:fundamentalPeriod}.

\subsection{$\vec{d}=(5,1^3),\,n_s=64$}
We start again by constructing a map $T:\,\mathbb{Z}^7\rightarrow M'$ that restricts to an isomorphism on $\widetilde{M}$~\eqref{eqn:mlattice}.
As a matrix we can choose this to take the form
\begin{align}
    T=\left(
\begin{array}{ccccccc}
 2 & 0 & 0 & 0 & 0 & 1 & 0 \\
 0 & 1 & 0 & 0 & 0 & 1 & 0 \\
 0 & 0 & 1 & 0 & 0 & 1 & 0 \\
 0 & 0 & 0 & 1 & 0 & 0 & 1 \\
 1 & 0 & 0 & 0 & 0 & 1 & 0 \\
 -1 & 2 & 0 & 0 & 1 & 1 & 0 \\
 0 & 0 & 2 & 0 & -1 & 0 & 1 \\
 0 & 0 & 0 & 2 & 0 & 0 & 1 \\
\end{array}
    \right)\,.
    \label{eqn:gextmatrix2}
\end{align}
We denote by $T'$ the full rank matrix that is obtained by adding the column $(0^7,1)$ to $T$.
We can then use $(T')^{-1}$ to translate the generators~\eqref{eqn:gorensteinGens} into $\mathbb{Z}^7\oplus\mathbb{Z}$ and drop the last component to obtain the image in $\mathbb{Z}^7$.

We can choose the subsets $I_1,I_2$ in~\eqref{eqn:tdec} such that
\begin{align}
    t_1=e_1^*+e_2^*+e_3^*+f_1^*+f_2^*\,,\quad t_2=e_4^*+f_3^*+f_4^*\,.
    \label{eqn:gextdec2}
\end{align}
Again these correspond to the columns 6,7 in~\eqref{eqn:gextmatrix2} such that we can directly identify $\mathbb{Z}^7=\mathbb{Z}^5\oplus\mathbb{Z}^2=M\oplus\mathbb{Z}^2$.

To obtain the polytope $\nabla^{\circ}$ associated to the dual ``nef partition'', we therefore drop the last two components, amounting to the projection from $\widetilde{M}$ to $M$.
The resulting matrix $\left[(T')^{-1}\right]_{\text{res.}}$, that maps a vector in $M'$ to a corresponding vector in $\mathbb{Z}^5$, takes the form
\begin{align}
    \left[(T')^{-1}\right]_{\text{res.}}=\left(
\begin{array}{cccccccc}
 \frac{1}{2} & -\frac{1}{2} & \frac{1}{2} & -\frac{1}{2} & 0 & 0 & 0 & 0 \\
 1 & 0 & 0 & 0 & -1 & 0 & 0 & 0 \\
 0 & 0 & 1 & -1 & 0 & 0 & 0 & 0 \\
 1 & 1 & 1 & 0 & -1 & -1 & -1 & 0 \\
 -1 & -1 & 1 & -1 & 1 & 1 & 0 & 0 \\
\end{array}
\right)\,.
\end{align}
Applying this to the generators~\eqref{eqn:gorensteinGens}, and taking into account the decomposition~\eqref{eqn:gextdec2}, we then find the polytopes $\nabla^{\circ}=\text{Conv}(\Delta_1,\Delta_2)$ with $\Delta_1,\Delta_2\subset\mathbb{Z}^5$ being as follows:
\begin{align}
\Delta_1=\text{Conv}\left(
\begin{array}{ccccccccc}
 3 & 1 & -2 & 0 & -2 & 0 & -2 & 0 & 0 \\
 1 & 1 & 1 & 1 & -4 & 0 & -4 & 0 & 0 \\
 1 & 1 & -4 & 0 & 1 & 1 & -4 & 0 & 0 \\
 0 & 0 & 0 & 0 & 0 & 0 & -5 & -1 & 0 \\
 2 & 2 & -8 & 0 & 2 & 2 & 2 & 2 & 0 \\
\end{array}
\right)
\label{eqn:5111delta1}
\end{align}
\begin{align}
\Delta_2=\text{Conv}\left(
\begin{array}{ccccccccc}
 1 & 1 & 0 & 0 & 0 & 0 & 0 & 0 & 0 \\
 1 & 1 & 1 & 1 & 0 & 0 & 0 & 0 & 0 \\
 1 & 1 & 0 & 0 & 1 & 1 & 0 & 0 & 0 \\
 0 & 2 & 0 & 2 & 0 & 2 & -1 & 1 & 0 \\
 0 & 0 & -2 & -2 & 0 & 0 & 0 & 0 & 0 \\
\end{array}
\right)
\label{eqn:5111delta2}
\end{align}
One can check using SageMath~\cite{sagemath} that $\nabla^{\circ}$ is reflexive.

We find also find the polytopes
\begin{align}
    \nabla_1=\left(
\begin{array}{ccccc}
 2 & 0 & 0 & 1 & -1 \\
 -1 & 0 & -1 & -1 & 1 \\
 0 & 0 & 1 & 0 & 0 \\
 0 & 0 & 0 & 0 & 0 \\
 -\frac{1}{2} & -\frac{1}{2} & -\frac{1}{2} & -\frac{1}{2} & \frac{1}{2} \\
\end{array}
\right)\,,\quad \nabla_2=\left(
\begin{array}{ccc}
 0 & 0 & 0 \\
 0 & 0 & 0 \\
 -1 & 1 & -1 \\
 0 & -1 & 1 \\
 \frac{1}{2} & -\frac{1}{2} & \frac{1}{2} \\
\end{array}
\right)\,,
\end{align}
and choose the shift such that
\begin{align}
\nabla_1'=\left(
\begin{array}{ccccc}
 2 & 0 & 0 & 1 & -1 \\
 -1 & 0 & -1 & -1 & 1 \\
 1 & 1 & 2 & 1 & 1 \\
 -1 & -1 & -1 & -1 & -1 \\
 -1 & -1 & -1 & -1 & 0 \\
\end{array}
\right)\,,\quad \nabla_2'=\left(
\begin{array}{ccc}
 0 & 0 & 0 \\
 0 & 0 & 0 \\
 -2 & 0 & -2 \\
 1 & 0 & 2 \\
 1 & 0 & 1 \\
\end{array}
\right)\,.
\end{align}
We have choosen the shift such that the nef partition $\widetilde{\Delta}^\circ=\text{Conv}(\nabla_1',\nabla_2')$ and $\widetilde{\nabla}=\text{Conv}(\nabla_1',0)+\nabla_2'$ again corresponds to $\widehat{\Xdr}$ and its mirror.

The defining polynomials on the algebraic torus take form
\begin{align}
\begin{split}
    p_1=&\frac{a_1 x_3 x_1^2}{x_2 x_4 x_5}+\frac{a_4 x_3 x_1}{x_2 x_4 x_5}+\frac{a_3 x_3^2}{x_2
   x_4 x_5}+\frac{a_2 x_3}{x_4 x_5}+\frac{a_5 x_2 x_3}{x_4 x_1}\,,\\
    p_2=&a_7+\frac{a_8 x_5 x_4^2}{x_3^2}+\frac{a_6 x_5 x_4}{x_3^2}\,,
\end{split}
\end{align}
and introducing again $\tilde{p}_1=a_0+p_1$ we calculate
\begin{align}
    \varpi_0'=\int\limits_{\Gamma_0}\frac{a_0\cdots a_8}{\tilde{p}_1p_2}\prod\limits_{i=1}^5\frac{dx_i}{x_i}\,.
    \label{eqn:funddef2}
\end{align}
In terms of the invariant coordinates
\begin{align}
    w=\frac{a_1 a_2 a_3 a_6}{a_0 a_4^2 a_7}\,,\quad v=-\frac{a_4 a_5 a_8}{a_0^2 a_7}\,,
\end{align}
this leads to the expansion
\begin{align}
\begin{split}
    \varpi_0'(v,w)=&\sum\limits_{n,k=0}^\infty\frac{(2k+n)!(k+n)!}{(k!)^2(k-2n)!(n!)^4}v^kw^n\\
        =&\sum\limits_{n=0}^\infty\frac{{_3}F_2\left(\frac12+\frac{k}{2},1+\frac{k}{2},1+k;1,1-2k;4v\right)}{\Gamma(1-2k)(k!)^2}w^n\,.
    \end{split}
\end{align}
Taking the limit
\begin{align}
\begin{split}
    \varpi_0(z')=&\lim\limits_{u\rightarrow 0}\frac{1}{-i \sqrt{u}}\varpi_0'\left(\frac{z'}{u},\frac{1}{u}\right)\,,\qquad z'=\frac{a_1 a_2 a_3 a_6}{\sqrt{a_4^5 a_5 a_7 a_8}}\\
    =&\frac12 {_4}F_3\left(\frac{1}{10},\frac{3}{10},\frac{7}{10},\frac{9}{10};1,1,1;5^52^{-8}z'^2\right)\,,
\end{split}
\end{align}
and rescaling by 2 and using the coordinate $z=-2^{-8}z'$ we reproduce~\ref{eqn:fundamentalPeriod}.

\section{Mirror periods, smooth cousins and topological B-branes}
\label{sec:branes}
\begin{table}[t!]
\centering
\begin{align*}
\begin{array}{|c|c|c|c|}\hline
\Xdnc&\text{Mirror Picard-Fuchs operator } \mathcal{D}_{\vec{d}}&\Delta(z)&\text{Smooth cousin } \widetilde{X}_{\vec{d}}\\\hline
\Xnc{(5,1^3)}&\theta ^4-2^85 z^2(5 \theta +1) (5 \theta +3) (5 \theta +7) (5 \theta +9)&1-2^85^5z^2&X_{10}\subset\mathbb{P}^4(1^3,2,5)\\
\Xnc{(4,2^2)}&\theta ^4-2^4 z (4 \theta +1)(2 \theta +1)^2 (4 \theta +3)&1-2^{10}z&X_{4,2}\subset\mathbb{P}^5\\
\Xnc{(3^2,1^2)}&\theta ^4-2^83^2z^2 (3 \theta +1)^2 (3 \theta +5)^2&1-2^83^6z^2&X_{6,6}\subset\mathbb{P}^5(1^2,2^2,3^2)\\
\Xnc{(3,1^5)}&\theta ^4-2^83 z^2 (\theta +1)^2 (3 \theta +1) (3 \theta +5)&1-2^83^3z^2&X_{6,2}\subset\mathbb{P}^5(1^5,3)\\
\Xnc{(2^4)}&\theta ^4-2^4z (2 \theta +1)^4&1-2^8z&X_{2,2,2,2}\subset\mathbb{P}^7\\
\Xnc{(1^8)}&\theta ^4-2^8z^2 (\theta +1)^4&1-2^8z^2&X_{2,2,2,2}\subset\mathbb{P}^7\\
\hline
\end{array}
\end{align*}
\caption{The Picard-Fuchs operators that annihilate the mirror periods of $\Xdnc$ and the corresponding discriminant polynomials $\Delta(z)$. We also indicate the smooth complete intersection Calabi-Yau 3-folds that share the same mirror variations of Hodge structure over $\mathbb{Q}$ or one that is related by a 2:1 covering.}
\label{tab:operators}
\end{table}
Using the general expression~\eqref{eqn:fundamentalPeriod} we obtain the fundamental period of the mirror $Y_{\vec{d}}$ of $\Xd^{\text{n.c.}}$ for each of the determinantal octic double solids.
It turns out that the period is always hypergeometric.

As a result, the fundamental period is -- up to a change of coordinate for odd decompositions -- equivalent to the fundamental period of the mirror $\widetilde{Y}_{\vec{d}}$ of a Calabi-Yau complete intersection $\widetilde{X}_{\vec{d}}$ in a weighted projective space~\cite{Hosono:1993qy,Hosono:1994ax,Doran:2005gu}.
We refer to $\widetilde{X}_{\vec{d}}$ as the ``smooth cousin'' of $\Xdnc$.
The Picard-Fuchs operator associated to each $\Xdnc$ as well as the corresponding smooth cousins are listed in Table~\ref{tab:operators}.

For the case $X_{(1^8)}^{\text{n.c.}}$ it turns out that $\widetilde{X}_{(1^8)}$ is a complete intersection of four quadrics $X_{2,2,2,2}$ in $\mathbb{P}^7$ and there is a twisted derived equivalence between $X_{(1^8)}^{\text{n.c.}}$ and $\widetilde{X}_{(1^8)}$~\cite{Kuznetsov2008,addington2009derived,Katz:2022lyl}}.
Physically, this is a consequence of the fact that the two backgrounds correspond to two different large volume limits in the same stringy K\"ahler moduli space, as can be seen from the GLSM analysis in~\cite{Caldararu:2010ljp}. 
In particular, they share the same mirror Calabi-Yau and are so-called Clifford double mirrors~\cite{borisov2016clifford}.

However, for each of the other determinantal octic double solids we find that the mirrors of $\Xdnc$ and $\widetilde{X}_{\vec{d}}$ are different and in the stringy K\"ahler moduli space of $\Xdnc$ there is no additional large volume limit associated to $\widetilde{X}_{\vec{d}}$.

The fact that different Calabi-Yau manifolds can share the same variation of Hodge structure, at least up to a non-integral change of basis in $\text{Sp}(4,\mathbb{Q})$, was already observed in~\cite{Aspinwall:1994uj}, with a prototypical example being the quintic in $\mathbb{P}^4$ and its $\mathbb{Z}_5$ quotient.

We will now discuss the integral basis of periods that corresponds to the central charges of a set of topological B-branes that generate the charge lattice.
It turns out that this is always different from an integral basis of periods on the mirror of the smooth cousin.

For any decomposition $\vec{d}$, a complete set of solutions that is annihilated by the corresponding Picard-Fuchs operator in Table~\ref{tab:operators} takes the form $\varpi_i,\,i=0,\ldots,3$ with $\varpi_i(z)=\log(z)^i+\mathcal{O}(z)$.
The parameter $z$ is choosen such that the flat coordinate $t(z)=\frac{1}{2\pi i}\log(z)+\mathcal{O}(z)$ is the complexified volume of a curve of degree $1$ in $\widehat{\Xd}$.
In particular, $z\rightarrow 0$ corresponds to the large volume limit of $\Xdnc$.

Using an integral basis of central charges of topological B-branes for $X_{2,2,2,2}$ and analytic continuation, it was found in~\cite{Katz:2022lyl} that the central charges of a corresponding basis on $X_{(1^8)}^{\text{n.c.}}$ take the form
\begin{align}
        \begin{split}
                 \Pi^{(6)}_{(1^8)}=&\frac{\kappa}{3}t^3-\frac{1}{12}t+2\frac{1}{(2\pi i)^3}\zeta(3)\left(\chi+\frac{7}{2}n_{(1^8)}\right)+\mathcal{O}(e^{2\pi i t})\,,\\
                \Pi^{(4)}_{(1^8)}=&t^2-\frac12t+\frac{1}{24}+\mathcal{O}(e^{2\pi i t})\,,\\
                \Pi^{(2)}_{(1^8)}=&t\,,\quad\Pi^{(0)}_{(1^8)}=\frac12\,,
        \end{split}
        \label{eqn:1to8central}
\end{align}
in terms of the topological invariants of the smooth deformation $X_{(8)}$ of $X_{(1^8)}$
\begin{align}
	\kappa=\int_{X_{(8)}}J^3=2\,,\quad c_2= \int_{X_{(8)}}J\cdot c_2(TX_{(8)})=44\,,\quad \chi=-296\,,
 \label{eqn:invx8}
\end{align}
where $J$ is a primitive ample divisor, and $n_{(1^8)}=84$ is the number of nodes of $X_{(1^8)}$.

Recall that a 0-brane supported at a smooth point $p$ of a Calabi-Yau 3-fold can be represented by a skyscraper sheaf $\mathcal{O}_p$ and has central charge $Z(\mathcal{O}_p)=1$.
From $\Pi^{(0)}_{(1^8)}=\frac12$ in~\eqref{eqn:1to8central} one can deduce that the corresponding direction of the 0-brane in the charge lattice of $X_{(1^8)}^{\text{n.c.}}$ is generated by a 2-brane that wraps one of the 2-torsional exceptional curves in $\widehat{\Xd}$. 

The charge lattice of topological B-branes on $X_{(1^8)}^{\text{n.c.}}$ can be identified with the twisted K-theory $K_0^\alpha(\widehat{X}_{(1^8)})$, and the exponentiated central charge is a homomorphism from the charge lattice to $\mathbb{C}^\times$.
This directly implies that the reduced twisted K-theory takes the form $K_{(2)}^\alpha(\widehat{X}_{(1^8)})\simeq \mathbb{Z}^2$, with the brane wrapping the 2-torsional exceptional curve and the 0-brane respectively corresponding to elements $(0,1),(0,2)$, as discussed in Section~\ref{sec:derived}.

To obtain a 2-brane with central charge $\Pi^{(2)}_{(1^8)}=t$, and therefore the second generator of $K_{(2)}^\alpha(\widehat{X}_{(1^8)})$, one can consider a line in $\mathbb{P}^3$ that has four tangencies with the ramification locus $S_{(1^8)}\subset\mathbb{P}^3$.
The corresponding double cover in $\widehat{X}_{(1^8)}$ is split and, as will be further discussed in Section~\ref{sec:enumerativeGeometry}, by wrapping 2-branes on either of the two sheets one obtains respective elements $(1,0),(1,1)\in K_{(2)}^\alpha(\widehat{X}_{(1^8)})$.

The interpretation of $\Pi^{(4)}_{(1^8)},\Pi^{(6)}_{(1^8)}$ is slightly more speculative.
The generator of the direction of the 6-brane in $K_0^\alpha(\widehat{X}_{(1^8)})$ has central charge $\Pi^{(6)}_{(1^8)}=\frac{\kappa}{3}t^3+\ldots$ while the central charge of the 6-brane $\mathcal{O}_{X_{(8)}}$ on the smooth deformation $X_{(8)}$ takes the form $\Pi^{(6)}_{(8)}=\frac{\kappa}{6}t^3+\ldots$.
It is tempting to interpret this discrepancy as an effect of the topologically non-trivial flat B-field in the following way.

Recall that on a smooth Calabi-Yau the 6-brane is conjecturally related under homological mirror symmetry to a 3-brane that wraps the section of the SYZ-fibration on the mirror~\cite{Strominger:1996it}.
As was discussed in~\cite[Section~5]{Donagi2008}, the presence of the B-field should correspond on the mirror Calabi-Yau to the absence of a section of the SYZ-fibration.
Instead, if the B-field corresponds to a $2$-torsional Brauer class, the mirror should only have a $2$-section, i.e. a 3-cycle that intersects the lagrangian fibers twice while the intersection points experience monodromy upon moving around a ramification locus in the base.
It is then natural to expect that the primitive element of the twisted K-theory charge lattice is not a 6-brane that wraps the Calabi-Yau but one that wraps it twice.

There is no such constraint on the 4-branes and one finds that to leading order $\Pi^{(4)}_{(1^8)}=t^2+\ldots$ agrees with the central charge of a 4-brane $\mathcal{O}_J$ wrapping a primitive ample divisor $J$ in $X_{(8)}$.

For a general determinantal octic double solid $\Xd$ we still have $K_{(2)}^\alpha(\widehat{\Xd})\simeq \mathbb{Z}^2$ and since the central charge of a zero brane can be normalized to be $Z(\mathrm D0)=1$, we find that the generator of the corresponding direction of the charge lattice can again be represented by a brane wrapping one of the 2-torsional exceptional curves and has central charge $\Pi^{(0)}_{\vec{d}}=1/2$.
The additional generator of $K_{(2)}^\alpha(\widehat{\Xd})$ with central charge $\Pi^{(2)}_{\vec{d}}=t$ can again be obtained by wrapping a 2-brane on one of the sheets of a split double cover in $\widehat{\Xd}$ of a suitable line in $\mathbb{P}^3$.
As a consequence, the integral variation of Hodge structure of the mirror $Y_{\vec{d}}$ of $\Xdnc$ always differs from that of the mirror $\widehat{Y}_{\vec{d}}$ of the smooth cousin $\widehat{X}_{\vec{d}}$~\footnote{Let us point out that the mirror Calabi-Yau manifolds of $X^{\text{n.c.}}_{(2^4)}$ and $X^{\text{n.c.}}_{(4,2^2)}$ appear to violate Conjecture~2.9 of~\cite{Doran:2005gu}.}.

The arguments explaining the leading behaviour of $\Pi^{(4)}_{(1^8)},\Pi^{(6)}_{(1^8)}$ also apply to the other determinantal octic double solids $\Xd$.
Combining this with preliminary results from numerical analytic continuation, we expect that an integral basis of topological B-branes on $\Xdnc$ can always be choosen such that the leading behaviour of the central charges is given by
\begin{align}
        \begin{split}
                 \Pi^{(6)}_{\vec{d}}=&\frac{\kappa}{3}t^3+p^{(2)}_{\vec{d}}(t)+2\frac{1}{(2\pi i)^3}\zeta(3)\left(\chi+\frac{7}{2}n_{\vec{d}}\right)+\mathcal{O}(e^{2\pi i t})\,,\\
                \Pi^{(4)}_{\vec{d}}=&\frac{\kappa}{2}t^2+p^{(1)}_{\vec{d}}(t)+\mathcal{O}(e^{2\pi i t})\,,\\
                \Pi^{(2)}_{\vec{d}}=&t\,,\quad\Pi^{(0)}_{\vec{d}}=\frac12\,,
        \end{split}
        \label{eqn:dcentral}
\end{align}
where $p^{(i)}_{\vec{d}}(t)\in \mathbb{Q}[t]$ are inhomogeneous polynomials of degree $i$ with rational coefficients.

Apart from the case $\vec{d}=(1^8)$, where the integral basis~\eqref{eqn:1to8central} was determined in~\cite{Katz:2022lyl}, there is one other example for that we can completely fix the asymptotic behaviour~\eqref{eqn:dcentral}.
For the case $\vec{d}=(2^4)$ there is a $\mathbb{Z}_2$ symmetry in the moduli space, corresponding to the fact that the phases $r\gg0$ and $r\ll 0$ in the associated GLSM~\eqref{eqn:glsmCharges} are identical.
We can therefore choose the same basis of periods around $z\rightarrow 0$ and $z\rightarrow \infty$ and the transfer matrix $T$ that one obtains by (numerical) analytic continuation from one point to the other is actually a monodromy~\footnote{The potential for such a symmetry in the moduli space of the smooth cousin $X_{2,2,2,2}$ of  $X^{\text{n.c.}}_{(2^4)}$ was already observed in~\cite{Joshi:2019nzi}. However, for $X_{2,2,2,2}$ the second large volume limit is actually associated to $X^{\text{n.c.}}_{(1^8)}$ and the symmetry is broken~\cite{Katz:2022lyl}.
The moduli space is only symmetric when both large volume limits are associated to $X^{\text{n.c.}}_{(2^4)}$.
}.

Demanding integrality of this matrix allows us to fix
\begin{align}
    p^{(2)}_{(2^4)}=\frac23 t\,,\quad p^{(1)}_{(2^4)}=-\frac13\,,
\end{align}
up to an integral change of basis.
The transfer matrix $T$ and the monodromies $M_{0},M_c$ around the large volume limit $z=0$ and the conifold point $z=1/256$ then take the form
\begin{align}
{\tiny
T=\left(
\begin{array}{cccc}
 -1 & 0 & 0 & 0 \\
 0 & 1 & 0 & 0 \\
 -1 & 0 & 1 & 0 \\
 -1 & -1 & 0 & -1 \\
\end{array}\right)\,,\quad
M_{0}=\left(
\begin{array}{cccc}
 1 & 2 & 2 & 4 \\
 0 & 1 & 2 & 2 \\
 0 & 0 & 1 & 2 \\
 0 & 0 & 0 & 1 \\
\end{array}\right)\,,\quad
M_{c}=\left(
\begin{array}{cccc}
 1 & 0 & 0 & 0 \\
 0 & 1 & 0 & 0 \\
 0 & 0 & 1 & 0 \\
 -2 & 0 & 0 & 1 \\
\end{array}\right)\,.
}
\label{eqn:2222mon}
\end{align}
They satisfy the relations
\begin{align}
    T^2=M_c^{-1}\,,\quad (TM_0^{-1})^2=-1\,\quad [M_c,T]=0\,.
\end{align}

To understand the general form of the central charges $\Pi^{(4)}_{\vec{d}},\Pi^{(6)}_{\vec{d}}$ will likely require both a better understanding of $D^b(\widehat{\Xd},\alpha)\simeq D^b(\mathbb{P}^3,\mathcal{B}_0)$ as well as a modification of the Gamma-class formula from~\cite{Iritani2009,Halverson:2013qca}.
This is beyond the scope of the current paper but we expect that the techniques developed in~\cite{Knapp:2020oba,Erkinger:2022sqs} could be applied.
Similarly, we leave a discussion of the geometric interpretation underlying the monodromies~\eqref{eqn:2222mon} and of the corresponding derived equivalences to future work.

\section{Torsion refined GV-invariants from mirror symmetry}\label{sec:direct}
We now use mirror symmetry~\cite{Candelas:1990rm} and direct integration of the holomorphic anomaly equations~\cite{Bershadsky:1993ta,Bershadsky:1993cx}, relating the topological string free energies $F_g$ at different genera, in order to calculate $\mathbb{Z}_2$ torsion refined Gopakumar-Vafa invariants for the determinantal octic double solids.

The direct integration procedure for smooth Calabi-Yau threefolds has first been described in~\cite{Huang:2006hq} and the necessary modifications to obtain the torsion refined invariants for genus $g\ge 2$ has recently been worked out in~\cite{Katz:2022lyl}, building on the earlier calculations of torsion refined GV-invariants at genus 0 and 1 from~\cite{Schimannek:2021pau}.
We will focus on the examples $X^{\text{n.c.}}_{(5,1^3)},\,X^{\text{n.c.}}_{(2^4)}$ and $X^{\text{n.c.}}_{(4,2^2)}$ where, as will be further discussed below, we are able to fix the holomorphic ambiguity respectively up to genus $25,\,14$ and $9$.

Instead of repeating the general discussion of the topological string B-model and the holomorphic anomaly equations from~\cite{Katz:2022lyl}, we will outline the concrete steps of the calculation for a general $\Xdnc$ and then discuss the results for concrete choices of $\vec{d}$.

\subsection{Special geometry and direct integration}
We first fix a normalized decomposition $\vec{d}\in\mathbb{N}^k$ of degree $8$ and length $\ge 3$.
The number $n_{\vec{d}}$ of nodes in $\Xd$ can be calculated using Lemma~\ref{lem:nodes} and the greatest common divisor of $\vec{d}$ will again be denoted by $q=\text{gcd}(d_1,\ldots, d_k)$.

The Picard-Fuchs operator corresponding to the mirror of $\Xdnc$ can be read off from Table~\ref{tab:operators} and annihilates a basis of mirror periods $\varpi_{i=0,\ldots,3}$ with leading behaviour
\begin{align}
    \varpi_0(z)=1+\mathcal{O}(z)\,,\quad \varpi_{i=1,\ldots,3}(z)=\varpi_0(z)\log(z)^i+\mathcal{O}(z)\,.
\end{align}
The mirror map is given by $t(z)=(2\pi i)^{-1}\varpi_1/\varpi_0$ and using $q=e^{2\pi i t}$ one can invert the map to obtain $z(q)=q+\mathcal{O}(q^2)$.

Using a K\"ahler transformation such that $\varpi_0\rightarrow 1$ and expressing the normalized periods in terms of the flat coordinate $t$ we obtain a symplectic basis $\Pi=(1,\,t,\,2\mathcal{F}_0-t\partial_t\mathcal{F}_0,\,\partial_t\mathcal{F}_0)$~\footnote{Note that, as was discussed in~\cite{Katz:2022lyl} for $X^{\text{n.c.}}_{(1^8)}$ and can be seen more generally from~\eqref{eqn:dcentral}, $\Pi$ does not correspond to the central charges of an integral basis of topological branes anymore.}, with the leading terms of the prepotential $\mathcal{F}_0$ given by~\cite{Schimannek:2021pau,Katz:2022lyl}
\begin{align}
    \begin{split}
        \mathcal{F}_0=-\frac16\kappa t^3+\frac{c_2}{24}t+\frac{\zeta(3)}{(2\pi i)^3}\left(\frac{\chi}{2}+\frac74n_{\vec{d}}\right)+\mathcal{O}(q)\,,
    \end{split}
\end{align}
in terms of the topological invariants~\eqref{eqn:invx8} of the smooth deformation $X_{(8)}$ of $X_{\vec{d}}$.
The Yukawa coupling takes the form
\begin{align}
    C_{zzz}=\frac{\kappa}{z^3\Delta(z)}\,,
\end{align}
where $\Delta(z)$ is the discriminant polynomial that is listed together with the Picard-Fuchs operator in Table~\ref{tab:operators}.

It now becomes important to distinguish between the topological string free energies $\mathcal{F}_g(t,\bar{t})$, their holomorphic limit $\mathcal{F}_g(t)=\lim_{\bar{t}\rightarrow i\infty} \mathcal{F}(t,\bar{t})$ and what we will call the genus $g$ generating functions $F_g(z,\bar{z})=\varpi_0^{2-2g}\mathcal{F}(t,\bar{t})$.
The latter transform as sections of $\mathcal{L}^{2-2g}$, where $\mathcal{L}$ is the K\"ahler line bundle on the moduli space, with covariant derivatives acting as $D_i=\partial_z-(2-2g)\partial_z K$ in terms of the K\"ahler potential $K$.

The genus one generating function takes the form
\begin{align}
\begin{split}
    F_1(z,\bar{z})=&-\frac12\left(4-\frac{\chi+2n_{\vec{d}}}{12}\right)K-\frac12\log\,\det\,G_{\bar{z}z}\\
    &-\frac{1}{24}\left(12+c_2\right)\log(z)-\frac{1}{12}\log\,\Delta(z)\,,
    \end{split}
    \label{eqn:Fg1}
\end{align}
where $G_{\bar{z}z}$ is the Weil-Petersson metric.

The holomorphic anomaly equations for genus $g\ge 2$~\cite{Bershadsky:1993ta,Bershadsky:1993cx} can be written as~\cite{Yamaguchi:2004bt,Alim:2007qj,Hosono:2008np}
\begin{align}
    \begin{split}
    &\frac{\partial F_g}{\partial \tilde{S}^{zz}}-K_z\,\frac{\partial F_g}{\partial \tilde{S}^{z}}+\frac12\,K_z K_z\frac{\partial F_g}{\partial \tilde{S}}\\
    =&\frac12\, D^2 F_{g-1}+\frac12 \sum\limits_{h=1}^{g-1}D F_{g-h} D F_h\,,\quad g\ge 2\,,
    \end{split}
    \label{eqn:haneq}
\end{align}
in terms of the shifted propagators $\tilde{S}^{zz},\tilde{S}^z,\tilde{S}$.
The latter are partially determined by the equations of the modified BCOV ring
\begin{align}
\begin{split}
    \partial_z \tilde{S}^{zz}=&\,C_{zzz}\tilde{S}^{zz}\tilde{S}^{zz}+2\tilde{S}^z-2s^z_{zz}\tilde{S}^{zz}+h^{zz}_z\,,
    \\
    \partial_z \tilde{S}^z=&\,C_{zzz}\tilde{S}^{zz}\tilde{S}^z+2 \tilde{S}-s^z_{zz}\tilde{S}^z-h_{zz}\tilde{S}^{zz}+h^z_z\,,
    \\
    \partial_z \tilde{S}=&\,\frac12\, C_{zzz}\tilde{S}^z\tilde{S}^z-h_{zz}\tilde{S}^z+h_z\,,
    \\
    \partial_z K_z=&\,K_zK_z-C_{zzz}\tilde{S}^{zz}+s^z_{zz}K_z-C_{zzz}\tilde{S}^z+h_{zz}\,,
    \label{eqn:bcov}
\end{split}
\end{align}
together with the expression for the Christoffel symbols associated to the Weil-Petersson metric 
\begin{align}
        \Gamma^z_{zz}=2K_z-C_{zzz}\tilde{S}^{zz}+s^z_{zz}\,,
        \label{eqn:christoffel}
\end{align}
up to holomorphic propagator ambiguities $s^z_{zz},h^{zz}_z,h^z_z,h_z,h_{zz}$.
When expressed in terms of the shifted propagators, the generating functions $F_{g\ge 2}$ are polynomials in $\tilde{S}^{zz},\tilde{S}^z,\tilde{S}$ with coefficients that are rational functions in $z$.
In particular, the anti-holomorphic dependence as well as the dependence on the K\"ahler potential is completely absorbed by the propagators.

The integrated holomorphic anomaly equation at genus $g=1$, that is satisfied by~\eqref{eqn:Fg1}, takes the form
\begin{align}
    \partial_z F_1=&\frac12\, C_{zzz}\tilde{S}^{zz}-\left(\frac{\chi+2n_{\vec{d}}}{24}-1\right)K_z+f^{(1)}(z)\,,
\end{align}
where $f^{(1)}(z)$ is the corresponding holomorphic ambiguity.
On the other hand, equations~\eqref{eqn:haneq} can be integrated -- using the independence of $F_{g\ge 2}$ from $K$ -- in order to obtain $F_g$ from the generating functions at genus $g'<g$ up to a holomorphic ambiguity $f^{(g)}(z)$.
The holomorphic ambiguities at $g\ge 2$ take the form
\begin{align}
	f^{(g)}(z)=\frac{1}{\Delta^{2g-2}}\sum\limits_{k=0}^{2g-2}f_kz^{qk}+\sum\limits_{k=1}^{N(g)}f_k'z^{qk}\,,\quad N(g)=\left\lfloor\frac{2(g-1)}{\rho_{\vec{d}}}\right\rfloor\,,
        \label{eqn:holomorphicAmbiguity}
\end{align}
where $\rho_{\vec{d}}$ is a model dependent regulator that depends on the behaviour of the generating functions at $z=\infty$.
We will determine the regulator for our examples empirically from the behaviour of the holomorphic ambiguties at low genera.
The minimal value for the regulator, if no additional regularity constraints from the point at infinity are known, is $\rho=2$.

In order to fix the holomorphic ambiguities, we use the following conjectural properties of the topological string free energies and the torsion refined Gopakumar-Vafa invariants:
\begin{enumerate}
    \item \textbf{Constant map contributions} Using the general results from~\cite[Section 3.6]{Katz:2022lyl}, the modified constant map contributions in the large volume limit of $\Xdnc$ take the form
    \begin{align}
    \begin{split}
                &\mathcal{F}_{g\ge 2}^{\text{const.}}(t,\bar{t})\\
                =&(-1)^{g-1}\frac{B_{2g}B_{2g-2}}{2g(2g-2)\left[(2g-2)!\right]}\left(\frac{\chi+2n_{\vec{d}}}{2}+(1-2^{2g-2})n_{\vec{d}}\right)\,,
        \end{split}
    \end{align}
    where $B_n$ is the $n$-th Bernoulli number.
    \item \textbf{Gap condition} The gap condition from~\cite{Huang:2006hq} holds without modification, and the leading behaviour of the free energies in a suitably normalized flat coordinate $t_c$ around the conifold points $\Delta=0$ takes the form
    \begin{align}
        \mathcal{F}_{g\ge 2}(t_c)=\frac{(-1)^{g-1}B_{2g}}{2g(2g-2)t_c^{2g-2}}+\mathcal{O}(t_c^0)\,.
    \end{align}
    \item \textbf{Castelnuovo vanishing} The torsion refined Gopakumar-Vafa invariants satisfy the same Castelnuovo vanishing as the corresponding unrefined invariants of the smooth deformation.
    In particular, they satisfy the bound
    \begin{align}
        n_g^{d,0}=n_g^{d,1}=0\text{  for  } g> g_{\text{max.}}(d)=\left\{\begin{array}{cl}
        \frac{d^2}{4}+\frac{d}{2}+1&,\,d\text{ even}\\
        \left\lfloor\frac{d^2}{4}+\frac{d}{2}\right\rfloor&,\,d\text{ odd}
        \end{array}
        \right.\,.
        \label{eqn:castelnuovoBound}
    \end{align}
\end{enumerate}

Let us stress again that the torsion refined invariants can only be extracted by combining the information from the topological string free energies associated to $\Xdnc$ and to the smooth deformation $X_{(8)}$ of $\Xd$.
However, the direct integration for $X_{(8)}$ and $\Xnc{(1^8)}$ has been carried out in~\cite{Katz:2022lyl} up to genus $g=32$ while the maximal genus for $X_{(8)}$ was pushed to $g=48$~\cite{Alexandrov:2023zjb}.

Note that the existence of Castelnuovo vanishing for Gopakumar-Vafa invariants has been proven in~\cite{doan2021gopakumarvafa} while the precise bound for $X_{(8)}$ follows from the proof of the corresponding bound for Pandharipande-Thomas invariants in~\cite{Alexandrov:2023zjb} together with the MNOP conjecture~\cite{maulik2004gromovwitten}~\footnote{See also~\cite{Liu:2022agh} for a recent proof in the context of the quintic.}.

In some cases we will also use the enumerative predictions for the torsion refined GV invariants that will be discussed in Section~\ref{sec:enumerativeGeometry}, generalizing earlier results from~\cite{Katz:2022lyl}, in order to fix additional free coefficients in the holomorphic ambiguity.
We will always highlight those numbers in the tables of torsion refined GV invariants in order to distinguish them from the genuine predictions that provide a non-trivial check of our results.

\subsection{$\vec{d}=(5,1^3),\,n_s=64$}
We choose the propagator ambiguities in~\eqref{eqn:bcov},\eqref{eqn:christoffel} to be
\begin{align}
    s^z_{zz}=-\frac{9}{5}\frac{1}{z}\,,\quad h^{zz}_z=\frac15z \,,\quad h_{zz}=\frac{3}{25}\frac{1}{z^2}\,,\quad h^z_z=0\,,\quad h_z=\frac{9}{2500}\frac{1}{z}\,,
\end{align}
and we find the regulator $\rho_{\vec{d}}$ in~\eqref{eqn:holomorphicAmbiguity} to be $\rho_{(5,1^3)}=10$.
From the results from~\cite{Katz:2022lyl} and the discussion in Section~\ref{sec:enumerativeGeometry} we obtain the prediction for the torsion refined Gopakumar-Vafa invariants
\begin{align}
    \begin{split}
    n_{g_{\text{max.}}}^{2b+1,0}=&n_{g_{\text{max.}}}^{2b+1,1}=\left\{\begin{array}{cl}
    14752&b=0\\
    (-1)^b14752(2b^2+4)&b\ge 1
    \end{array}\right.\,,\\[.3cm]
    n_{g_{\text{max.}}(2b)}^{2b,p}=&\left\{\begin{array}{cl}
        6&b=1\,,\quad p\equiv 1\\[.1cm]
        4(b^2-b+4)&b\ge 2\,,\quad p\equiv b\\[.1cm]
        0&\text{ else}
    \end{array}\right.\,,\\[.3cm]
    n_{g_{\text{max.}}(2b)-1}^{2b,p}=&\left\{\begin{array}{cl}
        386&b=1\,,\quad p\equiv 0\\[.1cm]
        480&b=1\,,\quad p\equiv 1\\[.1cm]
        -8 (b^4 - 60 b^2 + 67 b-189)&b\ge 2\,,\quad p\equiv b\\[.1cm]
        384(b^2-b+3)&b\ge 2\,,\quad p\equiv b+1
    \end{array}\right.\,.
    \end{split}
\end{align}
Note that $g_{\text{max.}}(d)$ is defined via the Castelnuovo bound in~\eqref{eqn:castelnuovoBound} and we always consider the $\mathbb{Z}_2$-charge -- or parity -- $p$ in $n_g^{d,p}$ as an element of $\mathbb{Z}/2\mathbb{Z}$.

Combining the generic formulas for the constant map contributions, the gap condition, the Castelnuovo vanishing for $X_{(8)}$ and the prediction $n^{8,1}_{21}=0$ we can fix the holomorphic ambiguities up to genus $g=25$.
Some of the resulting invariants are listed in Tables~\ref{tab:gv5111tA} and~\ref{tab:gv5111tB}.
\begin{table}[t!]
\begin{align*}
\tiny
\begin{array}{|c|ccccc|}\hline
n^{d,0}_g&d=1&2&3&4&5\\\hline
g=0&\gvHLc{14752}&64444512&711860273440&11596529531321056&233938237312624658400\\
1&\gvHLc{0}&20480&10732175296&902645866490432&50712027457008177856\\
2&\gvHLc{0}&\gvHLc{384}&-8275872&6249796276400&2700746768622436448\\
3&\gvHLc{0}&\gvHLc{0}&\gvHLc{-88512}&-87425677776&10292236849965248\\
4&\gvHLc{0}&\gvHLc{0}&\gvHLc{0}&198020184&-337281112359424\\
5&\gvHLc{0}&\gvHLc{0}&\gvHLc{0}&150666&6031964134528\\
6&\gvHLc{0}&\gvHLc{0}&\gvHLc{0}&\gvHLc{2232}&-43153905216\\
7&\gvHLc{0}&\gvHLc{0}&\gvHLc{0}&\gvHLc{24}&18764544\\
8&\gvHLc{0}&\gvHLc{0}&\gvHLc{0}&\gvHLc{0}&\gvHLc{177024}\\
9&\gvHLc{0}&\gvHLc{0}&\gvHLc{0}&\gvHLc{0}&\gvHLc{0}\\\hline\hline
n^{d,1}_g&d=1&2&3&4&5\\\hline
g=0&\gvHLc{14752}&64390400&711860273440&11596526493472256&233938237312624658400\\
1&\gvHLc{0}&20832&10732175296&902646226215424&50712027457008177856\\
2&\gvHLc{0}&\gvHLc{480}&-8275872&6249871001344&2700746768622436448\\
3&\gvHLc{0}&\gvHLc{6}&\gvHLc{-88512}&-87433826048&10292236849965248\\
4&\gvHLc{0}&\gvHLc{0}&\gvHLc{0}&198195616&-337281112359424\\
5&\gvHLc{0}&\gvHLc{0}&\gvHLc{0}&150784&6031964134528\\
6&\gvHLc{0}&\gvHLc{0}&\gvHLc{0}&\gvHLc{1920}&-43153905216\\
7&\gvHLc{0}&\gvHLc{0}&\gvHLc{0}&\gvHLc{0}&18764544\\
8&\gvHLc{0}&\gvHLc{0}&\gvHLc{0}&\gvHLc{0}&\gvHLc{177024}\\
9&\gvHLc{0}&\gvHLc{0}&\gvHLc{0}&\gvHLc{0}&\gvHLc{0}\\\hline
\end{array}
\end{align*}
\caption{Torsion refined Gopakumar-Vafa invariants for $X_{(5,1^3)}$ -- continued in Table~\ref{tab:gv5111tB}.
Invariants that are checked by the enumerative calculations from~\cite{Katz:2022lyl} and Section~\ref{sec:enumerativeGeometry} are highlighted in blue.
}
\label{tab:gv5111tA}
\end{table}

\begin{table}[t!]
\begin{align*}
\tiny
\begin{array}{|c|cccc|}\hline
n^{d,0}_g&d=6&7&8&9\\\hline
11&263936&1365366811756288&-4736147266849254718988&2362150336628273365296836096\\
12&\gvHLc{3456}&-17274516630240&241184142913118016336&-238002496862926514779851040\\
13&\gvHLc{0}&97442213760&-9638189985405718270&20760200170428536398007680\\
14&\gvHLc{0}&-33988608&289481011667110168&-1555728615943980103040320\\
15&\gvHLc{0}&\gvHLc{-324544}&-6124149316730056&98700288517404683773632\\
16&\gvHLc{0}&\gvHLc{0}&81786094440560&-5202996081724936520576\\
17&\gvHLc{0}&\gvHLc{0}&-550129902675&222451947587448686272\\
18&\gvHLc{0}&\gvHLc{0}&846264080&-7465750121080525472\\
19&\gvHLc{0}&\gvHLc{0}&371096&187682431656202624\\
20&\gvHLc{0}&\gvHLc{0}&\gvHLc{5000}&-3286547431924608\\
21&\gvHLc{0}&\gvHLc{0}&\gvHLc{64}&35286952874368\\
22&\gvHLc{0}&\gvHLc{0}&\gvHLc{0}&-174139266432\\
23&\gvHLc{0}&\gvHLc{0}&\gvHLc{0}&48799616\\
24&\gvHLc{0}&\gvHLc{0}&\gvHLc{0}&\gvHLc{531072}\\
25&\gvHLc{0}&\gvHLc{0}&\gvHLc{0}&\gvHLc{0}\\\hline\hline
n^{d,1}_g&d=6&7&8&9\\\hline
11&259498&1365366811756288&-4736147618262226212416&2362150336628273365296836096\\
12&\gvHLc{3576}&-17274516630240&241184196439416280128&-238002496862926514779851040\\
13&\gvHLc{40}&97442213760&-9638196546841305600&20760200170428536398007680\\
14&\gvHLc{0}&-33988608&289481641285797280&-1555728615943980103040320\\
15&\gvHLc{0}&\gvHLc{-324544}&-6124194557064832&98700288517404683773632\\
16&\gvHLc{0}&\gvHLc{0}&81788344992768&-5202996081724936520576\\
17&\gvHLc{0}&\gvHLc{0}&-550195366080&222451947587448686272\\
18&\gvHLc{0}&\gvHLc{0}&846863328&-7465750121080525472\\
19&\gvHLc{0}&\gvHLc{0}&376064&187682431656202624\\
20&\gvHLc{0}&\gvHLc{0}&\gvHLc{5760}&-3286547431924608\\
21&\gvHLc{0}&\gvHLc{0}&\gvHLi{0}&35286952874368\\
22&\gvHLc{0}&\gvHLc{0}&\gvHLc{0}&-174139266432\\
23&\gvHLc{0}&\gvHLc{0}&\gvHLc{0}&48799616\\
24&\gvHLc{0}&\gvHLc{0}&\gvHLc{0}&\gvHLc{531072}\\
25&\gvHLc{0}&\gvHLc{0}&\gvHLc{0}&\gvHLc{0}\\\hline
\end{array}
\end{align*}
\caption{Torsion refined Gopakumar-Vafa invariants for $X_{(5,1^3)}$ -- continuation of Table~\ref{tab:gv5111tA}.
Invariants that are checked by the enumerative calculations from~\cite{Katz:2022lyl} and Section~\ref{sec:enumerativeGeometry} are highlighted in blue.
Predicted invariants that have been used to fix the holomorphic ambiguity are highlighted in red.
}
\label{tab:gv5111tB}
\end{table}

\subsection{$\vec{d}=(2^4),\,n_s=80$}\label{sec:dint2222}
From the associated GLSM in~\ref{sec:glsmModels} one observes that for $\vec{d}=(2^4)$ the moduli space is $\mathbb{Z}_2$-symmetric.
More precisely, we find that changing to the coordinate
\begin{align}
    w=\frac{1}{2^{16}z}\,,
\end{align}
while at the same time performing a K\"ahler transformation such that sections of the K\"ahler line bundle $\mathcal{L}$ are rescaled by a factor $f(w)=1/(2^4\sqrt{w})$.

Using the general transformation behaviour of the propagators, that has been explicitly worked out in~\cite[Appendix B]{Katz:2022lyl}, there is a unique choice for the propagator ambiguities in~\eqref{eqn:bcov},\eqref{eqn:christoffel} such that $s^w_{ww}=s^z_{zz}\big\vert_{z\rightarrow w}$ and analogous relations hold for  $h^{zz}_z,h^z_z,h_z$,
\begin{align}
    s^z_{zz}=-\frac32\frac{1}{z}\,,\quad h^{zz}_z=0\,,\quad h_{zz}=\frac{1}{8}\frac{1}{z^2}\,,\quad h^z_z=0\,,\quad h_z=\frac{1}{256z}-1\,.
\end{align}
The symmetry in the moduli space then implies that the holomorphic ambiguities $f^{(g)}(z)$ have to satisfy
\begin{align}
    f(z)^{2-2g}f^{(g)}\left(\frac{1}{2^{16}z}\right)=f^{(g)}(z)\,.
    \label{eqn:ambiguityInvolution}
\end{align}
We find the regulator $\rho_{\vec{d}}$ in~\eqref{eqn:holomorphicAmbiguity} to be trivial, i.e. $\rho_{(2^4)}=2$.

The prediction for the torsion refined Gopakumar-Vafa invariants from~\cite{Katz:2022lyl} and the discussion in Section~\ref{sec:enumerativeGeometry} are
\begin{align}
    \begin{split}
    n_{g_{\text{max.}}(2b)}^{2b,p}=&\left\{\begin{array}{cl}
        6&b=1\,,\quad p\equiv 0\\[.1cm]
        4(b^2-b+4)&b\ge 2\,,\quad p\equiv 0\\[.1cm]
        0&\text{ else}
    \end{array}\right.\,,\\[.3cm]
    n_{g_{\text{max.}}(2b)-1}^{2b,p}=&\left\{\begin{array}{cl}
        384&b=1\,,\quad p\equiv 0\\[.1cm]
        480&b=1\,,\quad p\equiv 1\\[.1cm]
        -8 (b^4 - 48 b^2 + 55 b-153)&b\ge 2\,,\quad p\equiv 0\\[.1cm]
        480(b^2-b+3)&b\ge 2\,,\quad p\equiv 1
    \end{array}\right.\,.
    \end{split}
\end{align}

Combining the generic formulas for the constant map contributions, the gap condition, the Castelnuovo vanishing for $X_{(8)}$, the behaviour of the ambiguities under the involution~\eqref{eqn:ambiguityInvolution} and the prediction $n^{6,1}_{13}=0$ we can fix the holomorphic ambiguities up to genus $g=14$.
Some of the resulting invariants are listed in Tables~\ref{tab:gv2222tA} and~\ref{tab:gv2222B}.

\begin{table}[t!]
\begin{align*}
\tiny
\begin{array}{|c|ccccc|}\hline
n^{d,0}_g&d=1&2&3&4&5\\\hline
g=0&14784&64416224&711860299456&11596528005950816&233938237312747032384\\
1&\gvHLc{0}&20832&10732236416&902646051617216&50712027457946480768\\
2&\gvHLc{0}&\gvHLc{384}&-8273344&6249833642992&2700746768860537408\\
3&\gvHLc{0}&\gvHLc{6}&-88704&-87429700304&10292236854183040\\
4&\gvHLc{0}&\gvHLc{0}&\gvHLc{0}&198090456&-337281113930752\\
5&\gvHLc{0}&\gvHLc{0}&\gvHLc{0}&153930&6031964226304\\
6&\gvHLc{0}&\gvHLc{0}&\gvHLc{0}&\gvHLc{1752}&-43154064768\\
7&\gvHLc{0}&\gvHLc{0}&\gvHLc{0}&\gvHLc{24}&18754048\\
8&\gvHLc{0}&\gvHLc{0}&\gvHLc{0}&\gvHLc{0}&177408\\
9&\gvHLc{0}&\gvHLc{0}&\gvHLc{0}&\gvHLc{0}&\gvHLc{0}\\\hline\hline
n^{d,1}_g&d=1&2&3&4&5\\\hline
g=0&14720&64418688&711860247424&11596528018842496&233938237312502284416\\
1&\gvHLc{0}&20480&10732114176&902646041088640&50712027456069874944\\
2&\gvHLc{0}&\gvHLc{480}&-8278400&6249833634752&2700746768384335488\\
3&\gvHLc{0}&\gvHLc{0}&-88320&-87429803520&10292236845747456\\
4&\gvHLc{0}&\gvHLc{0}&\gvHLc{0}&198125344&-337281110788096\\
5&\gvHLc{0}&\gvHLc{0}&\gvHLc{0}&147520&6031964042752\\
6&\gvHLc{0}&\gvHLc{0}&\gvHLc{0}&\gvHLc{2400}&-43153745664\\
7&\gvHLc{0}&\gvHLc{0}&\gvHLc{0}&\gvHLc{0}&18775040\\
8&\gvHLc{0}&\gvHLc{0}&\gvHLc{0}&\gvHLc{0}&176640\\
9&\gvHLc{0}&\gvHLc{0}&\gvHLc{0}&\gvHLc{0}&\gvHLc{0}\\\hline
\end{array}
\end{align*}
\caption{Torsion refined Gopakumar-Vafa invariants for $X_{(2^4)}$ -- continued in Table~\ref{tab:gv2222B}.
Invariants that are checked by the enumerative calculations from~\cite{Katz:2022lyl} and Section~\ref{sec:enumerativeGeometry} are highlighted in blue.
}
\label{tab:gv2222tA}
\end{table}

\begin{table}[t!]
\begin{align*}
\tiny
\begin{array}{|c|ccc|}\hline
n^{d,0}_g&d=6&7&8\\\hline
5&-1760771999337079644&5094383347896744172553856&9657408241193249030377056580064\\
6&72538234044772552&-3538613097711461376704&302811038246656266173049956344\\
7&-2014447545740602&840840604025737981696&2730973526222775373837007680\\
8&33618973736984&-47725225126213705792&13520626841826392946717024\\
9&-268866994212&2050806873394447488&-995729033386724002067974\\
10&459073568&-64291180272034752&75267090274862000248592\\
11&275434&1365366850034176&-4736147442485274306636\\
12&\gvHLc{2712}&-17274520809792&241184169648999379088\\
13&\gvHLc{40}&97442319616&-9638193257675660990\\
14&\gvHLc{0}&-33959936&289481324372907352\\\hline\hline
n^{d,1}_g&d=6&7&8\\\hline
5&-1760771999426249856&5094383347897028676727040&9657408241193249610511196806208\\
6&72538234124970048&-3538613097704698901376&302811038246656250507734160480\\
7&-2014447581774656&840840604020596741632&2730973526222761259836700096\\
8&33618983223520&-47725225123135040128&13520626841826706308234240\\
9&-268868895680&2050806872095534336&-995729033387090648439296\\
10&459350816&-64291179887975808&75267090275152642943104\\
11&248000&1365366773478400&-4736147442626206624768\\
12&\gvHLc{4320}&-17274512450688&241184169703534917376\\
13&\gvHLi{0}&97442107904&-9638193274571362880\\
14&\gvHLc{0}&-34017280&289481328580000096\\\hline
\end{array}
\end{align*}
\caption{Torsion refined Gopakumar-Vafa invariants for $X_{(2^4)}$ -- continuation of Table~\ref{tab:gv2222tA}.
Invariants that are checked by the enumerative calculations from~\cite{Katz:2022lyl} and Section~\ref{sec:enumerativeGeometry} are highlighted in blue.
Predicted invariants that have been used to fix the holomorphic ambiguity are highlighted in red.
}
\label{tab:gv2222B}
\end{table}

\subsection{$\vec{d}=(4,2^2),\,n_s=72$}
For $\vec{d}=(4,2^2)$ we can choose the propagator ambiguities
\begin{align}
    s^z_{zz}=-\frac74\frac{1}{z}\,,\quad h^{zz}_z= \frac{5}{32}z\,,\quad h_{zz}=\frac18 \frac{1}{z^2}\,,\quad h^z_z=0\,,\quad h_z=\frac{1}{256}\frac{1}{z}\,.
\end{align}
The regulator is trivial, i.e. $\rho_{(4,2^2)}=2$.

The prediction for the torsion refined Gopakumar-Vafa invariants from~\cite{Katz:2022lyl} and the discussion in Section~\ref{sec:enumerativeGeometry} are
\begin{align}
    \begin{split}
    n_{g_{\text{max.}}(2b)}^{2b,p}=&\left\{\begin{array}{cl}
        6&b=1\,,\quad p\equiv 0\\[.1cm]
        4(b^2-b+4)&b\ge 2\,,\quad p\equiv 0\\[.1cm]
        0&\text{ else}
    \end{array}\right.\,,\\[.3cm]
    n_{g_{\text{max.}}(2b)-1}^{2b,p}=&\left\{\begin{array}{cl}
        432&b=1\,,\\[.1cm]
        -8 (b^4 - 54 b^2 + 61 b-171)&b\ge 2\,,\quad p\equiv 0\\[.1cm]
        432(b^2-b+3)&b\ge 2\,,\quad p\equiv 1
    \end{array}\right.\,.
    \end{split}
\end{align}

Using the boundary behaviour together with the Castelnuovo vanishing of $X_{(8)}$ and the predictions 
\begin{align}
    n^{2,0}_2=432\,,\quad n^{2,0}_3=6\,,\quad n^{3,0}_4=0\,,
\end{align}
we can only fix the holomorphic ambiguity up to genus $g=4$.
However, we observe that choosing a suitably normalized flat coordinate $t_{\infty}$ around $z=\infty$ the leading behaviour of the free energies takes the form
\begin{align}
\begin{split}
    \mathcal{F}_g(t_{\infty})=&(-1)^g(1-2^{1+2g})\frac{B_{2g}}{2g(2g-2)}t_{\infty}^{2-2g}\\
    &+(-1)^g\frac{2-3(g-3)}{2^{25}} t_{\infty}^{10-2g}+\mathcal{O}\left(t_{\infty}^{\text{min}(14-2g,0)}\right)\,.
    \end{split}
\end{align}
The first term is analogous to the behaviour around the point at infinity in the moduli space of the complete intersection $X_{3,2,2}\subset\mathbb{P}^7$ that was observed in~\cite{Huang:2006hq}.

Assuming that this behaviour is valid for all genera, we are able to fix the holomorphic ambiguity up to genus $g=9$.
Some of the resulting torsion refined Gopakumar-Vafa invariants are listed in Table~\ref{tab:gv422t}.

\begin{table}[t!]
\begin{align*}
\begin{array}{|c|ccccc|}\hline
n^{d,0}_g&d=1&2&3&4&5\\\hline
g=0&14912&64427168&711862230336&11596528489833504&233938237461490068672\\
1&\gvHLc{0}&20576&10732032384&902645974034944&50712027431616415616\\
2&\gvHLc{0}&\gvHLc{432}&-8273472&6249835104336&2700746767335000512\\
3&\gvHLc{0}&\gvHLc{6}&-89472&-87430081488&10292236801862528\\
4&\gvHLc{0}&\gvHLc{0}&\gvHLc{0}&198140200&-337281083689984\\
5&\gvHLc{0}&\gvHLc{0}&\gvHLc{0}&150250&6031959118080\\
6&\gvHLc{0}&\gvHLc{0}&\gvHLc{0}&\gvHLc{1992}&-43153195648\\
7&\gvHLc{0}&\gvHLc{0}&\gvHLc{0}&\gvHLc{24}&18737664\\
8&\gvHLc{0}&\gvHLc{0}&\gvHLc{0}&\gvHLc{0}&178944\\
9&\gvHLc{0}&\gvHLc{0}&\gvHLc{0}&\gvHLc{0}&\gvHLc{0}\\\hline\hline
n^{d,1}_g&d=1&2&3&4&5\\\hline
g=0&14592&64407744&711858316544&11596527534959808&233938237163759248128\\
1&\gvHLc{0}&20736&10732318208&902646118670912&50712027482399940096\\
2&\gvHLc{0}&\gvHLi{432}&-8278272&6249832173408&2700746769909872384\\
3&\gvHLc{0}&\gvHLi{0}&-87552&-87429422336&10292236898067968\\
4&\gvHLc{0}&\gvHLc{0}&\gvHLi{0}&198075600&-337281141028864\\
5&\gvHLc{0}&\gvHLc{0}&\gvHLc{0}&151200&6031969150976\\
6&\gvHLc{0}&\gvHLc{0}&\gvHLc{0}&\gvHLc{2160}&-43154614784\\
7&\gvHLc{0}&\gvHLc{0}&\gvHLc{0}&\gvHLi{0}&18791424\\
8&\gvHLc{0}&\gvHLc{0}&\gvHLc{0}&\gvHLc{0}&175104\\
9&\gvHLc{0}&\gvHLc{0}&\gvHLc{0}&\gvHLc{0}&\gvHLc{0}\\\hline
\end{array}
\end{align*}
\caption{Torsion refined Gopakumar-Vafa invariants for $X_{(4,2^2)}$.
Invariants that are checked by the enumerative calculations from~\cite{Katz:2022lyl} and Section~\ref{sec:enumerativeGeometry} are highlighted in blue.
Predicted invariants that have been used to fix the holomorphic ambiguity are highlighted in red.
}
\label{tab:gv422t}
\end{table}

\section{Enumerative Geometry}
\label{sec:enumerativeGeometry}

Recall from Proposition~\ref{prop:11classes} that $\Xd$ is smooth for decompositions of length $l=1$ and for $l=2$ the small resolution $\widehat{\Xd}$ is K\"ahler, so that standard methods of algebraic geometry can be applied to understand the enumerative invariants.  For that reason, we restrict to the case $l\ge3$.

We have $H_2(\widehat{\Xd},\mathbb{Z})=\mathbb{Z}\oplus\mathbb{Z}_2$, which we describe in terms of the conifold transition from   $\widehat{\Xdr}$ to $\widehat{\Xd}$ using the exact sequence (\ref{eqn:h2fromtransition}).  We also note that the exceptional curves $C^{(1)}$ of $\widehat{\Xdr}$ defined in Lemma~\ref{lem:resolution} have class $C_F$ in terms of the notation introduced at the beginning of Section~\ref{sec:smallrescompint}.  The exceptional curves $C^{(2)}$ of $\widehat{\Xdr}$ have class $2C_F$.

The exact sequence (\ref{eqn:h2fromtransition}) says that the classes represented by the curves $C_F$ and $C_B$ survive the conifold transition, while $2C_F$ vanishes after the transition.  With slight abuse of notation, we can and will represent classes in $H_2(\widehat{\Xd},\mathbb{Z})$ as $dC_B+pC_F$ with $d\in\mathbb{Z}_{\ge0}$ and $p\in\mathbb{Z}_2$.  We will describe the GV invariants for curves of class $dC_B+pC_F$ as $n_g^{d,p}$.

We will now discuss the predictions for some of the torsion refined Gopakumar-Vafa invariants, based on the proposed mathematical definition, generalizing earlier results for $\vec{d}=(1^8)$ from~\cite{Katz:2022lyl}.

First let us note that the discussion in~\cite{Katz:2022lyl} was based on the assumption that the preimage of a line $\ell\subset\mathbb{P}^3$ under $\pi:\widehat{\Xd}\to\mathbb{P}^3$ has class $(2,1)$ in $H_2(\widehat{\Xd},\mathbb{Z})\simeq \mathbb{Z}\oplus\mathbb{Z}_2$.
Using the conifold transitions from $\widehat{\Xdr}$ to $\widehat{\Xd}$ we can now derive and generalize this relation.

To this end we consider the mapping $\pi_r:\widehat{\Xdr}\to\mathbb{P}^3$ and using a standard intersection calculation on $\widehat{\Xdr}\simeq \Xab$ one finds that
\begin{align}
\pi_r^{-1}(\ell)=m C_F + 2 C_B\,,
\label{eqn:parity}
\end{align}
where $m$ is even (odd) if the decomposition $\vec{d}$ is even (odd).

After the conifold transition to $\widehat{\Xd}$, we then have $H_2(\widehat{\Xd},\mathbb{Z})\simeq\mathbb{Z}\oplus\mathbb{Z}_2$, with $\mathbb{Z}$ generated by $C_B$ and $\mathbb{Z}_2$ generated by $C_F$.
Thus $\pi^{-1}(\ell)$ has class $(2,0)\in \mathbb{Z}\oplus\mathbb{Z}_2$ for $\vec{d}$ even and class $(2,1)$ for $\vec{d}$ odd.

Letting $n_g^d$ denote the GV invariants of the smooth octic double solid, we also recall that
\begin{align}
    n^d_g=n^{d,0}_g+n^{d,1}_g\,,
    \label{eqn:gvconisum}
\end{align}
from the general results in~\cite{Li:1998hba}.

\paragraph{Odd decompositions}
If the entries of the decomposition $\vec{d}$ are odd, the discussion from~\cite{Katz:2022lyl} directly generalizes and one obtains the results
\begin{align}\small
    \begin{split}
    n_{g_{\text{max.}}(2b+1)}^{2b+1,0}=&n_{g_{\text{max.}}(2b+1)}^{2b+1,1}=\left\{\begin{array}{cl}
    14752&b=0\\
    (-1)^b14752(2b^2+4)&b\ge 1
    \end{array}\right.\,,\\[.3cm]
    n_{g_{\text{max.}}(2b)}^{2b,p}=&\left\{\begin{array}{cl}
        6&b=1\,,\quad p\equiv 1\\[.1cm]
        4(b^2-b+4)&b\ge 2\,,\quad p\equiv b\\[.1cm]
        0&\text{ else}
    \end{array}\right.\,,\\[.3cm]
    n_{g_{\text{max.}}(2b)-1}^{2b,p}=&\left\{\begin{array}{cl}
       6\cdot n_{\vec{d}}&b=1\,,\quad p\equiv 1\\[.1cm]
       6(144-n_{\vec{d}})&b=1\,,\quad p\equiv 0\\[.1cm]
       \begin{array}{c}-8 (b^4 - 108 b^2 + 115 b - 333)\\-6\cdot n_{\vec{d}} (b^2 - b + 3)\end{array}&b\ge 2\,,\quad p\equiv b\\[.3cm]
       6\cdot n_{\vec{d}}(b^2-b+3)&b\ge 2\,,\quad p\equiv b+1
    \end{array}\right.\,.
    \end{split}
\end{align}
with $g_{\text{max.}}(d)$ as defined in~\eqref{eqn:castelnuovoBound}.

Let us briefly recall the enumerative geometry behind some of these numbers, while referring to~\cite{Katz:2022lyl} for the rest.

A generic line $\ell$ in $\mathbb{P}^3$ intersects $\Sd$ in eight points and therefore the double cover $\pi^{-1}(\ell)\subset\widehat{\Xd}$ is a genus 3 curve of degree 2 in $\widehat{\Xd}$.
As discussed above, the preimage $\pi^{-1}(\ell)$ has odd $\mathbb{Z}_2$ parity.
The moduli space of rational curves in $\mathbb{P}^3$ is $\mathcal{M}=G(2,4)$ and one finds $n_3^{2,1}=(-1)^{\text{dim}(\mathcal{M})}\chi(\mathcal{M})=6$.

In order to get curves of degree $1$ in $\widehat{\Xd}$ one has to consider special configurations where the line $\ell$ has four tangencies with $\Sd$ and the double cover $\pi^{-1}(\ell)$ splits.
The number of such curves was calculated already in~\cite{schubert:1879} and is precisely $14752$.
Again using the fact that $\pi^{-1}(\ell)$ has odd $\mathbb{Z}_2$ parity one finds that the two sheets of the split double cover have to have opposite parity.
It follows that $n_0^{1,0}=n_0^{1,1}=14752$.

To obtain curves of degree $2$ and genus $g_{\text{max.}}(2)-1=2$ we can consider a line $\ell_p$ through any node $p$.
Away from the node, the line intersects $\Sd$ six times and one therefore obtains a genus two curve.
Removing the exceptional curve $C_p$, that resolves the node in $\widehat{\Xd}$, from $\pi^{-1}(\ell_p)$ gives an irreducible genus two curve with degree and parity $(d,p)=(2,0)$.
However, since we have to consider both small resolutions, we can also substract the flopped curve $-C_p$ and both $\ell_p+C_p$ and $\ell_p-C_p$ contribute to $n_2^{2,0}$.
The moduli space of lines through a given point in $\mathbb{P}^3$ is $\mathbb{P}^2$ and one finds $n_2^{2,0}=2\cdot \chi(\mathbb{P}^2)\cdot n_{\vec{d}}=6n_{\vec{d}}$.
The result for $n_2^{2,1}$ then follows from~\eqref{eqn:gvconisum} together with the fact that for the generic octic double solid one has $n_2^2=864$.

To obtain the GV invariants $n^{2b+1,0}_{g_{\text{max.}}(2b+1)}=n^{2b+1,1}_{g_{\text{max.}}(2b+1)}$ we proceed as follows.  The argument for $b=1$ is given in \cite{Katz:2022lyl}, so we assume $b\ge2$.  For each of the lines $\ell\subset\widehat{\Xd}$ of degree 1, we consider the moduli space of complete intersections of hypersurfaces of degrees $(1,b+1)$ which contain $\ell$.  A complete intersection of this type has degree $2(b+1)$ and is the union of $\ell$ and a curve $C$.  The curve $C$ has degree $2b+1$ in $\widehat{\Xd}$.  We compute the moduli space of such $C$ and the genus of these curves to determine the degree $2b+1$ GV invariants at the Castelnuovo bound.

First, the degree 1 hypersurface is necessarily $f^{-1}(H)$ for a hyperplane $H\subset \mathbb{P}^3$, recalling our notation $f:\widehat{\Xd}\rightarrow\mathbb{P}^3$.  Since $H$ necessarily contains the line $f(\ell)\subset\mathbb{P}^3$, the parameter space for $H$ is $\mathbb{P}^1$ and the moduli space fibers over $\mathbb{P}^1$.  We next describe the hypersurfaces $H'$ of degree $b+1$ in $f^{-1}(H)$ comprising the fiber by analyzing the space of global sections $s\in H^0(f^{-1}(H),f^*\mathcal{O}_H(b+1))$. This is a vector space of dimension $(b+1)^2-(b+1)+4=b^2+b+4$ \cite{Katz:2022lyl}.

We then require $s$ to vanish on $\ell$.  Since $s\vert_\ell$ is a global section of $\mathcal{O}_\ell(b+1)$ on $\ell\simeq\mathbb{P}^1$, its vanishing imposes $\dim H^0(\mathcal{O}_\ell(b+1))=b+2$ linear conditions on $H'$.  Thus the fiber is a projective space of dimension $(b^2+b+4)-(b+2)-1=b^2+1$.  We find the moduli space has dimension whose parity is the same as the parity of $b$, and has Euler characteristic $2(d^2+2)$.  We conclude that 
\begin{align}
    n^{2b+1,0}_{g_{\text{max.}}(2b+1)}=n^{2b+1,1}_{g_{\text{max.}}(2b+1)}=(-1)^b14752(2b^2+4)\,.
\end{align}

We next compute $g_{\text{max.}}(2b+1)$ from the geometry of these curves.  The curves $C$ and $\ell$ are said to be linked by the complete intersection, and the genus of $C$ can be computed directly from the genus of $\ell$ (0) by the general theory of linkage \cite{PS}.  However, it is easy to compute the genus by more elementary means by taking the degree $b+1$ hypersurface in a special form.

Let's choose $H'$ to be the union of a degree $b$ hypersurface $H''$ and the line $f(\ell)$.  $H''$ is just a curve of degree $b$ in the plane $H$.  The double cover $f^{-1}(H')\to H'$ is branched over $8b$ points.  From Riemann-Hurwitz, this enables us to compute the genus of $f^{-1}(H')$ as $b^2+b+1$.  Then $f^{-1}(H'\cup f(\ell))=f^{-1}(H')\cup\ell\cup \ell'$, where $f^{-1}(f(\ell))=\ell\cup\ell'$.  Thus our curve $C$ is $f^{-1}(H)\cup\ell'$.  Noting that $f^{-1}(H)$ and $\ell'$ meet at $b$ points (lying over the $b$ points $H\cap f(\ell)$) we conclude that $g_{\text{max.}}(2b+1)=b^2+2b$.

\paragraph{Even decompositions}
Let us now turn towards decompositions with entries that are even.

Recall that the preimage of a line in $\mathbb{P}^3$ then has class $(d,p)=(2,0)$ in $H_2(\widehat{\Xd},\mathbb{Z})\simeq \mathbb{Z}\oplus\mathbb{Z}_2$.
As a consequence, the invariants $n_{g_{\text{max.}}}^{2b+1,p}$ can depend on $p$ and it is not obvious how to calculate the parity of the individual sheets of a split double cover in $\widehat{\Xd}$ of a curve in $\mathbb{P}^3$.

On the other hand, the remaining predictions generalize straightforwardly and, taking into account the change of parity for $\pi^{-1}(\ell)$, one finds
\begin{align}
    \begin{split}
    n_{g_{\text{max.}}(2b)}^{2b,p}=&\left\{\begin{array}{cl}
        6&b=1\,,\quad p\equiv 0\\[.1cm]
        4(b^2-b+4)&b\ge 2\,,\quad p\equiv 0\\[.1cm]
        0&\text{ else}
    \end{array}\right.\,,\\[.3cm]
    n_{g_{\text{max.}}(2b)-1}^{2b,p}=&\left\{\begin{array}{cl}
       6\cdot n_{\vec{d}}&b=1\,,\quad p\equiv 0\\[.1cm]
       6(144-n_{\vec{d}})&b=1\,,\quad p\equiv 1\\[.1cm]
       \begin{array}{c}-8 (b^4 - 108 b^2 + 115 b - 333)\\-6\cdot n_{\vec{d}} (b^2 - b + 3)\end{array}&b\ge 2\,,\quad p\equiv 0\\[.3cm]
       6\cdot n_{\vec{d}}(b^2-b+3)&b\ge 2\,,\quad p\equiv 1
    \end{array}\right.\,.
    \end{split}
\end{align}

\section{Outlook}
\label{sec:outlook}

\begin{itemize}
\setlength\itemsep{.4em}
    \item We have given further evidence for the framework developed in \cite{Katz:2022lyl} for mathematically defining and computing torsion refined GV invariants.  To be made more rigorous beyond the cases which we compute, foundational results are still needed about the existence of appropriate perverse sheaves of vanishing cycles on the moduli spaces of analytic sheaves on the non-K\"ahler manifolds which we consider here.
    \item We have suggested the existence of a theory of non-commutative invariants analogous to the theory of Gopakumar-Vafa invariants.  Such a theory would be based on moduli spaces or stacks of $M_\beta$ of semistable $\mathcal{B}$-modules and a map from $M_\beta$ to the Chow variety $\mathrm{Chow}_1(X)$ of 1-dimensional cycles on $X$, analogous to (\ref{eqn:hilbertchow}).  This approach circumvents complications arising from the use of non-K\"ahler manifolds.  The resulting enumerative theory is expected to be closely related to the theory of torsion refined invariants, and possibly coincides with it.
    \item As we discussed in the introduction, physics suggests that a non-commutative resolution always exists for a projective Calabi-Yau threefold $X$ with isolated nodes that are resolved by torsional exceptional curves.
    In general, $X$ is not a double cover and the resolution will not be of Clifford type.
    
    An important goal will therefore be to find a general method to construct a corresponding sheaf $\widehat{\mathcal{B}}$ on $\widehat{X}$.
    In particular, we would like to find a procedure for choosing a representative of the Brauer class on $\widehat{X}$ such that the (derived) push-forward of $\widehat{\mathcal{B}}$ by $\rho:\widehat{X}\rightarrow X$ is $\mathcal{B}$.

    \item Topological strings on many examples for non-commutative resolutions that are not of Clifford type have been studied in~\cite{Schimannek:2021pau} and arise from elements of the Tate-Shafarevich group of genus one fibered Calabi-Yau threefolds that do not have a section.
    In all of those cases the analysis was possible due to the existence of a smooth dual.
    However, from the perspective of the GLSM it was found that the non-commutative resolutions are not associated to a phase limit but rather to a phase boundary where the GLSM becomes singular.
    
    This does not imply that a GLSM corresponding to the non-commutative resolutions does not exist but only that it is different from the one associated to the smooth dual.
    Constructing GLSM for non-commutative resolutions that are not of Clifford type is therefore an important open problem.
    A solution would open up a large class of new GLSM associated to examples of general type without a smooth dual.

    A closely related open problem is to find a combinatorial construction of the mirrors of non-commutative resolutions that are not of Clifford type, and thus generalizing the results from~\cite{borisov2016clifford}.
\end{itemize}

\section{Acknowledgements}
We would like to thank L.~Borisov, D.~Jaramillo~Duque, A.-K.~Kashani-Poor, A.~Klemm, J.~Knapp, I.~Melnikov, B.~Pioline, R.~Plesser, E.~Scheidegger and E.~Sharpe for useful conversations.
In particular we want to thank A.~Klemm and E.~Sharpe for collaboration on previous related work and L.~Borisov for providing us with the Gorenstein cones discussed in Section~\ref{sec:gorensteincones}.
The research of S.K. was supported by NSF grants DMS-1802242 and DMS-2201203.  The research of T.S.~is supported by the Agence Nationale de la Recherche (ANR) under contract number ANR-21-CE31-0021.
The authors also thank the Simons Center for Geometry and Physics for hospitality during part of the completion of this work.

\addcontentsline{toc}{section}{References}
\bibliography{names}
\end{document}